\documentclass[a4paper,11pt]{article}
\usepackage{jheppub}

\usepackage[most]{tcolorbox}
\usepackage{tikz}
\usepackage{amsmath}
\usepackage{amsmath, amsthm, amssymb, hyperref, color}
\usepackage[dvipsnames]{xcolor}

\usepackage[most]{tcolorbox}
\usepackage{amsthm}

\theoremstyle{remark}

\theoremstyle{plain}
\newtheorem{theorem}{Theorem}

\newtheorem{lemma}[theorem]{Lemma}
\newtheorem{corollary}[theorem]{Corollary}

\usepackage{esint}
\usepackage{babel}
\usepackage{color}
\usepackage{enumitem}

\newcommand{\la}{\langle}
\newcommand{\ra}{\rangle}
\DeclareMathOperator{\res}{res}

\usepackage{caption}
\usepackage{subfigure}

\usepackage{float}
\usepackage[dvipsnames]{xcolor}

\usepackage{amsfonts}
\usepackage{amssymb}
\usepackage{amsthm}

\usepackage{hyperref}
\allowdisplaybreaks
\hypersetup{
colorlinks=true,
linkcolor=blue,
linktocpage=true,
citecolor=violet
}

\usepackage{xargs}
\usepackage[colorinlistoftodos,prependcaption,textsize=tiny]{todonotes}

\newcommand{\qmap}{Center for Quantum Mathematics and Physics (QMAP), 
University of California, Davis, USA}
\newcommand{\mathdp}{Department of Mathematics, University of California, Davis, USA}

\usepackage{braket,physics,mathtools,enumitem}
\usepackage{tikz}
\usetikzlibrary{shapes}

\newcommand{\getfirst}[1]{%
  \expandafter\firstitem#1,\relax
}
\def\firstitem#1,#2\relax{#1}
\makeatletter
\newcommand{\getlast}[1]{%
  \expandafter\lastitem#1,\relax,\@nil
}
\def\lastitem#1,#2,\@nil{%
  \ifx\relax#2
    #1
  \else
    \expandafter\lastitem#2,\@nil
  \fi
}
\makeatother

\newcommand{\onshellgraph}[5]{
    \begin{tikzpicture}[scale=#5]
        \foreach \pos\index in #1{
            \node(\index) at \pos {\index};
        }
    
        \foreach \pos\index in #2{
            \node(\index)[draw,circle] at \pos {};
        }

        \foreach \pos\index in #3{
            \node(\index)[draw,fill,circle,text=white] at \pos {};
        }

        \foreach \doublet in #4{
            \draw[white,line width=3pt] (\getfirst\doublet) -- (\getlast\doublet);    
            \draw[line width=1pt] (\getfirst\doublet) -- (\getlast\doublet);
        }
    \end{tikzpicture}
}

\newcommand{\onshellgraphnolabel}[5]{
    \begin{tikzpicture}[scale=#5]
        \foreach \pos\index in #1{
            \node(\index) at \pos {};
        }
    
        \foreach \pos\index in #2{
            \node(\index)[draw,circle] at \pos {};
        }

        \foreach \pos\index in #3{
            \node(\index)[draw,fill,circle,text=white] at \pos {};
        }

        \foreach \doublet in #4{
            \draw[white,line width=3pt] (\getfirst\doublet) -- (\getlast\doublet);    
            \draw[line width=1pt] (\getfirst\doublet) -- (\getlast\doublet);
        }
    \end{tikzpicture}
}

\usetikzlibrary{decorations.markings}
\usetikzlibrary{arrows.meta}

\newcommand{\polygongraph}[3]{
    \begin{tikzpicture}[scale=#3]
        \foreach \pos\index in #1{
            \node(W\index) at \pos {};
        }
        
        \foreach \polygon in #2{
            \draw[fill=black!60!white,line width=1pt] (W\getfirst\polygon.center) \foreach \i in \polygon{-- (W\i.center)} -- cycle ;
        }

        \foreach \pos\index in #1{
            \node[draw,circle,fill=white] at \pos {};
            \node at \pos {\index};
        }
    \end{tikzpicture}
}

\newcommand{\decoratedpolygongraph}[5]{
    \begin{tikzpicture}[scale=#5]
        \foreach \pos\index in #1{
            \node(W\index) at \pos {};
        }
        
        \foreach \polygon in #2{
            \draw[fill=black!60!white,draw=none] (W\getfirst\polygon.center) \foreach \i in \polygon{-- (W\i.center)} -- cycle ;
        }

        \begin{scope}[decoration={
            markings,
            mark=at position 0.5 with {\arrow{>}}}
            ] 
            \foreach \doublet in #3{ 
            \draw[line width=1pt,postaction={decorate}] (\getfirst\doublet.center) -- (\getlast\doublet.center);
            }
            
            \foreach \doublet in #4{ 
                \draw[line width=1pt,red,postaction={decorate}] (\getfirst\doublet.center) -- (\getlast\doublet.center);
            }
        \end{scope}
        
        \foreach \pos\index in #1{
            \node[draw,circle,fill=white] at \pos {};
            \node at \pos {\index};
        }
    \end{tikzpicture}
}

\newcommand{\dualdecoratedpolygongraph}[6]{
    \begin{tikzpicture}[scale=#5]
        \foreach \pos\index in #1{
            \node(W\index) at \pos {};
        }
        
        \foreach \polygon in #2{
            \draw[fill=black!60!white,draw=none,fill opacity=0.3] (W\getfirst\polygon.center) \foreach \i in \polygon{-- (W\i.center)} -- cycle ;
        }

        \begin{scope}[decoration={
            markings,
            mark=at position 0.5 with {\arrow{>}}}
            ] 
            \foreach \doublet in #3{ 
            \draw[line width=1pt,postaction={decorate},draw opacity=0.3] (\getfirst\doublet.center) -- (\getlast\doublet.center);
            }
            
            \foreach \doublet in #4{ 
                \draw[line width=1pt,red,postaction={decorate},draw opacity=0.3] (\getfirst\doublet.center) -- (\getlast\doublet.center);
            }
        \end{scope}
        
        \foreach \pos\index in #1{
            \node[draw,circle,fill=white,draw opacity=0.3] at \pos {};
            \node at \pos {\index};
        }

        #6
    \end{tikzpicture}
}

\usepackage{soul} 

\begin{document}

\title{Grassmannian Geometries for Non-Planar On-Shell Diagrams}

\author[a]{Artyom Lisitsyn,}
\author[a]{Umut Oktem,}
\author[b]{Melissa Sherman-Bennett,}
\author[a]{Jaroslav Trnka}
\affiliation[a]{\qmap}
\affiliation[b]{\mathdp}
\emailAdd{alisitsyn@ucdavis.edu}
\emailAdd{ucoktem@ucdavis.edu}
\emailAdd{mshermanbennett@ucdavis.edu }
\emailAdd{trnka@ucdavis.edu}


\abstract{On-shell diagrams are gauge invariant quantities which play an important role in the description of scattering amplitudes. Based on the principles of generalized unitarity, they are given by products of elementary three-point amplitudes where the kinematics of internal on-shell legs are determined by cut conditions. In the ${\cal N}=4$ Super Yang-Mills (SYM) theory, the dual formulation for on-shell diagrams produces the same quantities as canonical forms on the Grassmannian $G(k,n)$. Most of the work in this direction has been devoted to the planar diagrams, which dominate in the large $N$ limit of gauge theories. On the mathematical side, planar on-shell diagrams correspond to cells of the positive Grassmannian $G_+(k,n)$ which have been very extensively studied in the literature in the past 20 years. In this paper, we focus on the non-planar on-shell diagrams which are relevant at finite $N$. In particular, we use the triplet formulation of Maximal-Helicity-Violating (MHV) on-shell diagrams to obtain certain regions in the Grassmannian $G(2,n)$. These regions are unions of positive Grassmannians with different orderings (referred to as oriented regions). We explore the features of these unions, and show that they are pseudo-positive geometries, in contrast to positive geometry of a single oriented region. For all non-planar diagrams which are \emph{internally planar} there always exists a strongly connected geometry, and for those that are \emph{irreducible}, there exists a geometry with no spurious facets. We also prove that the already known identity moves, square and sphere moves, form the complete set of identity moves for all MHV on-shell diagrams.}
\setcounter{tocdepth}{2}

\maketitle
\newpage

\section{Introduction}

Recent years have seen tremendous progress in our understanding of scattering amplitudes in the planar ${\cal N}{=}4$ Super Yang-Mills (SYM) theory (see \cite{Arkani-Hamed:2022rwr} for a recent review). This includes several powerful techniques: unitarity methods \cite{Bern:1994zx,Bern:1994cg,Bourjaily:2017wjl}, Wilson loop duality and dual conformal symmetry \cite{Alday:2007hr,Drummond:2008vq,Drummond:2009fd,Caron-Huot:2010ryg}, momentum twistor integrands \cite{Arkani-Hamed:2010zjl,Bourjaily:2013mma,Bourjaily:2015jna}, BCFW recursion relations \cite{Britto:2004ap,Britto:2005fq,Arkani-Hamed:2010zjl}, symbol alphabet \cite{Goncharov:2010jf,Golden:2013xva,Prlina:2017tvx,He:2020vob,Mago:2020kmp,He:2024fij,He:2025tyv},
higher-loop bootstrap \cite{Dixon:2011nj,Caron-Huot:2011dec,Dixon:2015iva,Caron-Huot:2016owq,Caron-Huot:2019vjl,Basso:2024hlx}, cluster adjancency \cite{Drummond:2018dfd,Drummond:2018caf}, antipodal duality \cite{Dixon:2021tdw,Dixon:2022xqh,Dixon:2023kop,Dixon:2025zwj}, and the connection to the positive Grassmannian \cite{Arkani-Hamed:2009ljj,Arkani-Hamed:2009nll,Mason:2009qx,Arkani-Hamed:2009kmp,Arkani-Hamed:2009pfk,Arkani-Hamed:2012zlh} and the amplituhedron \cite{Arkani-Hamed:2013jha,Arkani-Hamed:2013kca,Arkani-Hamed:2017vfh,Damgaard:2019ztj,Ferro:2022abq,He:2023rou,Arkani-Hamed:2023epq}. This has completely changed our view on perturbative scattering amplitudes, uncovered hidden symmetries and mathematical structures, and also massively improved our computational abilities. The bottom line in most of the developments is to start with the mathematical properties and symmetries as the primary defining principles and derive the physical principles as inevitable consequences.

The amplituhedron picture is the culmination of these efforts. The amplituhedron is a geometric object associated to tree-level amplitudes and all-loop planar integrands, and actual expressions for the amplitudes are recovered from the canonical differential form. Different triangulations of the same geometry provide different expansions of the amplitude such as those coming from Feynman diagrams, recursion relations, or new expansions with no known physical interpretation. The amplituhedron has been extensively studied from various perspectives, which provided many insights how the geometry encodes physical principles and symmetries \cite{Franco:2014csa,Arkani-Hamed:2014dca,Ferro:2015grk,Dennen:2016mdk,Ferro:2016zmx,Arkani-Hamed:2018rsk,Herrmann:2020qlt,Dian:2022tpf}. While the amplituhedron provides an all-loop order definition of the problem, finding the explicit expressions remains a big challenge. The BCFW recursion relations were conjectured to provide a particular triangulation of the geometry, which was recently proven to be the case \cite{Even-Zohar:2021sec,Even-Zohar:2023del,Even-Zohar:2024nvw,Galashin:2024ttp}. In an alternative strategy, the canonical form for the amplituhedron can be expanded in terms of `negative geometries' which form a new class of interesting objects \cite{Arkani-Hamed:2021iya,Brown:2023mqi,Henn:2023pkc,Chicherin:2024hes,Glew:2024zoh,Lagares:2024epo,Brown:2025plq,Chicherin:2025cua}. The negative geometries reproduce an IR finite cousin of the amplitude -- a Wilson loop with a Lagrangian insertion \cite{Engelund:2011fg,Alday:2013ip}. 

The amplituhedron is also interesting from a purely mathematical perspective as a generalization of the positive Grassmannian with many interesting combinatorial properties and connections to cluster algebras and other mathematical structures, see \cite{Karp:2016uax,Karp:2017ouj,Galashin:2018fri,Lukowski:2020dpn,Parisi:2021oql,Parisi:2024psm,Lam:2024gyg,Akhmedova:2023wcf,Even-Zohar:2025ydi,Ranestad:2024svp,Mazzucchelli:2025kzy} for recent progress on various aspects of the story. 

One very important question is how these mathematical structures extend beyond the planar limit. What geometry can we associate to amplitudes in the non-planar theory? We have some indications from physics \cite{Bern:2015ple,Bourjaily:2018omh,Bourjaily:2019gqu,Bourjaily:2019iqr} that non-planar amplitudes do enjoy very special analytic properties, though this is not obvious from any known physical principles. The best starting point for non-planar explorations is the study of \emph{on-shell diagrams} and associated on-shell forms. These are on-shell gauge-invariant objects which represent cuts of loop integrands, and are well-defined both in planar and non-planar sectors of the theory, and can be calculated as products of three-point amplitudes. From a graphical perspective, on-shell diagrams are bi-colored graphs with all trivalent vertices. In mathematics, on-shell diagrams in the planar sector are known as \emph{plabic graphs} and they are associated with cells in the positive Grassmannian; see \cite{Postnikov:2006kva,LusTPFlag,RietschCell,Wil-enumeration,GKL} for a more detailed mathematical discussion. The connection between on-shell diagrams and the Grassmannian $G(k,n)$ allowed for the dual formulation and the evaluation of the canonical dlog form using parameters of the cell in the Grassmannian. This connection holds beyond the planar limit and it is a cornerstone of our exploration in this paper. This was noted in the original paper on scattering amplitudes and the positive Grassmannian \cite{Arkani-Hamed:2012zlh}, but in the context of non-planar diagrams was further studied in \cite{Arkani-Hamed:2014bca,Franco:2015rma,Bourjaily:2016mnp,Cachazo:2019,Paranjape:2022ymg} (see also \cite{Herrmann:2016qea,Heslop:2016plj,Armstrong:2020ljm,Paranjape:2023qsq, Bourjaily:2023ycy} for work on on-shell diagrams in Einstein gravity, though the forms there are non-logarithmic and depend explicitly in kinematics in addition to Grassmannian variables). 

In particular, in \cite{Arkani-Hamed:2014bca} it was shown that all MHV (maximal-helicity-violating) on-shell diagrams can be evaluated in a remarkably simple way. On one hand, the expression for each diagram can be expanded in a basis of Parke-Taylor factors, that are relevant for planar diagrams with different orderings. Alternatively, the same expression can be obtained using an intriguing determinant formula. In the Grassmannian language, the MHV case corresponds to $G(2,n)$ which exhibits major simplifications. Higher degree cases were studied in \cite{Bourjaily:2016mnp} and already for $G(3,6)$ there are qualitatively new features. 

In this paper, we focus on the following question: what geometries can be associated to MHV non-planar on-shell diagrams? That is, are there natural geometries whose canonical forms match those from non-planar on-shell diagrams? For planar diagrams, there is an elegant answer to this question: the positive Grassmannian $G_+(k,n)$. Each planar diagram is associated with a certain cell in $G_+(k,n)$, a region defined by setting some minors to zero and asking that others are nonnegative. If we calculate a canonical form of this region, we reproduce the expression for the diagram calculated from physics. 

In the non-planar sector, we do have an expression for each diagram, written in terms of minors of the (no longer positive part of the) Grassmannian $G(k,n)$. But we do not know a region whose canonical form is this expression. In this paper, we find such a region in the $G(2,n)$ case, which is relevant for $n$-pt MHV amplitudes. This is interesting from a mathematical perspective: we will associate geometries with non-planar graphs, and obtain certain generalizations of the positive Grassmannian. From a physical perspective, this is a major step in the exploration of non-planar amplitudes using new mathematical tools. A detailed understanding of underlying combinatorial structures should lead to insights on non-planar amplitudes, similar to those we obtained in planar amplitudes from the positive Grassmannian and the amplituhedron formulations. 

The paper is organized as follows: In Section \ref{sec:onshell}, we review the connection between planar on-shell diagrams and the cells of the positive Grassmannian. In Section \ref{sec:non-planar}, we discuss the generalization to the non-planar sector and review the triplet formula for the MHV on-shell diagrams. In Section \ref{sec:determinant-and-moves}, we leverage this formula to prove that the previously known sphere move introduced in \cite{Cachazo:2019} are the only identity move that does not change the canonical form. Section \ref{sec:regions-and-geometries} contains the main result of this paper, and we study pseudo-positive geometries associated with MHV on-shell diagrams. In particular, we show that for the class of internally planar diagrams, the geometry is always connected. We end with Conclusions and Outlook. Three appendices contain proofs of various statements.

\section{From Planar On-shell Diagrams to Positive Grassmannian}\label{sec:onshell}

On-shell diagrams are on-shell gauge invariant objects built from elementary building blocks which are tree-level amplitudes. For massless theories in four-dimensions, elementary three-point amplitudes are fixed by Poincar\'e symmetry up to an overall constant. The on-shell conditions and the momentum conservation for external momenta $p_1$, $p_2$, $p_3$,
\begin{equation}
p_1^2=p_2^2=p_3^2=0, \quad p_1+p_2+p_3=0
\end{equation}
have two solutions. These can be best described using the spinor-helicity variables $p^\mu = \sigma^\mu_{\alpha\dot{\alpha}}\lambda_\alpha \widetilde{\lambda}_{\dot{\alpha}}$ as collinearity conditions between spinors, $\lambda_1\sim\lambda_2\sim\lambda_3$ or $\widetilde{\lambda}_1\sim\widetilde{\lambda}_2\sim\widetilde{\lambda}_3$. For amplitudes with spin, these kinematical solutions are associated with two types of amplitudes which differ by helicities of external particles. In the context of ${\cal N}=4$ Super Yang-Mills (SYM) theory, these are MHV and $\overline{\rm MHV}$ superamplitudes which we associate with two blobs,
\begin{equation}
    \vcenter{\hbox{ \onshellgraph{ {(0,0)/2, (1,1.732)/1, (2,0)/3} }{ {} }{ { (1,0.577)/B1} }{ {{1,B1},{2,B1},{3,B1}} }{0.8} }} \hspace{-1em} = \frac{\delta(P) \delta(\mathcal Q)}{\langle 12 \rangle \langle 23 \rangle \langle 31 \rangle}, \hspace{2em} \vcenter{\hbox{ \onshellgraph{ {(0,0)/2, (1,1.732)/1, (2,0)/3} }{ { (1,0.577)/W1} }{ {} }{ {{1,W1},{2,W1},{3,W1}} }{0.8} }} \hspace{-1em} = \frac{\delta(P) \delta(\tilde {\mathcal Q})}{[12][23][31]}
\end{equation}
%
where the momentum delta function is
\begin{equation}
\delta(P) \equiv \delta^{2{\times}2}(\lambda\cdot\widetilde{\lambda})
\end{equation}
and $\cdot$ is the index space (in this case $\lambda{\cdot}\widetilde{\lambda}=\lambda_1\widetilde{\lambda}_1+\lambda_2\widetilde{\lambda}_2+\lambda_3\widetilde{\lambda}_3$). The super-momentum delta functions are
\begin{equation}
\delta({\cal Q}) \equiv \delta^{2{\times}4}(\lambda_1\widetilde{\eta}_1+\lambda_2\widetilde{\eta}_2+\lambda_3\widetilde{\eta}_3),\quad 
\delta(\widetilde{\cal Q}) \equiv \delta^{1{\times}4}([12]\widetilde{\eta}_3+[23]\widetilde{\eta}_1+[31]\widetilde{\eta}_2).
\end{equation}
We can use the amalgamation procedure and glue these vertices into an \emph{on-shell diagram}. Each on-shell diagram represents a full scattering amplitude evaluated on the \emph{cut} in the context of generalized unitarity,
\begin{equation}
\qquad \qquad 
 \raisebox{-17mm}{\includegraphics[trim={0cm 0cm 0cm 0cm},clip,scale=0.82]{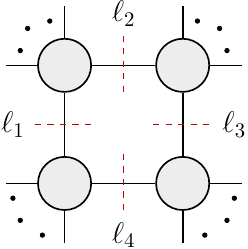}} \ = \int_{\rm Cut} {\cal I} ~ d^4\ell_1\dots d^4\ell_4
\end{equation}
where instead of Minkowski contour we changed the contour to encircle certain poles of the loop integrand. Generalized unitarity dictates that the residue of the amplitude is given by the product of three-point amplitudes. If the on-shell diagram is only a function of kinematics, which happens for $n_I=4L$ where $n_I$ is the number of cut conditions $=$ number of internal edges in the on-shell diagram, we call the diagram a \emph{leading singularity}. 
\begin{equation}\label{eqn:leading_singularity}
\raisebox{-18.5mm}{\includegraphics[trim={0cm 0cm 0cm 0cm},clip,scale=1.1]{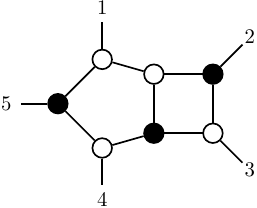}}
\begin{aligned}
     = \prod_k \int d^{\cal N}\widetilde{\eta}_k\int\frac{d^2\lambda_k\,d^2\widetilde{\lambda}_k}{{\rm GL}(1)}\left(\prod_j A_3^{(j)}\right)
\end{aligned}
\end{equation}
and the associated superfunction takes the form,
\begin{equation}
{\cal F}_G = F(\lambda,\widetilde{\lambda},\widetilde{\eta})\,\delta^4({\cal Q})\delta^4(P)\qquad \mbox{where} \qquad \delta(P) \equiv \delta^{2{\times}2}(\lambda{\cdot}\widetilde{\lambda}),\,\,\,\delta({\cal Q}) \equiv \delta^{2{\times}4}(\lambda\cdot\widetilde{\eta})
\end{equation}
where $F(\lambda,\widetilde{\lambda},\widetilde{\eta})$ is the superfunction that depends on $\lambda_i$, $\widetilde{\lambda}_i$ and is degree $4(k{-}2)$ in the Grassmann variables $\widetilde{\eta}_i$. On-shell diagrams are very important objects in any QFT, but in cut-constructible theories (such as ${\cal N}=4$ SYM), on-shell diagrams contain complete information about loop integrands and can be used for integrand reconstruction via recursion relations. We refer the reader to \cite{Arkani-Hamed:2012zlh} for more details.

Planar on-shell diagrams also appear in mathematics, where they are known as \emph{plabic graphs} (``plabic" stands for ``planar bicolored"). Plabic graphs are associated with cells in the positive Grassmannian $G_+(k,n)$ and can be labeled by (decorated) permutations \cite{Postnikov:2006kva}. The positive Grassmannian is a region inside the real Grassmannian $G(k,n)$ where the ordered Pl\"ucker variables $(i_1\,i_2\dots i_k)$ are nonnegative. The top cell is the $k(n{-}k)$-dimensional space where all Pl\"uckers are positive. Lower-dimensional cells can be described by fixing some Pl\"uckers to be zero, and the rest positive. The positive Grassmannian is an extremely interesting object, with many fascinating combinatorial properties; e.g. it has a topology of the ball with the facets being labeled by permutations. We refer the reader to \cite{Postnikov:2006kva,LusTPFlag,RietschCell,Wil-enumeration,GKL} for a more general discussion. 

The boundary measurement procedure takes a graph with the edges (or faces) labeled by variables $\alpha_i$, and produces a matrix $C(\alpha_i)$. As the $\alpha_i$ range over all positive real numbers, $C(\alpha_i)$ ranges over all elements of the associated positroid cell. In other words, boundary measurement gives a parametrization of the cell. Many different graphs will give a parametrization of the same cell. There are two identity moves that change the diagram but do not change the cell (and also do not change the associated differential form).
\begin{equation}
    \begin{gathered}
        \vcenter{\hbox{\onshellgraphnolabel{ {(0.3,0.3)/1, (2.7,0.3)/4, (2.7,2.7)/3, (0.3,2.7)/2,(2,3.2)/X} }{ {(1,1)/W1, (2,2)/W2} }{ {(1,2)/B1, (2,1)/B2} }{ {{1,W1},{W1,B1},{W1,B2},{B1,2},{W2,B1},{3,W2},{W2,B2},{B2,4}} }{1}}} \hspace{-1em} \longleftrightarrow \hspace{-1em} \vcenter{\hbox{\onshellgraphnolabel{ {(0.3,0.3)/1, (2.7,0.3)/4, (2.7,2.7)/3, (0.3,2.7)/2,(2,3.2)/X} }{ {(1,2)/B1, (2,1)/B2} }{ {(1,1)/W1, (2,2)/W2} }{ {{1,W1},{W1,B1},{W1,B2},{B1,2},{W2,B1},{3,W2},{W2,B2},{B2,4}} }{1}}} \\
        \text{square move}
    \end{gathered}
    \hspace{+2em}
    \begin{gathered}
    \vcenter{\hbox{\onshellgraphnolabel{ {(-0.7,0.3)/1, (2.7,0.3)/4, (2.7,2.7)/3, (-0.7,2.7)/2,(1,3.2)/5} }{ {(0,1)/W1, (2,2)/W2, (1,2.5)/W3} }{ {(0,2)/B1, (2,1)/B2} }{ {{1,W1},{W1,B1},{W1,B2},{B1,2},{W3,B1},{3,W2},{W3,W2},{W2,B2},{B2,4},{W3,5}} }{1}}} \longleftrightarrow   \vcenter{\hbox{\onshellgraphnolabel{ {(0.3,0.3)/1, (2.7,0.3)/4, (3.2,1.8)/3, (0.3,2.7)/2,(1.8,3.2)/5} }{ {(1,1)/W1, (2,2)/W2,(2.5,2.5)/W3} }{ {(1,2)/B1, (2,1)/B2} }{ {{1,W1},{W1,B1},{W1,B2},{B1,2},{W2,B1},{W2,B2},{B2,4},{W2,W3},{W3,5},{W3,3}} }{1}}} \\
    \text{merge/expand move}
    \end{gathered}
\end{equation}
If boundary measurement is injective, the plabic graph is called \emph{reduced}. A graph is reduced if and only if doing identity moves to it never produces a bubble, shown below.

\begin{equation}\label{eqn:bubble}
\vcenter{\hbox{\begin{tikzpicture}
  \draw[very thick,black] (0,0) circle (0.5);
  \draw[very thick,black]  (-1.5,0) -- (-0.5,0);
  \draw[very thick,black]  (0.5,0) -- (1.5,0);
  \filldraw[fill=white,draw=black] (-0.5,0) circle (0.2); 
  \filldraw[fill=black,draw=black] ( 0.5,0) circle (0.2); 
\end{tikzpicture}}}
\end{equation}

When a graph corresponds to the top cell, the matrix $C(\alpha_i)$ has all Pl\"ucker coordinates positive when $\alpha_i >0$. Lower-dimensional cells can be obtained from the top cell by setting certain edge variables to zero, which removes the corresponding internal edges in the graph and also makes some Pl\"ucker coordinates zero.

Note that the parameters of the positive Grassmannian $n$, $k$ are encoded respectively by the number of external legs of the plabic graph, and 
\begin{equation}
k = 2n_B + n_W - n_I
\end{equation}
where $n_B$, $n_W$ are the numbers of black and white vertices and $n_I$ is the number of internal edges of the graph. 

The connection between the mathematics of plabic graphs and on-shell diagrams in physics is not just at the graphical level. Very surprisingly, the superfunction associated with a planar on-shell diagram in ${\cal N}=4$ SYM theory can be obtained from the plabic graph and associated boundary measurement $C$-matrix, as long as the graph is reduced. We start with a natural logarithmic form $\omega_G$ for a plabic graph,
\begin{equation}
\omega_G = \frac{d\alpha_1}{\alpha_1}\frac{d\alpha_2}{\alpha_2} \dots \frac{d\alpha_m}{\alpha_m}
\end{equation}
where $\alpha_1, \dots, \alpha_m$ are the edge variables, and $m$ is the dimension of the associated cell in $G_+(k,n)$. The form $\omega_G$ can be also written in terms of maximal minors of the $C$-matrix as
\begin{equation}
\omega_G = \oint \frac{dC}{(12\dots k)(23\dots k{+}1)\dots (n1\dots k{-}1)}
\end{equation}
where the $GL(k)$ invariant measure on the Grassmannian is
\begin{equation}
dC \equiv \frac{d^{k{\times}n}C}{{\rm GL}(k)} \label{measure}
\end{equation}
and the contour is taken as the $k(n{-}k)-m$ dimensional residue. This form is the \emph{canonical form} of the cell, in the sense of positive geometries \cite{Arkani_Hamed_2017}.

To obtain the superfunction ${\cal F}_G$ we decorate $\omega_G$ with certain delta functions and integrate over all edge variables. The delta functions link the edge variables to kinematical variables $\lambda_i$, $\widetilde{\lambda}_i$ for the scattering process. Using
\begin{equation}
\int dx\,f(x) \delta(x-a) = f(a)
\end{equation}
we write the formula for the superfunctions ${\cal F}_G$ as 
\begin{equation}\label{eqn:superfunction}
{\cal F}_G = \int \frac{d\alpha_1}{\alpha_1}\frac{d\alpha_2}{\alpha_2} \dots \frac{d\alpha_m}{\alpha_m}\,\delta^{k{\times}2}(C\cdot\widetilde{\lambda})\,\delta^{(n{-}k){\times} 2}(C^\perp{\cdot}{\lambda})\,\delta^{k{\times}4}(C\cdot\widetilde{\eta}).
\end{equation}
The first two delta functions have a simple interpretation as a linearization of the momentum conservation. In terms of spinor helicity variables, this is the orthogonality of the $\lambda$ and $\widetilde{\lambda}$ planes in the $n$-dimensional index space (of $n$ external particles). 
\begin{equation}
\sum_{i=1}^n p_i^\mu = 0 \quad \leftrightarrow \quad \lambda\cdot\widetilde{\lambda} = 0
\end{equation}
This is a constraint on the spinor-helicity variables which is encoded in the delta function $\delta(\lambda\cdot\widetilde{\lambda})$. However, this is a quadratic condition that we can linearize using a $k$-dimensional plane $C$. We demand that the $C$ plane contains the $\lambda$ plane and is orthogonal to the $\widetilde{\lambda}$ plane, 
\begin{equation}
\raisebox{-15mm}{\includegraphics[trim={0cm 0cm 0cm 0cm},clip,scale=1.1]{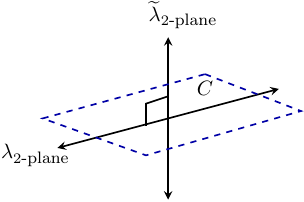}}.
\end{equation}
This automatically forces the orthogonality of both 2-planes and hence the momentum conservation. These geometric constraints are then encoded in the two delta functions 
\begin{equation}
\delta^{k{\times}2}(C\cdot\widetilde{\lambda})\,\delta^{(n{-}k){\times} 2}(C^\perp{\cdot}{\lambda})
\end{equation}
where the $C^\perp$ is the $(n{-}k)$-dimensional plane that contains the $\widetilde{\lambda}$ plane and is orthogonal to $C$. These two sets of conditions impose more than 4 conditions (orthogonality of $\lambda$, $\widetilde{\lambda}$ planes = momentum conservation). It imposes additional $2n{-}4$ constraints which localize $2n{-}4$ edge variables $\alpha_i$ as functions of spinor helicity variables $\alpha_i^\ast=\alpha_i(\lambda_j,\widetilde{\lambda}_k)$,
\begin{equation}\label{eqn:delta}
\delta^{k{\times}2}(C\cdot\widetilde{\lambda})\,\delta^{(n{-}k){\times} 2}(C^\perp{\cdot}{\lambda}) = \delta^{2{\times}2}(\lambda\cdot\widetilde{\lambda})\,\prod_{i{=}1}^{2n{-}4}\delta(\alpha_i-\alpha_i^\ast)
\end{equation}
For the reduced on-shell diagram associated to a cell of dimension $m=2n{-}4$ in the positive Grassmannian $G_+(k,n)$, all edge variables are determined as functions in kinematics. Then the superfunction $\mathcal{F}_G$ is a \emph{leading singularity}. The additional fermionic delta function $\delta^{k{\times}4}(C\cdot\widetilde{\eta})$ does not impose any more constraints, but should be understood as a polynomial in the Grassmann variables $\widetilde{\eta}_i$,
\begin{equation}
\delta^{k{\times}4}(C\cdot\widetilde{\eta}) = \delta^{2{\times}4}({\cal Q})\,\delta^{(k{-}2){\times4}}(\Xi)
\end{equation}
where $\delta^{(k{-}2){\times4}}(\Xi)\equiv \Xi^4$, where $\Xi$ is the homogenous polynomial of degree $k{-}2$ in the Grassmann variables $\widetilde{\eta}_k$. The leading singularity on-shell diagrams are the most prominent objects which play a crucial role in the scattering amplitudes. Note that if $m<2n{-}4$, some of the edge variables are unfixed -- in the physical picture, these are parameters of loop momenta which are not localized on cuts. If $m>2n{-}4$ there are more delta functions than edge variables, and apart from solving for all edge variables we get additional delta functions on $\lambda$, $\widetilde{\lambda}$ (beyond momentum conservation). 

Our case of interest is MHV amplitudes, where $k=2$. The top cell of $G_+(2,n)$ has dimension $2n{-}4$, hence the associated on-shell diagram is a leading singularity. The geometric constraint of the delta functions \eqref{eqn:delta} identifies the $C$-plane with the $\lambda$-plane as this is an unique 2-plane which contains a fixed 2-plane $\lambda$. So we have
\begin{equation}
C = \lambda
\end{equation}
and the $(2{\times}2)$ minors of the $C$ matrix $(ij)$ are identified with the angle brackets as 
\begin{equation}
(ij) = \langle ij\rangle \qquad \langle ij\rangle = \epsilon_{ab} \lambda_a^{(i)}\lambda_b^{(j)}.
\end{equation}
The canonical form on the top cell of $G_+(2,n)$ is given by
\begin{equation}
\omega_G = \frac{dC}{(12)(23)(34)\dots (n{-}1\,n)(n1)}
\end{equation}
where $dC$ is defined in (\ref{measure}) and no residue is taken because $C$ is a top cell. The associated superfunction for the on-shell diagram is
\begin{equation}
{\cal F}_G = \frac{\delta(P)\,\delta({\cal Q})}{\langle 12\rangle\langle 23\rangle\langle 34\rangle \dots \langle n1\rangle}.
\end{equation}
We can talk about three different objects: the canonical form $\omega_G$, the canonical function $f_G$ (which is just $\omega_G$ without a measure) and the superfunction $\mathcal{F}_G$. For $k=2$, the differences are unimportant, and hence it is enough to consider only the canonical function
\begin{equation} \label{eq:PT-top-cell}
f_G = \frac{1}{(12)(23)(34)\dots(n1)}.
\end{equation}
In the remainder of the paper, we will consider canonical functions almost exclusively. We draw the reader's attention to some particularly pertinent points.
\begin{itemize}
\item For each $n$ there is only {\bf one} cell of $G_+(2,n)$ with canonical function as in \eqref{eq:PT-top-cell}, which is the top cell. This cell has many different on-shell diagrams, related by the identity moves. Examples of these on-shell diagrams for $n=4,5$ are
\begin{equation}
\vcenter{\hbox{ \onshellgraph{ { (1.5,3.2)/{$G_+(2,4)$},(0.3,0.3)/1, (2.7,0.3)/4, (2.7,2.7)/3, (0.3,2.7)/2} }{ {(1,2)/W1, (2,1)/W2} }{ {(1,1)/B1, (2,2)/B2} }{ {{1,B1},{W1,B1},{W1,B2},{W1,2},{W2,B1},{3,B2},{W2,B2},{W2,4}} }{1} } } \hspace{2em} \vcenter{\hbox{ \onshellgraph{ {(2.5,3.2)/{$G_+(2,5)$},(0.3,0.3)/1,(0.3,2.7)/2,(4.7,2.7)/3,(4.7,0.3)/4,(2,0.3)/5} }{ {(1,2)/W1,(2,1)/W2,(4,2)/W3,(3,1)/W4} }{ {(1,1)/B1,(3,2)/B2,(4,1)/B3} }{ {{1,B1},{2,W1},{3,W3},{4,B3},{5,W2},{B1,W1},{B1,W2},{B2,W1},{B2,W4},{B2,W3},{B3,W4},{B3,W3},{W2,W4}} }{1} }}
\end{equation}
All diagrams for the top cell give rise to the same canonical function. That is, the canonical function is also invariant under identity moves.

\item Not every on-shell diagram with $m=2n{-}4$ edge variables is a leading singularity diagram---only the reduced ones are. Consider the following non-reduced diagram
\[
    \begin{aligned}
       \vcenter{\hbox{ \onshellgraph{ { (0.3,0.3)/1, (2.7,0.3)/4, (2.7,2.7)/3, (0.3,2.7)/2} }{ {(1,2)/W1, (2,2)/W2} }{ {(1,1)/B1, (2,1)/B2} }{ {{1,B1},{W1,B1},{W1,W2},{W1,2},{W2,B2},{3,W2},{B2,4},{B1,B2}} }{1} } }  \rightarrow  \qquad \begin{tikzpicture}[baseline=(current bounding box.center)]
   \draw[very thick,black] (0,0) circle (0.5);
  \filldraw[fill=white,draw=black] (0,.6) circle (0.2); 
  \filldraw[fill=black,draw=black] ( 0,-.6) circle (0.2); 

  \draw[very thick,black] (-0.1,.75) -- (-0.5,1.3); 
  \draw[very thick,black] (0.1,.75) -- ( 0.5,1.3); 
  \draw[very thick,black] (-0.1,-.75) -- (-0.5,-1.3); 
  \draw[very thick,black] (0.1,-.75) -- ( 0.5,-1.3); 

  \node at (-0.7,1.5) {$2$};
  \node at (0.7,1.5) {$3$};
  \node at (-0.7,-1.5) {$1$};
  \node at ( 0.7,-1.5) {$4$};
\end{tikzpicture}
    \end{aligned}
\]
where we used the merge-expand move to expose the bubble. Because the dimension is still $m=2n{-}4=4$, the presence of the bubble leads to an edge variable $\alpha_4$ which is not solved for from delta functions, leading to an `excess' delta function which imposes the below constraints on external kinematics
\begin{align}
{\cal F}_G &= \int \frac{d\alpha_1}{\alpha_1}\frac{d\alpha_2}{\alpha_2}\frac{d\alpha_3}{\alpha_3}\frac{d\alpha_4}{\alpha_4} \delta^{2{\times}2}(C{\cdot}\widetilde{\lambda})\,\delta^{2{\times}2}(C^\perp{\cdot}{\lambda})\,\delta^{2{\times}4}(C{\cdot}\widetilde{\eta})\nonumber\\
&= \int \frac{d\alpha_4}{\alpha_4} \frac{\delta(P)\,\delta({\cal Q})}{\la 23\ra\la 34\ra\la 41\ra}\delta(\la 12\ra).
\end{align}
This is in contrast with the reduced diagram where the edge variable $\alpha_4$ would be solved for from the delta function.
\end{itemize}

\subsection*{Boundary structure of $G_+(2,n)$}

It will be useful later to recall the following interpretation of the cells of $G_+(2,n)$ as point configurations on the projective line $\mathbb{P}^1$. Suppose $C$ is in the top cell, so has all positive Pl\"ucker coordinates. If we rescale each column of $C$ by a coefficient $c_i>0$ to make the first entry 1, we obtain
\begin{equation}
C' = \left(\begin{array}{cccccc} 1 & 1 & 1 & \dots & 1 & 1\\ x_1 & x_2 & x_3 & \dots & x_{n{-}1} & x_n \end{array}\right) 
\end{equation}
The minor $(ij)$ of the matrix for $i<j$ is
\begin{equation}
(ij) = c_ic_j (x_j-x_i) > 0 \quad\mbox{and hence}\quad x_j>x_i
\end{equation}
Therefore, we can interpret $C'$ as a configuration of $n$ ordered points on $\mathbb{P}^1$, as below.
\begin{center}
    \begin{tikzpicture}
        \draw[line width=1pt] (0,0) -- (6,0);

        \node[label={[yshift=0.06cm]$1$},circle,fill,inner sep=2pt,outer sep=0pt] at (1,0) {};
        \node[label={[yshift=0.06cm]$2$},circle,fill,inner sep=2pt,outer sep=0pt] at (2,0) {};
        \node[label={[yshift=0.06cm]$\cdots$}] at (3,0) {};
        \node[label={[yshift=0.06cm]$n-1$},circle,fill,inner sep=2pt,outer sep=0pt] at (4,0) {};
        \node[label={[yshift=0.06cm]$n$},circle,fill,inner sep=2pt,outer sep=0pt] at (5,0) {};
    \end{tikzpicture}
\end{center}
The only codimension 1 boundaries of the space correspond to merging of two adjacent points, $x_i=x_{i{+}1}$ (note that this includes $x_n=x_1$). There are two types of codimension-2 boundaries of the space, e.g. if $(12)=0$, we can set $(23)=0$ by either merging all three points $x_1=x_2=x_3$ or removing the point $2$ completely by setting $c_2=0$.
\begin{center}
    \begin{center}
    \begin{tikzpicture}
        \draw[line width=1pt] (0,0) -- (6,0);
        \node[label={[yshift=0cm]$1,2$},circle,fill,inner sep=2pt,outer sep=0pt] at (1,0) {};
        \node[label={[yshift=0.06cm]$3$},circle,fill,inner sep=2pt,outer sep=0pt] at (2,0) {};
        \node[label={[yshift=0.06cm]$\cdots$}] at (3,0) {};
        \node[label={[yshift=0.06cm]$n-1$},circle,fill,inner sep=2pt,outer sep=0pt] at (4,0) {};
        \node[label={[yshift=0.06cm]$n$},circle,fill,inner sep=2pt,outer sep=0pt] at (5,0) {};

        \draw[line width = 1pt,->] (2,-0.25) -- (1,-1);
        \draw[line width = 1pt,->] (4,-0.25) -- (5,-1);

        \draw[line width=1pt] (-4.5,-1.5) -- (1.5,-1.5);
        \node[label={[yshift=0cm]$1,2,3$},circle,fill,inner sep=2pt,outer sep=0pt] at (-3,-1.5) {};
        \node[label={[yshift=0.06cm]$\cdots$}] at (-1.5,-1.5) {};
        \node[label={[yshift=0.06cm]$n-1$},circle,fill,inner sep=2pt,outer sep=0pt] at (-0.5,-1.5) {};
        \node[label={[yshift=0.06cm]$n$},circle,fill,inner sep=2pt,outer sep=0pt] at (0.5,-1.5) {};

        \draw[line width=1pt] (4.5,-1.5) -- (10.5,-1.5);
        \node[label={$1$},circle,fill,inner sep=2pt,outer sep=0pt] at (5.5,-1.5) {};
        \node[label=$3$,circle,fill,inner sep=2pt,outer sep=0pt] at (6.5,-1.5) {};
        \node[label=$\cdots$] at (7.5,-1.5) {};
        \node[label=$n-1$,circle,fill,inner sep=2pt,outer sep=0pt] at (8.5,-1.5) {};
        \node[label=$n$,circle,fill,inner sep=2pt,outer sep=0pt] at (9.5,-1.5) {};
    \end{tikzpicture}
\end{center}
\end{center}
This simple picture provides a complete stratification of the space, all the way down to the lowest zero-dimensional boundaries which correspond to two points on the line (all other points are removed). At the level of on-shell diagrams, going to the boundary of the positive Grassmannian and taking the residue of $f_G$ corresponds to erasing an edge.

\section{Non-planar On-shell Diagrams}\label{sec:non-planar}

On-shell diagrams are well-defined objects even when the legs of the graph (cuts of loop amplitude) do not form a planar diagram. Here, a planar diagram is one that can be drawn in the disc with no crossing edges \emph{and} with the external legs on the boundary of the disk, as in the example below.
\begin{equation}
    \begin{aligned}
       \vcenter{\hbox{ \onshellgraph{ { (0.3,-0.4)/5, (2.7,0.3)/4, (2.7,2.7)/3, (0.3,3.4)/2,(-2.1,2.7)/1,(-2.1,0.3)/6} }{ {(1,2)/W1, (2,1)/W2,(0.3,2.7)/W3,(-0.4,1)/W4,(0.3,0.3)/W5,(-1.4,2.0)/W6} }{ {(1,1)/B1, (2,2)/B2,(-0.4,2.0)/B3,(-1.4,1.0)/B4} }{ {{W5,B1},{W1,B1},{W1,B2},{W2,B2},{3,B2},{W3,2},{W2,4},{B1,W2},{W1,W3},{B3,W3},{W5,5},{W5,W4},{W4,B3},{W6,B3},{W6,B4},{B4,W4},{W6,1},{B4,6}} }{1} } } 
    \end{aligned}
\end{equation}
Note that this notion of planarity differs from the standard definition in graph theory, where there are no external legs. 

In the physics setup, each on-shell diagram is dressed with a \emph{color factor} which can be written as a product of structure constants of the $SU(N)$ color group. We can also write same color functions in terms of traces of generators of $SU(N)$. In the planar $N\rightarrow\infty$ limit, the single traces dominate, while multiple traces are subleading. Only planar on-shell diagrams contain single trace contributions, while non-planar diagrams start the expansion at a subleading order (multiple trace). This limit is also called the \emph{planar limit} as the planar diagrams dominate for large $N$. The leading singularity at finite $N$ is given by the on-shell diagram where each vertex is dressed by the $SU(N)$ structure constant $f^{abc}$. The result is a product of the group part (given by the product of structure constants) and the kinematical function ${\cal F}_G$. For large $N$ we can expand the group part in powers of $1/N$ where the dominating piece is a single trace of $SU(N)$ generators $T^a$, while the subleading pieces are multiple traces (products of traces). For the planar diagrams we then get
\begin{equation}
 \vcenter{\hbox{ \onshellgraph{ {(0.3,0.8)/{$f^{a_6 a_2 a_7}$},(-0.5,1.5)/2, (2.6,3.1)/{$f^{a_7 a_3 a_8}$}, (1.75,4)/3, (3.5,2.65)/{$f^{a_8 a_9 a_{10}}$},(5.3,2.2)/{$f^{a_{10} a_4 a_{11}}$},(5.1,2.9)/4,(5.3,0.8)/{$f^{a_{11} a_5 a_{12}}$},(3.5,0.4)/{$f^{a_{12} a_9 a_{13}}$},(2.6,-0.1)/{$f^{a_{13} a_1 a_6}$},(5.1,0.1)/5,(1.75,-1)/1} }{ {(1.75,3)/W1,(1.75,0)/W2,(4.4,2.2)/W3,(3,0.8)/W4} }{ {(0.5,1.5)/B1,(3,2.2)/B2,(4.4,0.8)/B3} }{ {{2,B1},{3,W1},{4,W3},{5,B3},{1,W2},{B1,W1},{B1,W2},{B2,W1},{B2,W4},{B2,W3},{B3,W4},{B3,W3},{W2,W4}} }{1} }}  = \left[{\rm Tr}(T^{a_1}T^{a_2}\dots T^{a_{n}})+\dots\right] \times {\cal F}_G
\end{equation}
where we sum over internal labels $a_6,{\dots},a_{13}$. The non-planar diagrams are subleading at large $N$, they only contain multiple traces. In other words, they are suppressed by ${\cal O}(1/N)$ or higher when compared to their planar counterparts
\begin{equation}
 \vcenter{\hbox{ \onshellgraph{ {(0.5,-1.2)/1,(1.5,-1.1)/$f^{a_6 a_1 a_7}$,(0,1.5)/$f^{a_7 a_2 a_8}$,(-0.5,1)/2,(1.7,3.)/$f^{a_8 a_3 a_9}$,(4,-0.5)/$f^{a_{12} a_{13} a_{6}}$,(4,2.5)/$f^{a_9 a_{10} a_{11}}$,(5.2,1.5)/$f^{a_{11} a_{5} a_{12}}$,(2.2,1.5)/$f^{a_{10} a_{4} a_{13}}$,(0.5,3.2)/3,(1.8,1)/4,(5.2,1)/5} }{ {(1.2,-0.5)/W1,(1.2,2.5)/W2,(3.5,2)/W3,(3.5,0)/W4} }{ {(0.5,1)/B1,(2.5,1)/B2,(4.5,1)/B3} }{ {{1,W1},{B1,2},{W2,3},{W1,B1},{B1,W2},{W2,W3},{W3,B3},{W3,B2},{B2,W4},{B3,W4},{W1,W4},{B2,4},{B3,5}} }{1} }} =  {\cal O}\left(\frac{1}{N}\right) \times {\cal F}_G
\end{equation}
However, for finite $N$ we have to consider both planar and non-planar on-shell diagrams, as they all represent valid cuts of the amplitude. Planar diagrams come with an ordering of external legs on the disc, but non-planar diagrams do not have any natural ordering. 

Non-planar on-shell diagrams provide a way into the problem of ${\cal N}=4$ SYM amplitudes, which is challenging for a number of reasons: difficulty in defining an unique non-planar integrand, absence of dual conformal symmetry, no amplituhedron construction, no powerful symbol methods and no known integrability. On the other, the principles of generalized unitarity are universal, and non-planar on-shell diagrams are non-planar amplitudes evaluated on the cuts. Hence, they provide a perfect window to study non-planar ${\cal N}=4$ SYM amplitudes in a setup which allow one to more easily search for new symmetries and connections to mathematics. 

\subsection{First look: Grassmannian formula}

In physics, non-planar on-shell diagrams are defined as a product of three-point amplitudes in the same way as their planar counterparts. The result is a superfunction with the same $\delta(P)\delta(Q)$ conserving delta functions, obtained as in \ref{eqn:leading_singularity},
\begin{equation} \label{eqn:5ptnon-planar}
 \vcenter{\hbox{ \onshellgraph{ {(1.0,-.7)/1,(0,1)/2,(1.0,2.7)/3,(2.3,1)/4,(4.7,1)/5} }{ {(1.7,0)/W1,(1.7,2)/W2,(3.5,1.5)/W3,(3.5,0.5)/W4} }{ {(1,1)/B1,(3,1)/B2,(4,1)/B3} }{ {{1,W1},{B1,2},{W2,3},{W1,B1},{B1,W2},{W2,W3},{W3,B3},{W3,B2},{B2,W4},{B3,W4},{W1,W4},{B2,4},{B3,5}} }{1} }}
= \frac{\la13\ra\,\delta(P)\delta({\cal Q})}{\la12\ra\la23\ra\la 14\ra\la 15\ra\la 34\ra\la 35\ra}
\end{equation}
but no cyclic symmetry due to the loss of ordering. The general amalgamation prescription allows us to evaluate any non-planar on-shell diagram as the kinematical function (and also the color factor to get the complete dressed formula). For MHV leading singularity diagrams, the superconformal invariance dictates (as in the planar sector) that the result is a function of angle brackets $\la{.}{.}\ra$ only, 
\begin{equation}
{\cal F}_G = \frac{{\cal N}}{\prod_{ij}\la ij\ra}\,\delta(P)\delta({\cal Q})
\end{equation}
where the only poles are physical $\la ij\ra$ (not true for NMHV where spurious poles can appear in leading singularities, even in the planar sector). The numerator ${\cal N}$ is a polynomial in angle brackets $\la{.}{.}\ra$ which guarantees correct little group and mass dimension of $f_G$. Note that there is only one planar MHV leading singularity for each $n$, for which the numerator is just ${\cal N}=1$. 

The connection between non-planar on-shell diagrams and the Grassmannian retains some features from the planar case. In a non-planar graph, we can assign edge variables $\alpha_j$ and construct the $C$-matrix using a boundary measurement. The matrix $C(\alpha)$ no longer has any special positivity properties for $\alpha_j>0$ and there is no connection to cells in $G_+(k,n)$ but the $C$-matrix does label a subset of the Grassmannian $G(k,n)$. The amalgamation procedure (and its independence on the planarity of the graph) ensures that the canonical function $f_G$ can be obtained by the same formula \eqref{eqn:superfunction} as in the planar case. For the MHV case, $G(2,n)$, the delta functions again identify the $C$-matrix with the $\lambda$-matrix, $C=\lambda$, and each on-shell diagram is then given by the rational function $f_G$ on the Grassmannian $G(2,n)$. For the 5pt example above, analogously to the planar case (\ref{eq:PT-top-cell}),
\begin{equation} \label{eqn:5ptfunction}
f_G =  \frac{(13)}{(12)(23)(14)(15)(34)(35)}.
\end{equation}
Note that the role of edge variables and minors $(ij)$ are very similar to the planar case: $\alpha_k=0$ sends some minor $(ij)=0$ and erases an edge in an on-shell diagram. This strongly suggests there should be a positive geometry, like the positive Grassmannian $G_+(k,n)$ for planar on-shell diagrams, which is associated for each non-planar diagram.

\subsection{MHV diagrams and the triplet formula}\label{sec:triplet-formula}

The prescription to calculate $f_G$ for a leading singularity using edge variables and boundary measurements works for any on-shell diagram, planar or non-planar, for any $k$, as it follows from the purely local amalgamation procedure. MHV on-shell diagrams are more special as explored in detail in \cite{Arkani-Hamed:2014bca}. In particular, it was shown that there is an important shortcut to calculate $f_G$ using a \emph{triplet formula}. For an MHV leading singularity on-shell diagram with $n$ external legs, $n_W$ white vertices, $n_B$ black vertices, and $n_I$ internal edges, 
\begin{equation}
k = 2n_B + n_W - n_I = 2.
\end{equation}
Combining this with two expressions $n_W+n_B=3n-8$, $n_I=n_B+n_W+n-4$ which are valid for any leading singularity graph, we get
\begin{equation}
n_B = n-2.
\end{equation}
It can be shown that each of these black vertices are attached to exactly three external legs via white vertices. 
\begin{equation}
    \vcenter{\hbox{\begin{tikzpicture}
        \node(1) at (1,-0.7) {1};
        \node(2) at (0,1) {2};
        \node(3) at (1,2.7) {3};
        \node(4) at (2.3,1) {4};
        \node(5) at (4.7,1) {5};

        \node(W1)[draw,circle,fill=white] at (1.7,0) {};
        \node(W2)[draw,circle,fill=white] at (1.7,2) {};
        \node(W3)[draw,circle,fill=white] at (3.5,1.5) {};
        \node(W4)[draw,circle,fill=white] at (3.5,0.5) {};

        \node(B1)[draw,circle,fill] at (1,1) {};
        \node(B2)[draw,circle,fill] at (3,1) {};
        \node(B3)[draw,circle,fill] at (4,1) {};

        \draw[line width=1pt] (B1) -- (W1) -- (1);
        \draw[line width=1pt] (B1) -- (2);
        \draw[line width=1pt] (B1) -- (W2) -- (3);
        \draw[line width=1pt] (B2) -- (4);
        \draw[line width=1pt] (B2) -- (W3) -- (W2);
        \draw[line width=1pt] (B2) -- (W4) -- (W1);
        \draw[line width=1pt] (B3) -- (5);
        \draw[line width=1pt] (B3) -- (W3);
        \draw[line width=1pt] (B3) -- (W4);

        \draw[red,->,line width=1pt] (1,0.93) -- (0.2,0.93);
        \draw[red,->,line width=1pt] (0.95,0.9) .. controls (1.7,0) and (1.7,0) .. (1.1,-0.5);
        \draw[red,->,line width=1pt] (0.95,1.1) .. controls (1.7,2) and (1.7,2) .. (1.1,2.5);

        \draw[ForestGreen,->,line width=1pt] (3,0.93) -- (2.5,0.93);
        \draw[ForestGreen,->,line width=1pt] (2.95,0.9) .. controls (4.1,0.3) and (1.7,0.2) .. (1.2,-0.6);
        \draw[ForestGreen,->,line width=1pt] (2.95,1.1) .. controls (4.1,1.7) and (1.7,1.8) .. (1.2,2.6);

        \draw[blue,->,line width=1pt] (4,0.93) -- (4.5,0.93);
        \draw[blue,->,line width=1pt] (3.9,1) .. controls (3.5,0) and (1.7,0.2) .. (1.3,-0.7);
        \draw[blue,->,line width=1pt] (3.9,1) .. controls (3.5,2) and (1.7,1.8) .. (1.3,2.7);

        \node(B1)[draw,circle,fill] at (1,1) {};
        \node(B2)[draw,circle,fill] at (3,1) {};
        \node(B3)[draw,circle,fill] at (4,1) {};
    \end{tikzpicture}}}
    \hspace{+3em}
    T = \{{\color{red} (123)},{\color{ForestGreen} (134)},{\color{blue} (135)}\}
\end{equation}
Thus, labeling the black vertices with $1, \dots, n-2$, an MHV leading singularity diagram is encoded by $n-2$ triplets $\tau_i$ consisting of the three external legs black vertex $i$ is attached to. We can see how this works for the following examples for $n=4,5,6$ 
\[
    \begin{aligned}
       \vcenter{\hbox{ \onshellgraph{ { (0.3,0.3)/1, (2.7,0.3)/4, (2.7,2.7)/3, (0.3,2.7)/2} }{ {(1,2)/W1, (2,1)/W2} }{ {(1,1)/B1, (2,2)/B2} }{ {{1,B1},{W1,B1},{W1,B2},{W1,2},{W2,B1},{3,B2},{W2,B2},{W2,4}} }{1} } } \hspace{3 cm}
       \vcenter{\hbox{ \onshellgraph{ {(1.0,-.7)/1,(0,1)/2,(1.0,2.7)/3,(2.3,1)/4,(4.7,1)/5} }{ {(1.7,0)/W1,(1.7,2)/W2,(3.5,1.5)/W3,(3.5,0.5)/W4} }{ {(1,1)/B1,(3,1)/B2,(4,1)/B3} }{ {{1,W1},{B1,2},{W2,3},{W1,B1},{B1,W2},{W2,W3},{W3,B3},{W3,B2},{B2,W4},{B3,W4},{W1,W4},{B2,4},{B3,5}} }{1} }} \\
     \makebox[0pt][l]{\kern-30em\ensuremath{\displaystyle T} =~\{(124),(234)\} } \makebox[0pt][l]{\kern-13em\ensuremath{\displaystyle T} =~\{(123),(134),(135)\} }
    \end{aligned}
\]

\[
\begin{aligned}
       \vcenter{\hbox{ \onshellgraph{ { (4.15,1.73)/5, (-0.58,4.46)/4, (0.65,0.35)/3, (2.25,1.91)/2,(0.2,2.53)/1,(-0.58,-1)/6} }{ {(0,0)/W1, (0,3.46)/W2,(3,1.73)/W3,(1,2.53)/W4,(0.2,1.13)/W5,(1.8,1.13)/W6} }{ {(-1,1.73)/B1, (2,3.46)/B2,(2,0)/B3,(1,1.73)/B4} }{ {{W1,B1},{B1,W2},{W2,B2},{W3,B3},{B3,W1},{B2,W3},{W2,4},{W3,5},{W1,6},{W4,1},{W4,B4},{B1,W5},{W5,B4},{W6,B4},{W4,B2},{W6,B3},{W5,3},{W6,2}} }{0.8} } } 
 \\
     \makebox[0pt][l]{\kern-15em\ensuremath{\displaystyle T} = \{(123),(145),(256),(346)\}}
\end{aligned}
\]

The triplets allow us to construct a particular parametrization of the $(n{-}2)\times n$ matrix $C^\perp(\Vec{\alpha^*})$ where $\alpha^*$ are evaluated in terms of $\lambda$ spinors. Each triplet corresponds to a row in this matrix, and three non-zero entries are given by the labels in the triplet. For the triplet $(abc)$ we put the bracket $\la bc\ra$ in the $a^{\rm th}$ column, $\la ac\ra$ in the $b^{\rm th}$ column and $\la ab\ra$ in the $c^{\rm th}$ column. We do this operation for all triplets, giving us the parametrizaton of the matrix $C^\perp$. For the simplest 4pt example
\[
    \begin{aligned}
       \vcenter{\hbox{ \onshellgraph{ { (0.3,0.3)/1, (2.7,0.3)/4, (2.7,2.7)/3, (0.3,2.7)/2} }{ {(1,2)/W1, (2,1)/W2} }{ {(1,1)/B1, (2,2)/B2} }{ {{1,B1},{W1,B1},{W1,B2},{W1,2},{W2,B1},{3,B2},{W2,B2},{W2,4}} }{1} } } \ C^\perp = \begin{pmatrix}
 \langle24\rangle & \langle14\rangle & 0 & \langle 12 \rangle\\
0 & \langle 3 4 \rangle & \langle 24 \rangle & \langle 23 \rangle 
\end{pmatrix}
    \end{aligned}
\]
\begin{equation}\label{eq:4pt}
    T = \{ (124),(234)\}
\end{equation}
Next, we construct a matrix $M_{ab}$ by deleting two columns $a,b$. Choosing $a,b = 1,3$ gives
\begin{equation} M_{13} = 
\begin{pmatrix}
 \langle14\rangle & \langle 12 \rangle\\
 \langle 3 4 \rangle & \langle 23 \rangle 
\end{pmatrix}.
\end{equation}
Now we calculate the determinant of this matrix divided by $\la ab\ra$, 
\begin{equation}
     \frac{\det(M_{ab})}{\la ab\ra}.
\end{equation}
This quantity does not depend on the choice of $a,b$. We also note that $\det(M_{ab})$ can be written to have a factor of $\la ab \ra$, so dividing by $\la ab \ra$ yields a polynomial. The formula for the on-shell diagram superfunction is then given by 
\begin{equation}\label{eqn:triplet_form}
    {\cal F}_G = \frac{( \det(M_{ab}) / \la ab\ra )^2}{\prod_{\tau \in T} \la\tau_1\tau_2\ra\la\tau_2\tau_3\ra\la\tau_3\tau_1\ra} \delta(P)\delta({\cal Q})
\end{equation}
where the denominator is given by the product of three poles for each triplet $(\tau_1\tau_2\tau_3)$. In the 4pt example, $\det(M_{13}) = \la 13 \ra\la 24 \ra$. The triplet $(124)$ contributes the poles $\frac{1}{\langle 12 \rangle \langle 14 \rangle \langle 24 \rangle}$ and $(234)$ contributes $\frac{1}{\langle 23 \rangle \langle 24 \rangle \langle 34 \rangle}$. Then bringing everything together, we get the superfunction
\begin{equation}
    {\cal F}_G =  \frac{\langle 13 \rangle^2 \langle 2 4 \rangle^2}{\langle 1 3 \rangle^2} \frac{\delta(P)\delta({\cal Q})}{\langle 12 \rangle \langle 14 \rangle \langle 24 \rangle \langle 23 \rangle \langle 24 \rangle \langle 34 \rangle} = \frac{\delta(P)\delta({\cal Q})}{\langle 1 2 \rangle \langle 2 3 \rangle \langle 3 4 \rangle \langle 1 4 \rangle}
\end{equation}
as expected. We can use the same prescription to compute the form for any planar or non-planar MHV on-shell diagrams. As discussed earlier, minors of the $C$-matrix are equal to $(ij)=\la ij\ra$ after imposing the delta functions. Hence, we can also here just use $(ij)$ instead of $\la ij\ra$ and work with the rational function $f_G$, rather than the superfunction ${\cal F}_G$ for each graph. And also the entries of $C^\perp$ and $M_{ab}$ will be minors $(ij)$ of the $C$ matrix. 

Note that the numerator plays a crucial role in making sure that an associated form has only logarithmic singularities. We can demonstrate it on the 5pt example \eqref{eqn:5ptnon-planar} with triplets $(123),(134),(135)$. If we choose to delete columns $1,3$ the $M_{13}$ matrix is given by (remaining columns are $2,4,5$), 
\begin{equation} M_{13} = 
\begin{pmatrix}
 \la 13\ra & 0 & 0\\
0 & \la 13\ra & 0\\
0 & 0 & \la 13\ra
\end{pmatrix}
\end{equation}
and we get
\begin{equation}
    f_G = \frac{(13)^4}{(13)^3(12)(14)(15)(23)(34)(35)} = \frac{(13)}{(12)(23)(14)(15)(34)(35)}
\end{equation}
in agreement with \eqref{eqn:5ptfunction}. Apart from the trivial cancelation of the pole $(13)^3$, the numerator $N=(13)$ also guarantees that all lower-dimensional residues of $\omega_G$ stay logarithmic. For example, if we send $(12)=0$ by setting $1=\alpha 2$, ie. columns 1 and 2 in the $C$ matrix are proportional, we get for the associated canonical form (with the measure) $\omega_G=f_G\,dC$,
\begin{equation}\label{eqn:residue_example}
 \underset{1 = \alpha 2}{\rm Res}\,\omega_G = \frac{(23)\alpha^2 d\alpha \wedge d^4 C'}{\alpha^3 (23)(24)(25)(34)(35)} = \frac{d\alpha \wedge d^4 C'}{\alpha(24)(25)(34)(35)},
\end{equation}
where $d^4C'$ is the measure on the $G(2,4)$ of the four independent columns. The numerator was crucial in canceling the double pole $\alpha^2$ in the denominator: one factor came from the measure, but the second came from $(13) \rightarrow \alpha (23)$. This is much more involved for complicated higher pole examples, but the numerator always does the `magic' to preserve logarithmic singularities. From now on, we will denote the canonical function $f_T$ just to indicate that we derive it from the triplet formula rather than the on-shell diagram, but they are obviously equal, $f_T=f_G$.

Note that there are also on-shell diagrams with the same number of edges and vertices as the MHV leading singularity diagrams, which are not described by triplets. The canonical functions for these diagrams vanish because of the constraints imposed on external kinematics. In the planar sector, these diagrams are not reduced due to the presence of internal bubbles (\ref{eqn:bubble}), but we do not have the same notion of reducedness for non-planar diagrams yet. In any case, these diagrams are not described by triplets, have vanishing canonical functions, and are not part of our discussion.

\subsection{Parke-Taylor expansion}

The canonical function $f_T$ can be also expanded in terms of elementary building blocks called Parke-Taylor factors,
\begin{equation} \label{eqn:PTdecomp}
    f_T = \sum_{\sigma} {\rm PT}(\sigma) \quad\mbox{where}\quad {\rm PT}(\sigma) = \frac{1}{(\sigma_1\sigma_2)(\sigma_2\sigma_3)(\sigma_3\sigma_4)\dots(\sigma_n\sigma_1)}
\end{equation}
where the Parke-Taylor factor is written for a given permutation $\sigma= \sigma_1 \sigma_2 \dots \sigma_n$ and the sum is over a subset of all permutations given in (\ref{PTexp}).

Note that the Parke-Taylor factor ${\rm PT}(\sigma)$ is the canonical function of the positive Grassmannian $G_+(2,n)$ where the columns are ordered using $\sigma$. In fact, the collection of all such spaces for all orderings tiles the full Grassmannian $G(2,n)$. This is also evident from the picture of points on a projective line $\mathbb{P}^1$. The positive Grassmannian $G_+(2,n)$ with a canonical ordering $123\dots n$ corresponds to a set of ordered $n$ points, as discussed in Section \ref{sec:onshell}. For any point in $G_+(2,n)$ the associated  points are ordered in some way and hence the point also belongs to one of the permuted $G_+(2,n)$ subspaces. As a result, the full Grassmannian $G(2,n)$ is tiled by all permuted $G_+(2,n)$s, which also do not overlap. 

This is no longer true for $k>2$ but $k=2$ is special because of this feature. Note that the expansion of $f_T$ in terms of Parke-Taylor factors makes manifest that all singularities of the form are logarithmic -- this is not manifest in the determinant form as the numerator conspires with the denominator to remove all non-logarithmic poles. 

The prescription how to obtain a Parke-Taylor decomposition of a canonical function $f_T$ for a general set of triplets $T$ was given in \cite{Arkani-Hamed:2014bca}. The first step is to fix the \emph{orientation}, which is a permutation of labels for each triplet, modulo cyclic ordering. For a triplet $(123)$ there are two orientations: $(123)$, $(132)$. The orientation of the triplet changes the sign of $f_T$ in \eqref{eqn:PTdecomp} as in the denominator 
\begin{equation}
    (123)\rightarrow (12)(23)(13),\qquad (132) \rightarrow (13)(32)(12) = -(12)(23)(13)
\end{equation}
while the numerator (a square) does not depend on the orientation. For a fixed orientation, $f_T$ is given by 
\begin{equation}
    f_T = \sum_{\sigma \in S_n^{(T)}} {\rm PT}(\sigma). \label{PTexp}
\end{equation}
where the sum is given over all orderings $\sigma$ which are compatible with oriented triplets $T$
\begin{equation}\label{eqn:ST_definition}
    S_n^{(T)} = \{ \sigma \in S_n : \sigma_1 = 1, ~ \forall(ijk) \in T \ i,j,k \text{ appear in that order up to rotation}\}.
\end{equation}
As an example, let us take the 5pt non-planar diagram with triplets $(123)$, $(134)$, $(135)$. For this orientation we get
\begin{equation}
f_T = {\rm PT}(12345) + {\rm PT}(12354)
\end{equation}
while for the orientation $(123)$, $(134)$, $(153)$ it is 
\begin{equation}
\widetilde{f}_T = {\rm PT}(12534) + {\rm PT}(15234).
\end{equation}
Parke-Taylor factors are not independent but satisfy Kleiss-Kuijf relations and the $U(1)$ decoupling identity,
\begin{equation}
    {\rm PT}(12534) + {\rm PT}(15234) + {\rm PT}(12345) + {\rm PT}(12354) = 0
\end{equation}
and hence $\widetilde{f}_T=-f_T$. This is a general feature of triplets: different orientations give the same expressions up to a sign.

\subsection{Towards a positive geometry}

The basic question we want to address in this paper is if there exists a positive geometry associated with a non-planar on-shell diagram. Namely, 
\begin{center}
{\it Is there a positive geometry for which $f_T$ is the canonical function?}
\end{center}
For planar on-shell diagrams (plabic graphs) the answer is positive. In that case the canonical function is a single Parke-Taylor factor ${\rm PT}(123{\dots}n)$ which is the canonical function of the positive Grassmannian $G_+(2,n)$. For non-planar on-shell diagrams, we use the Parke-Taylor expansion (\ref{PTexp}) and interpret it geometrically as the union of the corresponding spaces. The space corresponding to ${\rm PT}(\sigma)$ is a top cell of the positive Grassmannian $G_+(2,n)$ with a particular ordering $\sigma$. Note that there are many options here, both because of different orientations of triplets and different Parke-Taylor expansions, but also because of additional choices to be discussed in Section 5. 

Here as an example, we consider just a planar 4pt graph (\ref{eq:4pt}), the choice of the orientation of triplets $(123)$, $(134)$ and using (\ref{PTexp}) gives us ${\rm PT}(1234)$ as the canonical function. We can characterize the positive geometry, the top cell of $G_+(2,4)$ by the positivity conditions on minors,
\begin{equation}
(12)>0,\,(13)>0,\,(14)>0,\,(23)>0,\,(24)>0,\,(34)>0 \label{G24space1}
\end{equation}
If we now change the orientation to $(123)$, $(143)$ we get 
\begin{equation}
    f_T = {\rm PT}(1423) + {\rm PT}(1243)
\end{equation}
which is equal to ${\rm PT}(1234)$ up to a sign using the $U(1)$ decoupling identity. If we now associate geometries for each Parke-Taylor factor,
\begin{align}
&{\rm PT}(1423): (12)>0,\,(13)>0,\,(14)>0,\,(23)>0,\,(24)<0,\,(34)<0\\
&{\rm PT}(1243): (12)>0,\,(13)>0,\,(14)>0,\,(23)>0,\,(24)>0,\,(34)<0
\end{align}
the union of these spaces is the region in $G(2,4)$ given by the following inequalities:
\begin{equation}
(12)>0,\,(13)>0,\,(14)>0,\,(23)>0,\,(34)<0. \label{G24space2}
\end{equation}
Note that the sign of the minor $(24)$ is unconstrained. This space has a very different topology than the positive Grassmannian $G_+(2,4)$. As a result, we see that even for a simple planar diagram we can associate multiple spaces with its canonical function $f_T$. In this case, the space (\ref{G24space1}) is strongly preferable to than (\ref{G24space2}), but for a general non-planar diagram the situation is less clear. We illustrate this in Section \ref{sec:regions-and-geometries}.

\section{Determinant Form and Identity Moves for Non-Planar Diagrams}\label{sec:determinant-and-moves}

In this section, we investigate the triplet formula \eqref{eqn:triplet_form} for non-planar diagrams further from a more mathematical perspective. While all on-shell MHV canonical functions from diagrams must come from triplets, the converse is not quite true. Some triplets give a form which is identically zero. We would like to exclude these triplets from our consideration. With this in mind, we first give a combinatorial characterization for when $f_T$ is not the zero polynomial, using the determinantal formula. Second, we use the determinantal formula to give a factorization algorithm for $f_T$ following from work of Castravet--Tevelev \cite{Castravet:2013}. Third, and most significantly, we turn to identity moves, meaning moves on triplets which preserve the form. For the planar case, the only identity move is the square move from Section 2. For arbitrary triplets, there are additional identity moves, discovered by \cite{Castravet:2013,Cachazo:2019}, called \emph{sphere moves}. We show that in fact the only identity moves are sphere moves. In the process, we introduce a new object (``doublets") which encode the same information as the triplets but in a sphere-move-invariant way.

From this point on, except for a few examples, we use a simplified notation for on-shell diagrams shown below, where we omit the external legs and instead number the white vertices adjacent to external legs.

\begin{equation}\label{eq:legless-ex}
     \vcenter{\hbox{ \onshellgraph{ { (4.15,1.73)/5, (-0.58,4.46)/4, (0.65,0.35)/3, (2.25,1.91)/2,(0.2,2.53)/1,(-0.58,-1)/6} }{ {(0,0)/W1, (0,3.46)/W2,(3,1.73)/W3,(1,2.53)/W4,(0.2,1.13)/W5,(1.8,1.13)/W6} }{ {(-1,1.73)/B1, (2,3.46)/B2,(2,0)/B3,(1,1.73)/B4} }{ {{W1,B1},{B1,W2},{W2,B2},{W3,B3},{B3,W1},{B2,W3},{W2,4},{W3,5},{W1,6},{W4,1},{W4,B4},{B1,W5},{W5,B4},{W6,B4},{W4,B2},{W6,B3},{W5,3},{W6,2}} }{0.8} } } 
 \\ \rightarrow   \vcenter{\hbox{ \onshellgraph{ { (3,1.73)/5, (0,3.46)/4, (0.2,1.13)/3, (1.8,1.13)/2,(1,2.53)/1,(0,0)/6} }{ {(0,0)/W1, (0,3.46)/W2,(3,1.73)/W3,(1,2.53)/W4,(0.2,1.13)/W5,(1.8,1.13)/W6} }{ {(-1,1.73)/B1, (2,3.46)/B2,(2,0)/B3,(1,1.73)/B4} }{ {{W1,B1},{B1,W2},{W2,B2},{W3,B3},{B3,W1},{B2,W3},{W4,B4},{B1,W5},{W5,B4},{W6,B4},{W4,B2},{W6,B3}} }{0.8} } }
\end{equation}

\subsection{Non-vanishing function and factorization}

First, we identify which triplets give a non-vanishing function $f_T$. One can show that if any Pl\"ucker coordinate of $C$ is identically zero, then in fact there is a triple $(ijk) \in T$ such that all three Pl\"ucker coordinates $(ij), (ik), (jk)$ are identically zero. In this case, the matrix $C^\perp(\vec \alpha^*)$ is never full-rank, as it has a row of zeros, and $\det(M_{ab})= 0 = f_T$. On the other hand, if all Pl\"ucker coordinates of $C$ are not identically zero, then $C^\perp(\vec \alpha^*)$ is generically full-rank, and represents the subspace perpendicular to $C$. So there is a polynomial $g$ in the Pl\"ucker coordinates such that for all $a,b$, $\det(M_{ab}) = \pm (ab) g$. We have that $g$ is generically nonzero, as $M_{ab}$ is generically full-rank. So in this case, $f_T$ is not identically zero.

We can also characterize entirely using the triples $T$ when $f_T$ is not the zero polynomial. Roughly, $f_T$ is nonvanishing if every subset of $T$ contains ``enough points". The formal statement is below, but first we consider a pair of examples:
\begin{align}
        & \vcenter{\hbox{ \onshellgraph{ {(-0.7,0)/2,(-0.7,1)/3,(-0.7,2)/4,(2.7,1)/1,(4.4,1)/5} }{ {(0,0)/W2,(0,1)/W3,(0,2)/W4,(2,1)/W1,(3.6,1)/W5} }{ {(1,0)/B1,(1,1)/B2,(1,2)/B3} }{ {{W2,B1},{W3,B1},{W1,B1},{W2,B2},{W4,B2},{W1,B2},{W3,B3},{W1,B3},{W4,B3},{1,W1},{2,W2},{3,W3},{4,W4},{5,W5}} }{ 1 } }} \hspace{0.5cm} \rightarrow \hspace{0.5cm} \vcenter{\hbox{ \onshellgraph{ {(0,0)/2,(0,1)/3,(0,2)/4,(2,1)/1,(3.5,1)/5} }{ {(0,0)/W2,(0,1)/W3,(0,2)/W4,(2,1)/W1,(3.5,1)/W5} }{ {(1,0)/B1,(1,1)/B2,(1,2)/B3} }{ {{W2,B1},{W3,B1},{W1,B1},{W2,B2},{W4,B2},{W1,B2},{W3,B3},{W1,B3},{W4,B3}} }{ 1 } }} \\
        & \vcenter{\hbox{ \onshellgraph{ {(-0.7,0)/2,(-0.7,1)/3,(-0.7,2)/4,(2,0.3)/1,(4.4,0.3)/6,(4.4,1.7)/5} }{ {(0,0)/W2,(0,1)/W3,(0,2)/W4,(2,1)/W1,(3.7,0.3)/W6,(3.7,1.7)/W5} }{ {(1,0)/B1,(1,1)/B2,(1,2)/B3,(3,1)/B4} }{ {{W2,B1},{W3,B1},{W1,B1},{W2,B2},{W4,B2},{W1,B2},{W3,B3},{W1,B3},{W4,B3},{W1,B4},{W5,B4},{W6,B4},{1,W1},{2,W2},{3,W3},{4,W4},{5,W5},{6,W6}} }{ 1 } }} \hspace{0.5cm} \rightarrow \hspace{0.5cm}
        \vcenter{\hbox{ \onshellgraph{ {(0,0)/2,(0,1)/3,(0,2)/4,(2,1)/1,(3.7,0.3)/6,(3.7,1.7)/5} }{ {(0,0)/W2,(0,1)/W3,(0,2)/W4,(2,1)/W1,(3.7,0.3)/W6,(3.7,1.7)/W5} }{ {(1,0)/B1,(1,1)/B2,(1,2)/B3,(3,1)/B4} }{ {{W2,B1},{W3,B1},{W1,B1},{W2,B2},{W4,B2},{W1,B2},{W3,B3},{W1,B3},{W4,B3},{W1,B4},{W5,B4},{W6,B4}} }{ 1 } }}
\end{align}
The top diagram is disconnected, with the external edge $5$ attached to a white vertex that has no other edges. The bottom diagram amends this by introducing an extra triplet and external edge. The canonical function $f_T$ vanishes for both as a consequence of there being the same group of three triplets that contain only four indices.

\begin{lemma}[Non-vanishing condition]
    The polynomial $  f_T$ is not the zero polynomial if and only if any subset of $d$ triplets contains at least $d+2$ indices. Formally, for any $S \subset \{1,\cdots,n-2\}$,
    \begin{equation}\label{eqn:non-vanishing}
        \bigg|\bigcup_{i \in S} \tau_i\bigg| \geq |S|+2.
    \end{equation}
\end{lemma}
\begin{proof}
    Consider a subset $S \subset \{1,\cdots,n-2\}$ of rows in the matrix $C^\perp(\vec \alpha^*)$.
    In order to have a non-vanishing $  f_T$, the number of columns with non-zero entries in at least one of these rows must be at least $|S|+2$. Otherwise, one could choose $\{a,b\}$ so that in $M_{ab}$, rows $S$ have rank at most $|S|-1$ and thus the determinant of $M_{ab}$ is zero. Rows correspond to triplets, and non-zero entries correspond to indices in the triplets, so this shows $f_T$ is nonzero only if \eqref{eqn:non-vanishing} holds.

    On the other hand, if \eqref{eqn:non-vanishing} holds, then we can apply the factorization algorithm below. The algorithm terminates with the collection of irreducible factors of the numerator of $  f_T$, which are either individual Pl\"uckers or the irreducible ``hypertree divisors" of Castravet--Tevelev \cite{Castravet:2013}. In particular, all the factors are nonzero, so $  f_T$ is also nonzero.
\end{proof}

We now turn to factorization. We call a subset $R$ of $T$ a \emph{valid subset} if $R$ consists of $d<n-2$ triplets and uses exactly $d+2$ indices. If $R$ is valid, then $R$ again corresponds to an on-shell diagram on a smaller number of points. As illustrated below, if $T$ has a valid subset $R$, then $f_T$ factors.

\[
    \vcenter{\hbox{
        \begin{tikzpicture}
            \node[ellipse,draw,fill=gray!50!white,minimum height=5em,minimum width=3em](R) at (0,0) {$R$};
    
            \node[ellipse,draw,fill=gray!50!white,minimum height=5em,minimum width=3em](T) at (2.5,0) {$T \setminus R$};
            
            \node[circle,draw,inner sep=0pt](a1) at (1,1) {$a_1$};
            \draw[line width=1pt] (R) -- (a1);
            \draw[line width=1pt] (T) -- (a1);
    
            \node[circle,draw,inner sep=0pt](a2) at (1,0.4) {$a_2$};
            \draw[line width=1pt] (R) -- (a2);
            \draw[line width=1pt] (T) -- (a2);
    
            \node() at (1,-0.2) {$\vdots$};
    
            \node[circle,draw,inner sep=0pt](an) at (1,-1) {$a_r$};
            \draw[line width=1pt] (R) -- (an);
            \draw[line width=1pt] (T) -- (an);
        \end{tikzpicture}
    }}
    =
    \left(\vcenter{\hbox{
        \begin{tikzpicture}
            \node[ellipse,draw,fill=gray!50!white,minimum height=5em,minimum width=3em](R) at (0,0) {$R$};
            
            \node[circle,draw,inner sep=0pt](a1) at (1,1) {$a_1$};
            \draw[line width=1pt] (R) -- (a1);
    
            \node[circle,draw,inner sep=0pt](a2) at (1,0.4) {$a_2$};
            \draw[line width=1pt] (R) -- (a2);
    
            \node() at (1,-0.2) {$\vdots$};
    
            \node[circle,draw,inner sep=0pt](an) at (1,-1) {$a_r$};
            \draw[line width=1pt] (R) -- (an);
        \end{tikzpicture}
    }}\right)
    \times
    \left(\vcenter{\hbox{
        \begin{tikzpicture}
            \node[ellipse,draw,fill=gray!50!white,minimum height=5em,minimum width=3em](T) at (2.5,0) {$T \setminus R$};
            
            \node[circle,draw,inner sep=0pt](a1) at (1,1) {$a_1$};
            \draw[line width=1pt] (T) -- (a1);
    
            \node[circle,draw,inner sep=0pt](a2) at (1,0.4) {$a_2$};
            \draw[line width=1pt] (T) -- (a2);
    
            \node() at (1,-0.2) {$\vdots$};
    
            \node[circle,draw,inner sep=0pt](an) at (1,-1) {$a_r$};
            \draw[line width=1pt] (T) -- (an);

            \node[circle,fill=black](b1) at (0.4,0.7) {};
            \draw[line width=1pt] (a1) -- (b1);
            \draw[line width=1pt] (a2) -- (b1);
            \draw[line width=1pt] (b1) -- (0.8,0);

            \node() at (0.4,0) {$\vdots$};

            \node[circle,fill=black](b2) at (0.4,-0.8) {};
            \draw[line width=1pt] (a1) .. controls (-0.2,1.2) and (-0.2,0) .. (b2);
            \draw[line width=1pt] (an) -- (b2);
            \draw[line width=1pt] (b2) -- (0.8,-0.6);
        \end{tikzpicture}
    }}\right)
    \times
    \prod_{i=1}^r (a_i a_{i+1})
\]

Stated more formally, the factorization is as follows.

\begin{lemma}[Factorization of the form]\label{lem:factorization}
    Let $T$ be a set of $n-2$ triplets on $[n]$. Suppose $R$ is a valid subset of $T$. Then
    \begin{equation}\label{eqn:factorization-body}
          f_T =   f_{R} \times   f_{(T \setminus R) \cup P} \times \prod_{i=1}^{r} ( a_i a_{i+1}),
    \end{equation}
    where $A = \{a_1,\cdots,a_{r}\}$ consists of the indices appearing both in $R$ and $T$ (we set $a_{r+1}=a_1$) and 
    $P = \{(a_1,a_i,a_{i+1}) | i \in \{2,\cdots,r-1\}\}$ is a set of triplets defining a triangulation of a polygon with vertices $A$.
\end{lemma}

\begin{proof}
    The main idea of the proof lies within the determinantal formula. By choosing the removed columns to be within $R$, one finds a block-triangular matrix, factorizing the determinant. In order to get back a set of triplets corresponding to a diagram, one of these factors can be rewritten to include the auxiliary set of triplets $P$, whose contribution is corrected by the extra linear factors. Details of the proof can be found in Appendix \ref{app:factorization}.
\end{proof}


We make a few remarks on the factorization formula. If all the indices of $R$ are contained in $T \setminus R$, and $R$ is planar\footnote{We include the trivial case where $R$ is a single triplet.}, then the factorization is trivial. That is, one may choose $P = R$, so the first and last factors are inverses and the middle factor is again $f_T$.

There are two distinct cases when the factorization is nontrivial. Of course we may assume that the form is nonzero. The first case of nontrivial factorization is when $R$ consists of a single triplet $(ijk)$ and one index, say $k$, does not appear in any other triplet of $T$. We have that $i,j$ will both appear in $T$ (otherwise the form would be zero), so $A= \{i,j\}$ and $P$ is empty. We obtain the factorization 
\begin{equation}
    \begin{aligned}
          f_T &=   f_R \times   f_{T \setminus R} \times ( ij) ^2 \\
        &=   f_{T \setminus R} \times \frac{(ij)}{( ik) ( kj)}.
    \end{aligned}
\end{equation}
So the triplet $(ijk)$ contributes a monomial in Pl\"ucker coordinates to the form. The second case of nontrivial factorization is when $R$ is non-planar. In this case, the numerator of $f_R$ is a high-degree, possibly quite complicated, polynomial in Pl\"uckers.


Applying \eqref{eqn:factorization-body} repeatedly in the two nontrivial cases, one eventually arrives at a product whose factors are Pl\"ucker coordinates, their inverses, and forms for $T$ where $T$ has no lone indices and no valid subsets which are non-planar. We call the latter sets of triplets \emph{irreducible}, and they are in direct correspondence with the \emph{irreducible hypertrees} studied by \cite{Castravet:2013}. For such irreducible $T$, when $f_T$ is written in lowest terms, the numerator is the square of an irreducible polynomial. This irreducible polynomial vanishes on a \emph{hypertree divisor} of $\mathcal{M}_{0,n}$. We note that the factorization of $f_T$ is not necessarily unique, since there may be some cancellation between distinct factors.


\subsection{Sphere moves}

As is already clear from the planar case, there may be many different triplets which give rise to the same nonzero form. It is well-known that all planar diagrams with the same form are related by square moves:
\begin{equation}
    \begin{aligned}
        \vcenter{\hbox{ \onshellgraph{ { (0.3,0.3)/1, (2,1)/4, (2.7,2.7)/3, (1,2)/2} }{ {(1,2)/W2, (2,1)/W4,(0.3,0.3)/W1,(2.7,2.7)/W3} }{ {(1,1)/B1, (2,2)/B2} }{ {{W1,B1},{W2,B1},{W4,B1},{W2,B2},{W3,B2},{W4,B2}} }{1} } } &\leftrightarrow \vcenter{\hbox{ \onshellgraph{ { (1,1)/1, (2.7,0.3)/4, (2,2)/3, (0.3,2.7)/2} }{ {(0.3,2.7)/W2, (2.7,0.3)/W4,(1,1)/W1,(2,2)/W3} }{ {(1,2)/B1, (2,1)/B2} }{ {{W1,B1},{W2,B1},{W3,B1},{W1,B2},{W3,B2},{W4,B2}} }{1} } } \\
        f_{(124),(234)} &= f_{(123),(134)}
    \end{aligned}
\end{equation}

In terms of triplets, $T$ has a square if $(ijk), (ikl) \in T$ for some $i,j,k,l$. The square move replaces these two triplets with $(jil),(jkl)$. For non-planar diagrams, some diagrams with the same form are \emph{not} related by square moves. Castravet--Tevelev show in \cite{Castravet:2013} that the more general \emph{sphere moves} also preserve the numerator of $f_T$; using very different techniques, \cite{Cachazo:2019} showed that the whole form is preserved under sphere moves. We will show that in fact the sphere moves are the only moves on triplets preserving the form.

To describe sphere moves most concisely, we switch from thinking of a triplet $(ijk)$ as a ``tripod" (one black vertex attached to the three white vertices $i,j,k$) to thinking of $(ijk)$ as a triangle with vertices labeled $i,j,k$. A sphere move comes from triangulating the sphere in any way, labeling the vertices with some subset of $[n]$, and then coloring half of the triangles black and the other half white so that every vertex is in both a black and a white triangle. The sphere move replaces the triplets $B$ in $T$ corresponding to the black triangles with the triplets $W$ corresponding to the white triangles.

\begin{figure}
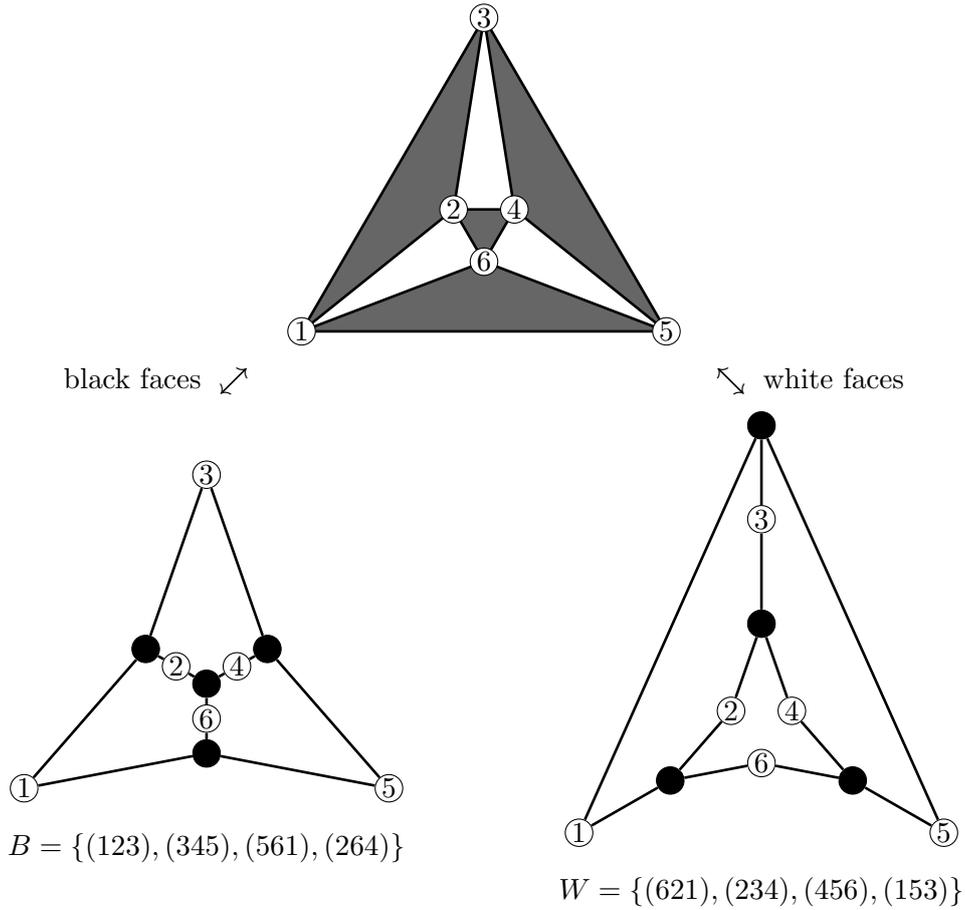

    \centering
    \[
    \begin{gathered}
    \vcenter{\hbox{ \polygongraph{{(0,0)/1,(3.3333,2.6944)/2,(4,6.9282)/3,(4.6667,2.6944)/4,(8,0)/5,(4,1.5396)/6}}{{{1,2,3},{3,4,5},{5,6,1},{2,6,4}}}{0.6} }} \\
    \text{black faces } \swarrow\!\!\!\!\!\!\nearrow \hspace{+16em} \nwarrow\!\!\!\!\!\!\searrow \text{ white faces} \\
    \begin{gathered}
    \hbox{ \onshellgraph{{(0,0)/1,(3.3333,2.6944)/2,(4,6.9282)/3,(4.6667,2.6944)/4,(8,0)/5,(4,1.5396)/6}}{{(0,0)/W1,(3.3333,2.6944)/W2,(4,6.9282)/W3,(4.6667,2.6944)/W4,(8,0)/W5,(4,1.5396)/W6}}{{(2.6667,3.0792)/B1,(5.3333,3.0792)/B2,(4,0.7698)/B3,(4,2.3094)/B4}}{{{W1,B1},{W2,B1},{W3,B1},{W3,B2},{W4,B2},{W5,B2},{W5,B3},{W6,B3},{W1,B3},{W2,B4},{W6,B4},{W4,B4}}}{0.6} } \\ B = \{(123),(345), (561),(264)\}
    \end{gathered}
    \hspace{+4em}
    \begin{gathered}
    \hbox{ \onshellgraph{{(0,0)/1,(3.3333,2.6944)/2,(4,6.9282)/3,(4.6667,2.6944)/4,(8,0)/5,(4,1.5396)/6}}{{(0,0)/W1,(3.3333,2.6944)/W2,(4,6.9282)/W3,(4.6667,2.6944)/W4,(8,0)/W5,(4,1.5396)/W6}}{{(2,1.1547)/B1,(4,4.6188)/B2,(6,1.1547)/B3,(4,9)/B4}}{{{W1,B1},{W2,B1},{W6,B1},{W2,B2},{W3,B2},{W4,B2},{W4,B3},{W5,B3},{W6,B3},{W1,B4},{W3,B4},{W5,B4}}}{0.6} } \\
    W = \{(621),(234), (456),(153)\}
    \end{gathered}
    \end{gathered}
    \]
    \captionsetup{width=.9\linewidth}
    \vspace{-1em}
    \caption{The 6-point sphere move. For convenience, in the bottom row we omit the external legs and label the white vertex they were attached to instead.}\label{fig:6point-triangulation}
\end{figure}

For example, consider the triangulation of the sphere into 8 triangles shown above in Figure \ref{fig:6point-triangulation} (since we draw the sphere on the plane, one triangle is the unbounded region). The black triangles correspond to the triplets $B=\{(123),(345),(561),(264)\}$, and the white to $W=\{(621)$, $(234)$, $(456),(153)\}$. The corresponding sphere move can be applied to any $T$ containing the triplets $B$ and replaces $B$ with $W$ and does not change the associated form. Notice that the square move is a particularly special kind of sphere move. In this case, the triangulation consists of 4 triangles, and looks like a tetrahedron with two faces colored black and two colored white.

Note also that every sphere move changes a valid subset of $T$ into a different valid subset.  This is easy to see by a simple counting argument using the Euler characteristic. With $|B| = |W| = d-2$, the triangulation will have $2d-4$ faces and $3d-6$ edges, so requiring that the Euler characteristic is $\chi = 2$ reveals that the number of vertices is $d$. So both $B$ and $W$ are valid. One could define a similar move for any triangulation of any surface, with half the triangles colored white and half black. However, if this surface is not a sphere, such a move would involve triplets which fail the nonvanishing condition. Indeed, for a surface which is not a sphere, we have $\chi=2-2g\leq 0$. We also have $\chi=2|B| - 3|B| + V$ so the number of vertices is at most $|B|$. Thus any set of triples $T$ containing $B$ has vanishing form.
\subsection{Sphere moves are the only moves}

To show that triplets have the same form if and only if they are related by sphere moves, it is useful to introduce another combinatorial object, inspired by the derivation of the sphere moves in \cite{Cachazo:2019}. For $T$ a collection of triplets, the corresponding set of \emph{doublets} is
\begin{equation}
    D(T) = \{ij \, | \, \text{$i$ and $j$ appear together in an odd number of triplets}\}.
\end{equation}
The idea is that the doublet sets are themselves invariant under sphere moves, so they are a better labeling set for the form than the set of triplets. With these definitions, we state the following theorem.
\begin{theorem}\label{thm:moves}
    Let $T$ and $T'$ be two sets of triplets with nonvanishing forms (i.e. $T, T'$ satisfy \eqref{eqn:non-vanishing}). The following are equivalent:
    \vspace{-0.6em}
    \begin{enumerate}[label=(\arabic*)]
        \item There is a sequence of sphere moves relating $T$ and $T'$.
        \vspace{-0.6em}
        \item The corresponding forms are equal: $f_T = f_{T'}$.
        \vspace{-0.6em}
        \item The corresponding doublets are the same: $D(T) = D(T')$.
    \end{enumerate}
\end{theorem}
We give a sketch of the proof here; see the appendix for additional details. (1) $\implies$ (2) follows from \cite{Castravet:2013,Cachazo:2019}.

For (2) $\implies$ (3), we claim that $D(T)$ is also equal to 
$$\{ij| (ij) \text{ appears with an odd exponent in }f_T\}.$$ 
For example, taking the set of triplets $T = \{(123),(134),(135)\}$, one has
\begin{equation}
    f_T = \frac{(13)}{(12)(32)(14)(34)(15)(35)}, \quad D(T) = \{12,23,14,34,15,35,13\}.
\end{equation}
We prove this statement in general by examining residues and applying factorization in Lemmas \ref{lem:lonepair}, \ref{lem:poleirred}, and \ref{lem:doublets-from-form}.




To prove (3) $\implies$ (1), we give an algorithm in Lemma \ref{lem:doublet-to-sphere} to build sets $R \subset T, R' \subset T'$ so that $R \cup R'$ triangulate the sphere.
As an example, we illustrate the algorithm for the 6-point non-planar diagram. We have two triplets $B = \{(123),(345), (561),(264)\}$ and $W = \{(621),(234), (456),(153)\}$ that give the same doublets, namely 
\begin{align}
      D(B)=D(W) = \{12,23,13,34,24,35,45,56,15,26,24,46\} \nonumber
\end{align} 
Now we construct a sphere move from $B$ to $W$ by running the algorithm in Lemma \ref{lem:doublet-to-sphere}. We start by picking a random triplet $(123)$ from set $B$. So now $R=\{(123)\}$,$R'=\{\}$, and $D(R) = \{12,23,13\}$. Treating $R$ as a collection of black triangles, so far we have
\begin{equation}
    \hbox{\polygongraph{{(0,0)/1,(3.3333,2.6944)/2,(4,6.9282)/3}}{{{1,2,3}}}{0.6}}
\end{equation}


\par
Moving onto step 2), $W \setminus R'=\{(621),(234),(456),(153)\}$. All triplets except $(456)$ have a pair of indices $\{i,j\}$ in the doublet set $D(R) \setminus D(R') = \{12,23,13\}$. These triplets are added to $R'$, so $R'=\{(621),(234),(153)\}$. Interpreting $R'$ as a collection of white triangles, we have attached three white triangles to the lone black triangle, (the triangle corresponding to $(135)$ is the unbounded region).

\begin{equation}
\begin{tikzpicture}[
  scale=0.55,
  line cap=round, line join=round,
  v/.style={circle,draw,fill=white,inner sep=0.7pt,minimum size=1mm,font=\large}
]

\coordinate (1) at (0,0);
\coordinate (5) at (8,0);
\coordinate (3) at (4,6.9282);
\coordinate (6) at (4,1.5396);
\coordinate (2) at (3.3333,2.6944);
\coordinate (4) at (4.6667,2.6944);
\fill[black!60] (1)--(2)--(3)--cycle;
\draw[very thick] (1)--(5)--(3)--(1);
\draw[very thick] (1)--(2)--(6)--(1);
\draw[very thick] (3)--(2)--(4)--(3);
\node[v] at (0,0) {1};
\node[v] at (3.3333,2.6944) {2};
\node[v] at (4,6.9282) {3};
\node[v] at (4.6667,2.6944) {4};
\node[v] at (8,0) {5};
\node[v] at (4,1.5396) {6};
\end{tikzpicture}
\end{equation}

\par
There aren't any triplets in $W \setminus R' = \{(456)\}$ containing a pair of indices in $D(R') \setminus D(R)$ so we skip step 3). All the triplets in $B \setminus R=\{(345),(561),(264)\}$ contain a pair of indices in $D(R \cup R')=\{26,34,24,15,53\}$ so we now have $R=\{(123),(345),(561),(264)\}$.
\par
Cycling around, we skip step 1) since $W \setminus R=\{\}$. Following step 2), we add $(456)$ to $R'$, which gives us a complete triangulation as show below \begin{equation}
    \hbox{ \polygongraph{{(0,0)/1,(3.3333,2.6944)/2,(4,6.9282)/3,(4.6667,2.6944)/4,(8,0)/5,(4,1.5396)/6}}{{{1,2,3},{3,4,5},{5,6,1},{2,6,4}}}{0.6} }
\end{equation}
The algorithm terminates here and we can see from the graph that $R$ and $R'$ form a complete bicolored triangulation of a sphere.

\begin{tcolorbox}[colback=white!95!black]
\begin{center}
\vspace{-0.1cm}
The upshot of Theorem \ref{thm:moves} is that two general non-vanishing MHV on-shell diagrams have the same canonical functions $f_G$ if and only if they are related by sphere moves and merge/expand moves. This is a generalization of the analogous statement for the planar case, where two diagrams have the same canonical functions if and only if they are related by square moves and merge/expand moves.
\vspace{-0.1cm}
\end{center}
\end{tcolorbox} 

\section{Oriented Regions and Geometries for Non-planar Diagrams}\label{sec:regions-and-geometries}

Our goal in this section is to associate regions of the Grassmannian with MHV on-shell diagrams $G$ (equivalently triplets $T$), just as $G_+(2,n)$ is associated with planar on-shell diagrams. These regions will have canonical form $\omega_G= \omega_T$, in the sense of \cite{Arkani_Hamed_2017}.

\subsection{Parke-Taylor Factors as Canonical Functions of Oriented Regions}

To prepare for a geometric interpretation of the decomposition in \eqref{eqn:PTdecomp}, we identify geometries whose canonical functions are the Parke-Taylor factors. Recall that $G_+(2,n)$ has canonical function ${\rm PT}(12 \dots n)$. We also recall from Section 2 that the interior of $G_+(2,n)$ has a nice interpretation in terms of $n$ ordered points on $\mathbb{P}^1$. If we reorder these points via a permutation $\sigma$, or equivalently reorder the columns of $C$ using $\sigma$, we obtain another space isomorphic to $G_+(2,n)$, which is described by all Pl\"ucker coordinates having some fixed sign and whose canonical function is ${\rm PT}(\sigma).$ We also have the freedom of exchanging any column of $C$ for its negative; this gives a space with the same canonical function. We call these spaces \emph{oriented regions}. 

We can specify an oriented region by specifying a way to permute and negate the columns of $C \in G_+(2,n)$. We use the notation
\begin{equation}
    R(\varepsilon_{i_1} i_1,\varepsilon_{i_2} i_2, \cdots, \varepsilon_{i_n} i_n) = \{ C \in G(2,n) : \varepsilon_{i_a} \varepsilon_{i_b} (i_a i_b) \geq 0 ~ \forall ~ a<b\},
\end{equation}
where $i_1 \dots i_n$ is a permutation of $1, \dots, n$ and $\varepsilon_i\in \{\pm\}$ is a choice of sign.



This notation is redundant, allowing for cyclic shifts, overall sign-flips, and reversals up to a sign,
\begin{align}
    R(\varepsilon_{i_1} i_1,\varepsilon_{i_2} i_2, \cdots, \varepsilon_{i_n} i_n) &= R(\varepsilon_{i_2} i_2 , \cdots, \varepsilon_{i_n} i_n , -\varepsilon_{i_1} i_1), \\
    R(\varepsilon_{i_1} i_1,\varepsilon_{i_2} i_2, \cdots, \varepsilon_{i_n} i_n) &= R(-\varepsilon_{i_1} i_1,-\varepsilon_{i_2} i_2, \cdots, -\varepsilon_{i_n} i_n), \\
    R(\varepsilon_{i_1} i_1,\varepsilon_{i_2} i_2, \cdots, \varepsilon_{i_n} i_n) &= R(\varepsilon_{i_1} i_1,-\varepsilon_{i_n} i_n, \cdots, -\varepsilon_{i_2} i_2).
\end{align}
We fix the redundancy by making the following convention choice: we fix $i_1 = 1$, and $\varepsilon_{1} = \varepsilon_2 = +$. This leaves us with $2^{n-2} \cdot (n-1)!$ distinct oriented regions that tile the real Grassmannian. As a shorthand for oriented regions denoted with this convention, we will write $R(1,\varepsilon \sigma)$, where $\varepsilon$ is understood to denotes the signs $\varepsilon_3,\cdots,\varepsilon_n$ and $\sigma$ denotes the permutation of $2, \dots, n$.

Noting that $R(1,2,\cdots,n) = G_+(2,n)$ has canonical function ${\rm PT}(1,2,\cdots,n)$, it is easy to see that $R(1,\varepsilon \sigma)$ has canonical function ${\rm PT}(1,\sigma)$. Consequently, each Parke-Taylor factor is the canonical function of $2^{n-2}$ distinct oriented regions. It is important to note that in both cases, the canonical function is unique only up to an overall sign; however, this subtlety will not be relevant to our further discussion as the decompositions from \eqref{eqn:PTdecomp} contain only Parke-Taylor factors with positive sign.

As each oriented region is simply a relabeling and reflection of the positive Grassmannian, its boundary structure is as discussed in Section \ref{sec:onshell}. We call two oriented regions \emph{adjacent} if they share a codimension 1 boundary. This occurs when the signed permutations labeling the regions differ by a simple transposition:
\begin{align}
    R(1,\cdots,\varepsilon_i i, \varepsilon_j j,\cdots) \Big|_{(ij) = 0} &= R(1,\cdots,\varepsilon_j j, \varepsilon_i i,\cdots) \Big|_{(ij) = 0}, \\
    R(1,\cdots,\varepsilon_i i) \Big|_{(1i) = 0} &= R(1,-\varepsilon_i i,\cdots) \Big|_{(1i) = 0},
\end{align}
where $|_{f=0}$ denotes the intersection with $\{f=0\}$. The second case is the same as the first, but looks different because of our conventions on how to notate oriented regions. This geometric statement about adjacent regions is in direct correspondence with a statement about residues of adjacent Parke-Taylor factors,
\begin{equation}
    \underset{(ij) =0}{\rm Res} {\rm PT}(1\cdots i j \cdots) = -\underset{(ij) =0}{\rm Res} {\rm PT}(1\cdots ji\cdots).
\end{equation}

\subsection{Unions of Oriented Regions}

We now use oriented regions to obtain spaces with canonical function $f_T$. Since for each $f_T$ there are a multitude of Parke-Taylor decompositions (one for each orientation of the triplets, and then more if one may apply sphere moves), and each Parke-Taylor factor is the canonical function of $2^{n-2}$ different oriented regions, there is a very large set of geometries with canonical function $f_T$, most of which have undesirable properties such as unnecessary spurious boundaries. We will restrict our attention as follows.
\begin{itemize}
    \item We consider only those geometries that come directly from the formula \eqref{PTexp} for $f_T$ in terms of triplets 
    \begin{equation}\label{eqn:PR-decomposition}
        f_T = \sum_{\sigma \in S^{(T)}_n} {\rm PT}(\sigma) \rightsquigarrow R_T = \bigcup_{\sigma \in S_n^{(T)}} R(\varepsilon_\sigma \sigma).
    \end{equation}
    There are many other formulas for $f_T$. For example, one may add a sum of PT factors which is equal to 0. We do not consider geometries from such formulas.
    \item We consider geometries only up to re-signing all of the appearances of a given index. 
    \item We choose signs for oriented regions so that the number of adjacencies is maximized. 
\end{itemize}
To convince the reader that these restrictions are reasonable, we will quickly illustrate some bad features that occur already at 4 points when we do \emph{not} place these restrictions. Then, with the restrictions in mind, we analyze all possible geometries for the simplest non-planar on-shell diagram, which serves as a starting point for understanding more general unions.

\subsubsection*{4-point square}

The unique, up to relabeling, 4-point on-shell diagram has canonical function ${\rm PT}(1234)$, which can be thought of as coming from triplets in the sense of (\ref{PTexp})
\[
T = \{(123),(134)\} \rightarrow S_4^{(123),(134)} = \{1234\}.
\]
Consequently, one positive geometry with this form is simply $R(1234) = G_+(2,4)$ as one expects from this being a planar diagram. Since the canonical function is not affected by reflections, the positive geometries $R(1,2,-3,4)$ or $R(1,2,-3,-4)$ also have the desired form, and are isomorphic. This motivates our constraint to consider geometries only up to re-signing all appearances of a single index.

Applying a Kleiss-Kuijf relation, we can find a more complicated geometry. As one has ${\rm PT}(1234) = -({\rm PT}(1243) + {\rm PT}(1423))$, the positive geometry $R(1243) \cup R(1423)$ also has canonical function ${\rm PT}(1234)$. This geometry has an undesirable property, called a \emph{null boundary}: there is a codimension-1 boundary along $(13) = 0$ but $(13)$ is not a pole of the original form. Nonetheless, we still consider this geometry as it does come from an orientation of the triplets in the sense of (\ref{PTexp}) as
\[
T' = \{(123),(143)\} \rightarrow S_4^{(123),(143)} = \{1243,1423\}.
\]
If we flip some signs, say to obtain $R(1,2,4,3) \cup R(1,-4,2,3)$, we do not change the form. As this is not an overall reflection, the topology of the resulting geometry has changed. In particular, we now have a second null boundary when $(24) = 0$ as the two regions are no longer glued on the corresponding facet. This motivates our constraint to consider only those geometries where signs are chosen to make as many regions adjacent as possible.

Furthermore, one could in principle include extra oriented regions for which the canonical function adds to zero. For example, ${\rm PT}(1432) + {\rm PT}(1342) + {\rm PT}(1234) = 0$, so we can take any of the aforementioned geometries and then take the union with $R(1432) \cup R(1342) \cup R(1324)$, which introduces many null boundaries. This motivates our constraint to consider only those geometries that directly come from triplets.

\subsubsection*{5-point non-planar}

The unique 5-point nonplanar on-shell diagram (up to relabeling) is
\begin{equation}\label{eqn:5point-non-planar}
    \vcenter{\hbox{\onshellgraph{ {(1,0)/1,(-1,1)/2,(1,2)/3,(1,1)/4,(3,1)/5} }{ {(1,0)/W1,(-1,1)/W2,(1,2)/W3,(1,1)/W4,(3,1)/W5} }{ {(-0.5,1)/B1,(0.5,1)/B2,(2,1)/B3} }{ {{B1,W1},{B1,W2},{B1,W3},{B2,W1},{B2,W3},{B2,W4},{B3,W1},{B3,W3},{B3,W5}} }{1}}} = \vcenter{\hbox{ \onshellgraph{ {(1,0)/1,(0,1)/2,(1,2)/3,(1,1)/4,(3,1)/5} }{ {(1,0)/W1,(0,1)/W2,(1,2)/W3,(1,1)/W4,(3,1)/W5} }{ {(0.5,0.5)/B1,(0.5,1.5)/B2,(2,1)/B3} }{ {{B1,W1},{B1,W2},{B1,W4},{B2,W3},{B2,W2},{B2,W4},{B3,W1},{B3,W3},{B3,W5}} }{1} }}
\end{equation}
with triplets $T = \{(123),(134),(135)\}$ or, applying a square move, $T' = \{(124),(234),(135)\}$. The canonical function is
\begin{equation}\label{eqn:5point-form}
    f_T = \frac{(13)}{(12)(23)(14)(34)(15)(35)},
\end{equation}
and (up to a sign) has many PT expressions induced by different orientations of $T$ or $T'$,
\begin{equation}
    \begin{aligned}
        f_T &= {\rm PT}(12345) + {\rm PT}(12354) \\
        &= {\rm PT}(13245)+{\rm PT}(13254)+ {\rm PT}(13425)+{\rm PT}(13452) + {\rm PT}(13524) + {\rm PT}(13542) \\
        &= {\rm PT}(13425)+{\rm PT}(13452) +{\rm PT}(13542) + {\rm PT}(14235).
    \end{aligned}
\end{equation}
The first line comes from both $S_5^{(123),(134),(135)}$ and $S_5^{(124),(234),(135)}$, the second comes from $S_5^{(132),(134),(135)}$, and the third from $S_5^{(142),(234),(135)}$. This list exhausts the possible decompositions up to relabeling. The first decomposition can come from either choice of triplets for a particular choice of orientations; this is caused by $(123),(134)$ and $(124),(234)$ making the same requirement that the indices $(1234)$ to appear in that cyclic order. Corresponding to these decompositions, we have several geometries that have the correct canonical function:
\begin{align}
        R_{(123),(134),(135)} &= R_{(124),(234),(135)} = R(12345) \cup R(12354),\label{eqn:5point_R1} \\
        R_{(132),(134),(135)} &=  R(13245) \cup R(13254) \cup R(13425) \cup R(13452)
       \cup R(13524) \cup R(13542),\label{eqn:5point_R2} \\
        R_{(142),(234),(135)} &= R(13425) \cup R(13452) \cup  R(13542) \cup R(14235).\label{eqn:5point_R3}
\end{align}
The first two geometries (\ref{eqn:5point_R1},~\ref{eqn:5point_R2}) are strongly connected, i.e. the oriented regions are connected by their facets; we will find in Section \ref{sec:hierarchy} that this is because we chose the orientation of triplets coming from reading around each black vertex in \eqref{eqn:5ptnon-planar} clockwise. The last geometry (\ref{eqn:5point_R3}) is not strongly connected, as $R(14235)$ does not share any facets with the remaining regions.

The latter two geometries (\ref{eqn:5point_R2},~\ref{eqn:5point_R3}) both have a null codimension-1 boundary at $(13) = 0$, with the oriented regions having it as a facet despite it not being a pole of the canonical function. The first geometry (\ref{eqn:5point_R1}) appears not to have such a problem, but it turns out that a closer look reveals null boundaries at higher codimension. Taking the residue as $2 = \alpha 1$ (or $2 = \alpha 3$), one finds a residue (in the sense of (\ref{eqn:residue_example}))
\begin{equation}\label{eqn:first_residue}
    \underset{2 = \alpha 1}{\rm Res} f_T = \frac{1}{\alpha (14)(34)(15)(35)},
\end{equation}
which has no pole at $(13) = 0$. However, finding the boundary of $R_{(123),(134),(135)}$ at $2 = \alpha 1$, one can write it as the closure of an affinely parametrized set,
\begin{equation}\label{eqn:spurious_boundary}
    R_{(123),(134),(135)} \bigg|_{2 = \alpha 1} = \overline{\left\{ \begin{pmatrix}  1 & \alpha & x & 0 & y \\ 0 & 0 & z & 1 & w \end{pmatrix} \in G(2,5) : \alpha,x,z,w \geq 0, xw-yz \geq 0  \right\}},
\end{equation}
which does have a boundary at $(13) = 0$. Note that the variable $y$ intentionally does not have an associated inequality in (\ref{eqn:spurious_boundary}), as the sign of the minor $(45)$ is not fixed in $R_{(123),(134),(135)}$. Studying this boundary region one does find a further boundary at $(13) = 0$, reached by $z \rightarrow 0$ in the parametrization above. More specifically, we get
\begin{equation}\label{eqn:spurious_boundary2}
    R_{(123),(134),(135)} \bigg|_{2 = \alpha 1, 3=\beta 1} = \overline{\left\{ \begin{pmatrix}  1 & \alpha & \beta & 0 & y \\ 0 & 0 & 0 & 1 & w \end{pmatrix} \in G(2,5) : \alpha,\beta,w \geq 0 \right\}}
\end{equation}
Note that the parameter $y$ is completely free once we set $(123)=0$. Then the geometry becomes $(\alpha,\beta,w) \in \mathbb{R}^3_{\geq 0} \times y \in \mathbb{R}$. When we take the projective closure, the geometry associated with the $y$ factor is $\mathbb{P}^1(\mathbb{R})$. This is a circle topologically and hence has no boundary points, meaning the canonical function is 0. Thus we can see from the geometry that the boundary lying on $(13)=0$ is null and the canonical function has no pole there. 

\begin{tcolorbox}[colback=white!95!black,left=3pt,right=3pt,top=7pt,bottom=7pt]
\begin{center}
\vspace{-0.1cm}
  Even for the simplest non-planar on-shell diagram, it is \emph{impossible} to find a Grassmannian region $R_T$ with correct canonical form $f_T$ which is a positive geometry.
\vspace{-0.1cm}
\end{center}
\end{tcolorbox} 

We now turn to the topology of $R(12345) \cup R(12354)$. The regions $R(12345)$ and $R(12354)$ glue on the facet where $(45)=0$. They also glue on the (codimension-6) vertex where the first three columns vanish:
\begin{equation}
    R(12345) \bigg|_{\text{cols } 1,2,3 \text{ vanish}} = \begin{pmatrix}  0 & 0& 0 & 1 & 0 \\ 0 & 0 & 0 & 0 & 1 \end{pmatrix} \overset{GL(2)}{=} \begin{pmatrix}  0 & 0& 0 & 1 & 0 \\ 0 & 0 & 0 & 0 & -1 \end{pmatrix} = R(12354) \bigg|_{\text{cols } 1,2,3 \text{ vanish}}.
\end{equation}
As each oriented region is topologically a ball, gluing them together on a facet and a vertex results in something like a pinched torus. A caricature of this geometry is illustrated in Figure \ref{fig:5-point-pinched}. One must keep in mind that $R(12345)$ and $R(12354)$ are 6-dimensional, the intersection with $(45)=0$ is 5-dimensional, the intersection with $(13)=0$ is 4-dimensional, and the point where the first three columns vanish is 0-dimensional.

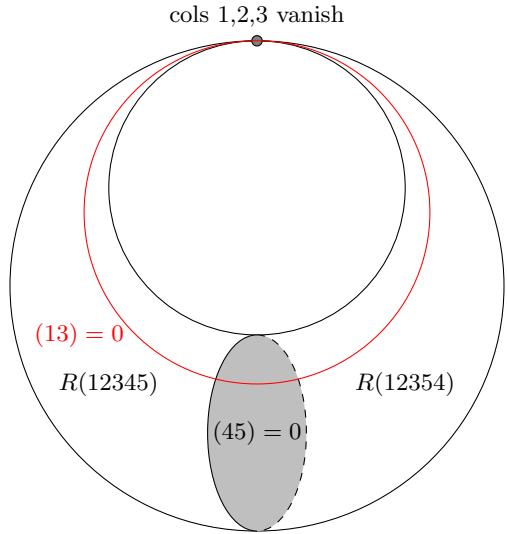
\begin{figure}
    \centering
    \begin{tikzpicture}[scale=0.65]
        \draw (0,0) circle (5);
        \draw (0,2) circle (3);
        \draw[fill=gray!50!white,draw=none] (0,-3) ellipse (1 and 2);
        \draw (0,-1) arc (90:270:1 and 2);
        \draw[dashed] (0,-1) arc (90:-90:1 and 2);
        \draw[fill=gray] (0,5) circle (3pt);
        \draw[red] (0,1.5) circle (3.5);

        \node() at (0,-3) {\footnotesize{$(45)=0$}};
        \node() at (3,-2) {\footnotesize{$R(12354)$}};
        \node() at (-3,-2) {\footnotesize{$R(12345)$}};
        \node() at (-3.6,-1) {\footnotesize{\color{red} $(13) = 0$}};
        \node() at (0,5.5) {\footnotesize{cols 1,2,3 vanish}};
    \end{tikzpicture}
    \caption{A caricature of $R(12345) \cup R(12354)$.\label{fig:5-point-pinched}}
\end{figure}

Another feature to consider is how PT-like residues are represented in the geometry. Taking the residue at $\alpha = 0$ of (\ref{eqn:first_residue}), one finds ${\rm PT}(1435)$. The label $2$ has vanished as we have sent the entire column to zero, so we might expect the corresponding geometry to have reduced to the oriented region $R(1435)$ embedded into $G(2,5)$ by inserting a zero-column. Unfortunately, this is not the case for any of the geometries\footnote{Technically, one can find a geometry by relabeling the indices that does give a suitable boundary that is a single oriented region. For example, $R(14325) \cup R(14352)$ has boundary $R(1435)$ as desired. However, this will always introduce an analogous problem for another index, as with the vanishing of column $4$ for this relabeling.}, as taking the boundary when columns $1$ and $2$ are parallel followed by the vanishing of column $2$ still leaves us with a union of oriented regions with the canonical function decomposed as ${\rm PT}(1345) + {\rm PT}(1354)$. Consequently, the boundaries of the geometries for non-planar diagrams do not generally keep the nice recursive structure that appeared for planar diagram geometries.

\subsection{Hierarchy of Diagrams and Corresponding Results}\label{sec:hierarchy}

In this section, we discuss the behavior of the geometries for arbitrary collections of triplets $T$, and so for arbitrary on-shell diagrams. As mentioned above, there are \emph{many} geometries for each collection of triplets, and some are substantially better behaved than others. We will focus on two kinds of well-behavedness, which require some definitions.

For a collection $T$ of triplets, the region $R_T$ has facets along hypersurfaces $(ij)=0$. If $f_T$ has a pole along a facet, we call the facet \emph{physical}; otherwise, we call the facet \emph{null}. Note that the residue of $f_T$ along a null facet is zero. More generally, a boundary of $R_T$ is a \emph{null boundary} if, taking iterated residues of $f_T$, one obtains the zero function on the boundary. Positive geometries have no null boundaries; if a geometry with null boundaries otherwise satisfies the definition of positive geometries, it is a \emph{pseudo-positive geometry}. We will look for geometries which are as close to positive geometries as possible, meaning that null boundaries occur in codimension at least 2. We call a pseudo-positive geometry a \emph{good-facet geometry} if all facets are physical. Positive geometries are good-facet geometries.

We would like the decomposition of $R_T$ into oriented regions to be as analogous to a triangulation of a polytope as possible. We call $R_T$ \emph{strongly connected} if one can walk in $R_T$ from any region $R(1, \epsilon \sigma) \subset R_T$ to any other $R(1, \epsilon' \sigma')$ passing only through facets.

\vspace{-0.05cm}

\begin{tcolorbox}[colback=white!95!black,left=3pt,right=3pt,top=3pt,bottom=3pt]
\begin{center}
To summarize, for an arbitrary collection $T$ of triplets, our main two questions are:

\vspace{-0.35cm}

\begin{enumerate}
    \item Can we find a good-facet geometry $R_T$? That is, a geometry where $f_T$ has a pole along \emph{every} facet?
    \vspace{-0.25cm}
    \item Can we find a strongly connected geometry $R_T$?
\end{enumerate}
\end{center}
\end{tcolorbox} 

Note that if the diagram is planar, the geometry $G_+(2,n)$ is both a good-facet geometry and strongly connected. So we are looking for geometries which are in some sense comparably nice to the planar geometry\footnote{Other nice properties of $G_+(2,n)$, such as being topologically a ball, are too restrictive to ask for.}. We also note that the two properties are quite different. The first is about how well the geometry reflects the canonical function $f_T$, while the second is about how ``tightly" the geometry is packed together.

We find two particularly nice classes of diagrams which have good-facet or strongly connected geometries. Recall from Section 4.1 that a diagram is {irreducible} if $f_T$ does not admit a non-trivial factorization\footnote{A more precise characterization is given by irreducible hypertrees \cite{Castravet:2013}.} from Lemma \ref{lem:factorization}. We explain below that if $T$ is irreducible, then for any orientation of $T$, there is a good-facet geometry $R_T$.

The second nice class of diagrams are the \emph{internally planar} diagrams. A diagram is internally planar if it can be drawn in a planar way once the external legs have been deleted. One example is the 6-point diagram in \eqref{eq:legless-ex}. We emphasize that, while all planar diagrams are also internally planar, most internally planar diagrams are \emph{not} planar on-shell diagrams. This is because planar on-shell diagrams by definition must admit a planar drawing with external legs ending on the boundary of the disk. We prove in Appendix C that all internally planar diagrams have a strongly-connected geometry $R_T$. (In fact, the orientation of the triplets which gives this geometry is dictated by the planar drawing of the diagram without legs; one reads around each black vertex clockwise.)

It is a fact of life, however, that for some on-shell diagrams, no geometries are strongly connected (see Section \ref{sec:7point-irred}), or no geometries are good-facet (see Section \ref{sec:8point-red}). Even for internally planar irreducible diagrams, which have strongly connected good-facet geometries, all geometries may have null boundaries of higher codimension and so the diagram admits no positive geometry (see Section \ref{sec:6point-irred}). Our results are summarized in the diagrams below. We note that planar diagrams are internally planar but are technically \emph{not} irreducible; however, they admit strongly connected good-facet geometries, so we treat them as ``morally" irreducible.


\[
    \begin{tikzpicture}
        \node() at (-8,0) {};
        \node() at (8,0) {};
        
        \node[draw,rectangle](general) at (0,0) {General MHV On-Shell Diagrams};

        \node() at (-2,-0.65) {$\bigcup$};
        \node[draw,rectangle](int_planar) at (-3,-1.3) {Internally Planar};

        \node() at (2,-0.65) {$\bigcup$};
        \node[draw,rectangle](irred) at (3,-1.3) {Irreducible};

        \node() at (-2,-1.95) {$\bigcup$};
        \node() at (2,-1.95) {$\bigcup$};
        \node[draw,rectangle](int_planar_irred) at (0,-2.6) {Internally Planar Irreducible};

        \node() at (0,-3.25) {$\bigcup$};
        \node[draw,rectangle]() at (0,-3.9) {Planar};

        \node[text=blue,anchor=west]() at (-8,-0.3) {\small All diagrams have};
        \node[text=blue,anchor=west]() at (-8,-0.65) {\small a strongly connected geometry.};
        \draw[blue,->] (-7,-0.8) to[out=-90,in=-180] (int_planar.west);
    
        \node[text=blue,anchor=west]() at (3.4,-0.3) {\small All diagrams have a};
        \node[text=blue,anchor=west]() at (3.4,-0.65) {\small good-facet geometry.};
        \draw[blue,->] (6,-0.8) to[out=-90,in=0] (irred.east);

        \node[text=red,anchor=west]() at (-7.8,-2.6) {\small For some diagrams,};
        \node[text=red,anchor=west]() at (-7.8,-2.95) {\small all geometries have};
        \node[text=red,anchor=west]() at (-7.8,-3.3) {\small a null boundary.};
        \draw[red,->] (-4.6,-2.75) to[out=0,in=180] (int_planar_irred.west);

        \node[text=red,anchor=west]() at (3.5,-2.6) {\small Some diagrams have no};
        \node[text=red,anchor=west]() at (3.5,-2.95) {\small strongly connected};
        \node[text=red,anchor=west]() at (3.5,-3.3) {\small geometries.};
        \draw[red,->] (4.5,-2.4) to[out=90,in=-90] (irred.south);

        \node[text=red,anchor=west]() at (-6.25,1.1) {\small For some diagrams,};
        \node[text=red,anchor=west]() at (-6.25,0.75) {\small all geometries have null facets.};
        \draw[red,->] (-1.2,0.8) to[out=0,in=90] (general.north);
    \end{tikzpicture}
\]

\newpage
\noindent We can also visualize the different sets of on-shell diagrams by the Venn diagram:

\[
    \begin{tikzpicture}
        \draw (0,0) ellipse (7 and 4.5);
        \draw[fill=gray,fill opacity=0.5] (-2,-0.3) ellipse (4 and 3);
        \draw[fill=gray,fill opacity=0.5] (2,-0.3) ellipse (4 and 3);
        \draw[fill=white] (0,-1) ellipse (1 and 1);

        \node[align=center]() at (0,3.2) {General MHV Diagrams \\ ${\color{red} {\bullet}~ {\bullet}~ {\bullet}}$ };
        \node[align=center]() at (-4,0) {Internally Planar \\ ${\color{blue} \bullet}~{\color{red} {\bullet}~ {\bullet}}$ };
        \node[align=center]() at (0,1) {Irreducible \\ Internally Planar \\ ${\color{blue} \bullet}~{\color{red} {\bullet}}~{\color{blue} \bullet}$ };
        \node[align=center]() at (4,0) {Irreducible \\ ${\color{red} \bullet}~{\color{red} {\bullet}}~{\color{blue} {\bullet}}$ };
        \node[align=center]() at (0,-1) {Planar \\ ${\color{blue} {\bullet}~ {\bullet}~ {\bullet}}$};
    \end{tikzpicture}
\]
In the image we use dots $\bullet \bullet \bullet$ to represent the three properties of our interest in this order:
\begin{itemize}
    \item All diagrams in the set admit strongly connected geometries.
    \vspace{-0.15cm}
    \item All diagrams in the set have good-boundary geometries.
    \vspace{-0.15cm}
    \item All diagrams in the set have good-facet geometries.
\end{itemize}
Blue $\color{blue} \bullet$ signifies the property is always present, and red $\color{red} \bullet$ signifies that there exists some counterexample. We can see that the only diagrams which have all three properties are planar diagrams, while the generic (non-planar) MHV diagrams fail to have any of these three. The `best' set of non-planar on-shell diagrams are irreducible internally planar diagrams which only fail to satisfy the second property.  


Before moving to examples, we discuss why irreducible diagrams have good-facet geometries for any orientation of the triplets $T$. Dividing the triplets into maximal planar subdiagrams\footnote{This can always be done uniquely. The only way for it to not be unique would be if three triplets share the same doublet $(abc),(abd),(abe)$ with any pair being part of some maximal planar diagram. }, e.g.
\begin{equation}\label{eqn:8-point-maximally-decomposed}
    T = \{(123)\} \cup \{(345)\} \cup \{(567)\} \cup \{(178)\} \cup \{(284),(246)\},
\end{equation}
one finds (see Lemma \ref{lem:poleirred}) that the denominator of the form $f_T$ is precisely the product of the denominators from each planar subdiagram:

\begin{equation}
    f_T = \frac{((17) (28) (34) (56) - 
 (18) (26) (34) (57) + 
 (18) (24) (35) (67))^2}{(12)(23)(31) \times (34)(45)(53) \times (56)(67)(75) \times (17)(78)(81) \times (28)(86)(64)(42)}.
\end{equation}
If $(ab)$ is not a pole of $f_T$, then either $a,b$ do not appear together in any triplet, or they are non-adjacent indices in a maximal planar subdiagram. In the former case, for each oriented region $R(\cdots ab \cdots)$ there is a corresponding $R(\cdots ba \cdots)$ glued on the facet $(ab)=0$. In the latter case, the natural orientation of the triplets in the planar subdiagram keeps the indices non-adjacent in the oriented regions $R(\cdots a \cdots b \cdots)$. So, choosing suitable orientations within each maximal planar subdiagram, one can ensure that $(a b)=0$ is not a facet of $f_T$, and thus one can always find a good-facet geometry.



We note that we do not characterize the diagrams admitting strongly connected or good-facet geometries, though we show that this certainly includes internally planar and irreducible diagrams, respectively. We also do not investigate which diagrams admit positive geometries. However, the previous analysis of the simplest non-planar diagram, the 5-point example in \eqref{eqn:5point-non-planar}, shows that all geometries $R_T$ have null boundaries and so are not positive geometries.  It seems likely that more generally, diagrams with a 5-point non-planar subdiagram admit only pseudo-positive $R_T$. This motivates the following question: 

\vspace{0.1cm}

\begin{tcolorbox}[colback=white!95!black,left=3pt,right=3pt,top=6pt,bottom=6pt]
\begin{center}
Are planar on-shell diagrams the only diagrams which admit a positive geometry $R_T$?
\end{center}
\end{tcolorbox} 

\vspace{0.1cm}

If such a geometry exists for a non-planar diagram, it should be connected as a disconnected geometry does not admit an unique canonical form. For two disconnected positive geometries, the canonical form for its union would be
\begin{equation}
    \Omega = \Omega_1 \pm \Omega_2
\end{equation}
and there is no way how to fix a relative sign. In fact, the situation is more complicated as for some connected geometries we can have a similar problem. Namely, if two positive geometries share only a vertex, and not a higher dimensional boundary, then the ambiguity of the canonical form is present too. In any case, from our explorations (and using the 5pt argument above) it seems likely that the only positive geometries are positive Grassmannians $G_+(2,n)$, associated with planar diagrams.  

\subsection{Examples}

In the remainder of the section, we take a closer look at strongly connected and god-facet examples, and provide examples of geometries lacking these properties.

\subsubsection{6-point internally planar irreducible}\label{sec:6point-irred}

The smallest irreducible diagram appears at 6 points, defined by the set of triplets $T = \{(123),(345),(561),(264)\}$. It is also the unique 6-point irreducible diagram, and is internally planar as shown in this embedding:
\begin{equation}
    \vcenter{\hbox{ \onshellgraph{{(0,0)/1,(3.3333,2.6944)/2,(4,6.9282)/3,(4.6667,2.6944)/4,(8,0)/5,(4,1.5396)/6}}{{(0,0)/W1,(3.3333,2.6944)/W2,(4,6.9282)/W3,(4.6667,2.6944)/W4,(8,0)/W5,(4,1.5396)/W6}}{{(2.6667,3.0792)/B1,(5.3333,3.0792)/B2,(4,0.7698)/B3,(4,2.3094)/B4}}{{{W1,B1},{W2,B1},{W3,B1},{W3,B2},{W4,B2},{W5,B2},{W5,B3},{W6,B3},{W1,B3},{W2,B4},{W6,B4},{W4,B4}}}{0.6} }}
\end{equation}

With the orientation of triplets induced by this embedding, one finds a decomposition labeled by
\begin{equation}
    S_6^{(T)} = \{142536,142563,145236,145263,125364,125634,152364,152634\}
\end{equation}
and the canonical function is
\begin{equation}
    \begin{aligned}
    f_T &= \frac{((15) (26) (34)-(16) (24) (35))^2}{(12) (13) (15) (16) (23) (24) (26) (34) (35) (45) (46) (56)} \\
    &= {\rm PT}(142536) + {\rm PT}(142563) + {\rm PT}(145236) + {\rm PT}(145263) \\
    & \quad \quad + ~{\rm PT}(412536) + {\rm PT}(412563) + {\rm PT}(415236) + {\rm PT}(415263).
    \end{aligned}
\end{equation}

This canonical function is the canonical function of the union of oriented regions
\begin{equation}
    \begin{aligned}
    R_T &= R(142536) \cup R(142563) \cup R(145236) \cup R(145263) \\
    & \quad \quad \cup ~ R(412536) \cup R(412563) \cup R(415236) \cup R(415263) \\
    & = \left\{\begin{matrix} (12) \geq 0, & (13) \geq 0, & & (15) \geq 0 , & (16) \geq 0, \\ & (23) \geq 0 , & (24) \leq 0 , & & (26) \geq 0, \\ & & (34) \leq 0, & (35) \leq 0 , \\ & & & (45) \geq 0, & (46) \geq 0, \\ & & & & (56) \geq 0. \end{matrix}\right\}
    \end{aligned}
\end{equation}

This geometry is strongly connected, as one can see by the gluing on the facets $(14),(25),(36)=0$. This property is specific to the orientation chosen; flipping one of the triplets can lead to a geometry that is not strongly connected:
\begin{equation}
    \begin{aligned}
        R_{(123),(345),(561),(246)} &= [R(123456)] \cup [R(124536) \cup R(124563)] \\
        & \quad \quad \cup ~ [R(125346) \cup R(152346)] \cup [R(145623) \cup R(415623)],
    \end{aligned}
\end{equation}
where the strongly connected components have been grouped with braces.

This is a good-facet geometry, as one can see from the defining inequalities matching the poles of the canonical function. This property is independent of the orientation chosen; each triplet is its own maximal planar subdiagram, so any orientation avoids null facets by making the non-pole facets $(14),(25),(36)=0$ be only internal boundaries of the union.

However, this geometry has null boundaries. Taking the residues on $6 = \alpha 5$ and $\alpha = 0$, the canonical function is
\begin{equation}
    \res_{\alpha = 0} \res_{6 = \alpha5} f_T = \frac{(23)}{(12)(13)(24)(34)(25)(35)},
\end{equation}
and the corresponding geometry living in $G(2,5) \cong \{ C \in G(2,6) : \text{column 6 vanished}\}$ is
\begin{equation}
    R_T \bigg|_{6 = \alpha 5} \bigg|_{\alpha = 0} = R(14253) \cup R(41253).
\end{equation}
This is just a relabeling of the canonical function (\ref{eqn:5point-form}) and first geometry (\ref{eqn:5point_R1}) from the 5-point non-planar diagram. As analyzed earlier, the 5-point non-planar diagram has a null boundary for any geometry that is chosen. Since the 6-point irreducible diagram contains the 5-point non-planar diagram on a boundary, one deduces that \emph{any} geometry chosen for the 6-point irreducible diagram must also have a null boundary at some codimension.

\subsubsection{7-point irreducible}\label{sec:7point-irred}
The unique irreducible diagram at 7 points is defined by the set of triplets
\[
T = \{(123),(145),(167),(246),(357)\}
\] 
\begin{equation}
    \onshellgraph{ {(0,0)/6,(1.5,2)/2,(0,4)/4,(2,3)/1,(4,4)/5,(3,2)/3,(4,0)/7} }{ {(0,0)/W1,(1.5,2)/W2,(0,4)/W3,(2,3)/W4,(4,4)/W5,(3,2)/W6,(4,0)/W7} }{ {(0.5,2)/B1,(2,3.5)/B2,(3.5,2)/B3,(2,0.5)/B4,(2.5,2)/B6} }{ {{W1,B1},{W2,B1},{W3,B1},{W3,B2},{W4,B2},{W5,B2},{W5,B3},{W6,B3},{W7,B3},{W7,B4},{W1,B4},{W2,B6},{W4,B4},{W4,B6},{W6,B6}} }{1}
\end{equation}
and has canonical function
\begin{equation}
  f_T= \frac{((1 4) (1 7) (2 6) (3 5) - 
   (1 5) (1 6) (2 4) (3 7))^2}{((1 2) (1 
    3) (1 4) (1 5) (1 6) (1 7) (2 3) (2 
    4) (2 6) (3 5) (3 7) (4 5) (4 6) (5 
    7) (6 7))}.
\end{equation}
The diagram is not internally planar, so there is no preferred orientation of the triplets.
Using the orientation above, $|S_7^{(T)}| = 24$ and we can write it as
\begin{equation}
    \begin{aligned}
        S_7^{(T)} & = \{1234567,1234657,1243567,1243657,1246357\} \\
        & \quad \quad \cup ~\{1624735,1627435,1627345,1672435,1672345\}  \\ & \quad \quad \cup ~\{1456273,1456723,1465273,1465723,1462573\} \\
        & \quad \quad \cup ~ \{1245673,1246573\} \cup \{1462735,1467235\} \cup \{1623457,1624357\}  \\
        & \quad \quad \cup ~ \{1462357\} \cup \{1246735\} \cup \{1624573\},
    \end{aligned}
\end{equation}
where the permutations have been grouped by the maximal strongly connected components. We do not show the many other possible decompositions / geometries here, but it is easy to check (with computer assistance) that all possible orientations also give a non-strongly connected geometry.

This is however a good-facet geometry. Like with the 6-point irreducible diagram, this property is independent of the orientation chosen.

\subsubsection{8-point internally planar irreducible}\label{sec:8point-irred}

One 8-point irreducible diagram is
\begin{equation}
    \onshellgraph{ {(0,0)/1,(1,2)/2,(0,4)/3,(2,3)/4,(4,4)/5,(3,2)/6,(4,0)/7,(2,1)/8} }{ {(0,0)/W1,(1,2)/W2,(0,4)/W3,(2,3)/W4,(4,4)/W5,(3,2)/W6,(4,0)/W7,(2,1)/W8} }{ {(0.5,2)/B1,(2,3.5)/B2,(3.5,2)/B3,(2,0.5)/B4,(1.5,2)/B5,(2.5,2)/B6} }{ {{W1,B1},{W2,B1},{W3,B1},{W3,B2},{W4,B2},{W5,B2},{W5,B3},{W6,B3},{W7,B3},{W7,B4},{W8,B4},{W1,B4},{W2,B5},{W4,B5},{W8,B5},{W4,B6},{W6,B6},{W8,B6}} }{1}
\end{equation}
from which one may read the oriented triplets $T = \{(123),(345),(567),(178),(284),(486)\}$. Then, $|S_8^{(T)}| = 47$, and can be described as follows:
\begin{equation}
    \begin{aligned}
        & S_8^{(T)} = \{\sigma \in S_8 / \mathbb Z_8 \, | \, \sigma \text{ related to } 14725836 \text{ by adjacent transpositions} \\ & \hspace{18em}\text{such that no index is swapped twice}\}.
    \end{aligned}
\end{equation}

For example, $17428536 = 1\overset{\curvearrowright\!\!\!\!\!\!\curvearrowleft}{47}2\overset{\curvearrowright\!\!\!\!\!\!\curvearrowleft}{58}36$ is in the set, but $14257836$ is not. The adjacent transpositions also include swapping $1$ and $6$, as the permutations defined only up to cyclic shifts. The corresponding canonical function is
\begin{equation}
    f_T = \frac{((17) (28) (34) (56)+(18) ((24) (35) (67)-(26) (34) (57)))^2}{(12) (13) (17) (18) (23) (24) (28) (34) (35) (45) (46) (56) (57) (67) (68) (78)}
\end{equation}

The resulting geometry is strongly connected, with $R(14725836)$ being `at the center' of the clump of positive regions, so none of its facets are boundaries of $R_T$.

One may write
\begin{equation}
    R_T = \left\{ \begin{matrix} (12) \geq 0, & (13) \geq 0, & & (15) \geq 0, & & (17) \geq 0, & (18) \geq 0, \\ & (23) \geq 0, & (24) \leq 0, & & (26) \geq 0, & & (28) \geq 0, \\ & & (34) \leq 0, & (35) \leq 0, & & (37) \leq 0, & \\ & & & (45) \geq 0, & (46) \geq 0, & & (48) \geq 0, \\ & & & & (56) \geq 0, & (57) \leq 0, & \\ & & & & & (67) \leq 0, & (68) \leq 0, \\ & & & & & & (78) \geq 0. \end{matrix} \right\}.
\end{equation}
These are the inequalities defining $R(14725836)$ with those involving adjacent indices removed, as they are allowed to swap. This is also a good-facet-geometry, which is expected as the diagram is irreducible. Unlike the 6 and 7-point cases in \ref{sec:6point-irred} and \ref{sec:7point-irred}, this does depend on the orientation: only those orientations where $(284),(486)$ are oriented as a square produce a geometry with no null facets.

\subsubsection{8-point reducible with null facets}\label{sec:8point-red}

Consider the diagram defined by triplets $T = \{(123),(124),(125),(346),(357),(458)\}$,
\begin{equation}
    \vcenter{
    \hbox{
    \onshellgraph{ {(0,0)/1,(2,0)/2,(0,1.5)/3,(1,1.5)/4,(2,1.5)/5,(0,3)/6,(2,3)/8,(1,3)/7} }{ {(0,0)/W1,(2,0)/W2,(0,1.5)/W3,(1,1.5)/W4,(2,1.5)/W5,(0,3)/W6,(2,3)/W8,(1,3)/W7} }{ {(0,0.75)/B1,(1,0.75)/B2,(2,0.75)/B3,(0.25,2.25)/B4,(1.75,2.25)/B5,(1,2.25)/B6} }{ {{B1,1},{B1,2},{B1,3},{B2,1},{B2,2},{B2,4},{B3,1},{B3,2},{B3,5},{B4,3},{B4,4},{B4,6},{B5,4},{B5,5},{B5,8},{B6,3},{B6,5},{B6,7}} }{1}
    }
    }
\end{equation}

This diagram has many lone indices, so it is clearly reducible, and the canonical function is easy to compute
\begin{equation}
    f_T = \frac{(34)(35)(45)(12)}{(13)(23)(14)(24)(15)(25)(36)(46)(37)(57)(38)(58)}.
\end{equation}
Accounting for possible square moves, and for every possible orientation of triplets, one finds for every geometry that at least one of the following is a null facet: $(34),(35),(45) = 0$. 

Furthermore, since the diagram is not internally planar, we do not have an orientation induced by a corresponding embedding. Checking all possible geometries, one finds that none are strongly connected.

\subsection{Geometries without a Diagram}

What makes the forms, or corresponding geometries, coming from on-shell diagrams special? From the analysis done earlier in the section it is evident that properties corresponding to planar diagrams are not preserved for a general on-shell diagram; the geometries one generally finds are not strongly connected and have null facets. To begin answering this question, one may study those forms/geometries that do not come from on-shell diagrams to see what may be missing.

Enumerating the canonical functions for pseudo-positive geometries consisting of unions of oriented regions in $G(2,5)$, one finds (up to permutation of indices) three that do not come from on-shell diagrams:
\begin{equation}
    R_A =  R(15352) \cup R(15423) \rightsquigarrow f_A = \frac{(12)(34)(35)+(13)(23)(45)}{(12)(13)(15)(23)(24)(34)(35)(45)}
    \end{equation}
    \begin{equation}
    R_B = R(15234) \cup R(15423) \cup R(15342) \rightsquigarrow f_B = \frac{(13)(23)(25)(45)-(12)(24)(35)^2}{(12)(13)(15)(23)(24)(25)(34)(35)(45)}
    \end{equation}
    \begin{equation}
    R_C =R(13452) \cup R(15324) \rightsquigarrow f_C = \frac{(14)(15)(23)(24)(35)-(12)(13)(25)(34)(45)}{(12)(13)(14)(15)(23)(24)(25)(34)(35)(45)}        
    \end{equation}

It is easy to check that none of these regions are strongly connected, and they all have null facets. Interestingly, though $R_A, R_B, R_C$ do not come from 5-point on-shell diagrams (equivalently, sets of three triplets), they still arise from choosing a larger collection of triplets and taking a union as in \eqref{eqn:PR-decomposition}:
\begin{align}
    T_A &= \{(153),(154),(142),(423)\}, \\
    T_B &= \{(152),(153),(234),(154)\}, \\
    T_C &= \{(132),(134),(152),(345)\}.
\end{align}
This is interesting, because obviously not all unions of regions come from oriented triplets. For example, consider the following geometry
\begin{equation}
    R(12345)  \cup R(12354) \cup R(12453). \label{union1}
\end{equation}
Here the first two regions have labels $4,5$ swapped, but the region found by swapping those indices in the third region, $R(12543)$, is not included. This makes the geometry inconsistent with any triplets. However, the canonical function for this region is equal (up to relabeling) to $f_B$. The task of enumerating all distinct classes of forms is not feasible at 6pt, so we leave it as an open question:  Is there a canonical function of a union of oriented regions that is not represented by triplets? Note that in the absence of the on-shell diagram connection, it is ambiguous if the canonical function contains $+{\rm PT}(\dots)$ or $-{\rm PT}(\dots)$. In fact, for the region (\ref{union1}) we need to consider 
\begin{equation}
    f = {\rm PT}(12345) + {\rm PT}(12354) - {\rm PT}(12453)
\end{equation}
with these particular signs to guarantee that all highest dimensional residues are $\pm1$. If we flip the sign, 
\begin{equation}
    f' = {\rm PT}(12345) + {\rm PT}(12354) + {\rm PT}(12453)
\end{equation}
some of the residues are $\pm2$ and we can no longer talk about pseudo-positive geometries. In a different context, a similar feature was observed in \cite{Dian:2022tpf} where the authors defined the notion of \emph{graded positive geometries}. It seems conceivable that using a `wrong' orientation of the region $R(12453)$ in the union (\ref{union1}) yields such geometry. Note that on-shell diagrams by definition never lead to residues other than $\pm1$ and we never need to discuss the `graded version' of pseudo-positive geometries in this context.

The purpose of studying the canonical functions associated with Grassmannian regions -- unions of oriented regions -- is motivated by another, even more important, question which in some sense motivated the study of non-planar on-shell diagrams in the first place:

\begin{tcolorbox}[colback=white!95!black,left=3pt,right=3pt,top=6pt,bottom=6pt]
\begin{center}
\noindent{\bf Question:} Are the regions in the Grassmannian associated with non-planar on-shell diagrams special? 
\end{center}
\end{tcolorbox}

In our previous discussion, we concluded that the geometries for generic non-planar on-shell diagram (internally non-planar and non-irreducible) do not exhibit any nice properties: some of their facets and boundaries are null and the spaces are disconnected. This is exactly what we expect from a generic union of $R$-regions. So at the level of geometry, we do not see anything special. Nevertheless, the canonical functions $f_T$ are very special thanks to the formula (\ref{eqn:triplet_form}). In particular, the numerator of the canonical function is a perfect square of a polynomial ${\cal P}[(\dots)]$ in Pl\"uckers multiplied by linear factors,
\begin{equation}
    f = \frac{\prod ({\dots}) \times {\cal P}[({\dots})]^2}{\prod ({\dots})} \label{special}
\end{equation}
We can see that the 5pt canonical functions $f_A$, $f_B$, $f_C$ not associated with on-shell diagrams do \emph{not} have this form. The same is true for the class of functions we were able to check at 6pt. Hence it is natural to conjecture that only the regions associated with oriented triplets produce the canonical function of the form (\ref{special}). The big open question is what are the geometric implications of this very special property.

\section{Conclusions and Outlook}

In this paper, we explored the MHV non-planar on-shell diagrams and the correspondence to Grassmannian geometries in $G(2,n)$. In the planar case there is only one top cell of $G_+(2,n)$ and one planar on-shell diagram, but there are many non-planar diagrams of the same (maximal) dimensionality for general $n$. Our analysis relied mostly on the triplet description, which determines the canonical function either using the determinant formula or the Parke-Taylor expansion. Our first result is an algorithm for identifying diagrams with vanishing forms. In the planar sector, the vanishing of the form indicates the presence of internal bubbles and the fact that the diagram is not reduced. The same type of diagrams do not have a triplet representation. But there is a new class of diagrams which are described by triplets but still have vanishing form, which we are able to identify. Our second result concerns the identity moves. In the planar case, there are only square moves (and trivial merge-expand move) while in the non-planar case there are also sphere moves proposed in \cite{Cachazo:2019}. We show using a novel doublet representation that indeed square moves and sphere moves are the only identity moves.

Our main result concerns the (pseudo-)positive Grassmannian geometries that we associate with non-planar on-shell diagrams with non-vanishing canonical functions. The Parke-Taylor decompositions provide candidate spaces, which are unions of permuted positive parts $G_+(2,n)$s. However, these spaces are generally disconnected and the canonical function does not encode properly the boundary structure, i.e. there are boundaries for which the form vanishes -- we call them null facets. We showed that for internally planar diagrams there is a choice of Parke-Taylor decomposition such that the space is connected, and for irreducible diagrams there are no null facets. The privileged class of internally planar irreducible diagrams enjoys both of these special properties, and they provide an interesting non-planar extension of the planar on-shell diagrams associated with the positive Grassmannian.

Our paper is just an invitation into the very complex problem of non-planar on-shell diagrams and the Grassmannian geometries. There are two main future directions on the mathematical side: 1) Further study of MHV diagrams and $G(2,n)$ geometries, their boundary structures, topological and combinatorial properties. 2) Generalization to N$^{k{-}2}$MHV on-shell diagrams starting with six-point NMHV case, ie. $G(3,6)$. The complete classification of all non-planar on-shell forms was given in \cite{Bourjaily:2016mnp}, and the list contains exotic terms with polynomial and even algebraic poles. The first step is the classification of all terms with only monomial poles, which is closely related to problems considered in \cite{Early:2024asu,Cachazo:2023ltw}. On the physics side, the canonical forms of non-planar on-shell diagrams are leading singularities of loop integrands, and hence they play a role of coefficients that multiply polylogarithms (and generalizations) in the expansion of non-planar ${\cal N}=4$ SYM amplitudes. If the dual conformal symmetry does extend in some form to the non-planar sector, such kinematic structures should be most accessible to study precisely at the level of leading singularities. We plan to study all these problems in future work. 

\section*{Acknowledgments}

We thank Nick Early, Chris Eur, Melvyn Nabavi, Shruti Paranjape, Alexander Postnikov, Lizzie Pratt, Marcelo Augusto Ferreira dos Santos, Bernd Sturmfels, Lauren Williams for very useful discussions, and Taro Brown for an early collaboration on this work. This work was supported by the DOE grant No.SC0009999 and the funds of the University of California.  MSB was partially supported by the National Science Foundation under Award No.~DMS-2444020. 
  Any opinions, findings, and conclusions or recommendations expressed in this material are
those of the author(s) and do not necessarily reflect the views of the National Science
Foundation.

\vspace{0.2cm}

\begin{center}
    \rule{0.8\textwidth}{0.4pt}
\end{center}

\vspace{0.2cm}

\appendix

\section{The factorization formula}\label{app:factorization}

\newtheorem*{lem:factorization}{Lemma \ref{lem:factorization}}
\begin{lem:factorization}[Factorization of the form]
    Let $T$ be a set of $n-2$ triplets on $[n]$. Suppose $R$ is a valid subset of $T$. Then
    \begin{equation}
          f_T =   f_{R} \times   f_{(T \setminus R) \cup P} \times \prod_{i=1}^{r} ( a_i a_{i+1}),
    \end{equation}
    where $A = \{a_1,\cdots,a_{r}\}$ consists of the indices appearing both in $R$ and $T$ (we set $a_{r+1}=a_1$) and 
    $P = \{(a_1,a_i,a_{i+1}) | i \in \{2,\cdots,r-1\}\}$ is a set of triplets defining a triangulation of a polygon with vertices $A$.
\end{lem:factorization}

\begin{proof}
    Let $S = T \setminus R$ to simplify the notation. Since $R$ is a valid subset, it consists of $m-2$ triplets on $m$ indices. Then, $S$ must contain the remaining $n-m$ triplets of $T$, as well as at least $n-m+r$ vertices. One must have $r \geq 2$, as otherwise $T$ would fail the non-vanishing condition.

    Then, the matrix $M_T$ can be written as
    \begin{equation}
        M_T = \begin{pmatrix} M_{R,1} & M_{R,2} & 0 \\ 0 & M_{S,1} & M_{S,2} \end{pmatrix},
    \end{equation}
    where $M_{R,1}$ is $(m-2) \times (m-r)$, $M_{R,2}$ is $(m-2) \times r$, $M_{S,1}$ is $(n-m) \times r$, and $M_{S,2}$ is $(n-m) \times (n-m)$. Then, computing the form $f_T$ one finds
    \begin{equation}
        \begin{aligned}
            f_T &= \frac{[\det( M_{T,ab}) / (ab)]^2}{\prod_{\tau \in T} (\tau_1\tau_2)(\tau_2\tau_3)(\tau_3\tau_1)} \\
            &= \frac{\left[\det(M_{R,ab}) / (ab)\right]^2}{\prod_{\tau \in R} (\tau_1\tau_2)(\tau_2\tau_3)(\tau_3\tau_1)} \times \frac{[\det(M_{S,2})]^2}{\prod_{\tau \in S} (\tau_1\tau_2)(\tau_2\tau_3)(\tau_3\tau_1)} \\
            &= f_{R} \times \frac{[\det(M_{S,2})]^2}{\prod_{\tau \in S} (\tau_1\tau_2)(\tau_2\tau_3)(\tau_3\tau_1)},
        \end{aligned}
    \end{equation}
    where we begin to see the factorization of the form into a piece dependent on the valid subset of triplets and some term dependent on the rest. This latter term may be rewritten further by introducing a set of extra triplets $P = \{(a_1,a_i,a_{i+1}) | i \in \{2,\cdots,r-1\}\}$ that form an arbitrary polygon over the indices $a_i$ shared by $R$ and $S$. Then,
    \begin{equation}
        \begin{aligned}
            \det(M_{P \cup S,a_1a_2}) / (a_1a_2) &= \det\begin{pmatrix} M_{P,a_1a_2} & 0 \\ M_{S,1,a_1a_2} & M_{S,2} \end{pmatrix} / (a_1a_2) \\
            &= \det(M_{S,2}) \times [\det(M_{P,a_1a_2}) / (a_1a_2)] \\
            &= \det(M_{S,2}) \times \prod_{i=2}^{r-1} (a_1a_i).
        \end{aligned}
    \end{equation}

    This allows us to rewrite
    \begin{equation}
        \begin{aligned}
        f_T &= f_R \times \frac{[\det(M_{S,2})]^2}{\prod_{\tau \in S} (\tau_1\tau_2)(\tau_2\tau_3)(\tau_3\tau_1)} \\
        &= f_R \times \frac{[\det(M_{P \cup S,a_1a_2}) / (a_1a_2)]^2}{\prod_{\tau \in S} (\tau_1\tau_2)(\tau_2 \tau_3) (\tau_3\tau_1) \times \prod_{i=2}^{r+1} (a_1 a_i)^2} \\
        &= f_R \times \frac{[\det(M_{P \cup S,a_1 a_2}) / (a_1 a_2)]^2}{\prod_{\tau \in P \cup S} (\tau_1\tau_2)(\tau_2\tau_3)(\tau_3\tau_1)}  \times \prod_{i=1}^r (a_i a_{i+1}) \\ 
        &= f_R \times f_{P \cup S} \times \prod_{i=1}^r (a_i a_{i+1}),
        \end{aligned}
    \end{equation}
    so the form corresponding to $T$ factorizes into two other forms, with an extra factor from the polygon on which $R$ and $S$ were glued.
\end{proof}

\section{Sphere moves are the only moves}\label{app:sphere}

\begin{lemma}[Pole for lone pair]\label{lem:lonepair}
    Suppose that in the triplets $T$, the indices $i,j$ appear together only in one triplet $(ijk)$. Then, defining $T' = (T \setminus \{(ijk)\}) \bigg|_{j \mapsto i}$, one has
    \begin{equation}
        [(ij) \times f_T] \bigg|_{j \mapsto i} = -f_{T'}.
    \end{equation}
\end{lemma}
\begin{proof}
    Computing the LHS prior to mapping the indices,
    \begin{equation}
        \begin{aligned}
            (ij) \times f_T &= \frac{[\det(M_{ij}) / (ij)]^2}{(jk)(ki) \prod_{(a,b,c) \in T \setminus \{(ijk)\}}(ab)(bc)(ca)} \\
            &= \frac{\bigg[\det\begin{pmatrix} (ij) & 0 \\ \vdots  & M_{ijk}\end{pmatrix} / ij\bigg]^2}{(jk)(ki) \prod_{(a,b,c) \in T \setminus \{(ijk)\}}(ab)(bc)(ca)} \\
            &= \frac{\det(M_{ijk})^2}{(jk)(ki) \prod_{(a,b,c) \in T \setminus \{(ijk)\}}(ab)(bc)(ca)}
        \end{aligned}
    \end{equation}
    where $M_{ijk}$ is the matrix $M$ with columns $i,j,k$ and the row corresponding to triplet $(ijk)$ removed. Then, mapping the index $j \mapsto i$, one finds
    \begin{equation}
        \begin{aligned}
            [(ij) \times f_T] \bigg|_{j \mapsto i} &= -\frac{[\det(M_{ijk}) / (ik)]^2}{\prod_{(a,b,c) \in T \setminus \{(ijk)\}}(ab)(bc)(ca)} \bigg|_{j \mapsto i} \\
            &= -\frac{[\det(M'_{ik}) / (ik)^2]}{\prod_{(a,b,c) \in T'}(ab)(bc)(ca)} \\
            &= -f_{T'},
        \end{aligned}
    \end{equation}
    as desired.
\end{proof}

\begin{lemma}[Poles for irreducible hypertrees]\label{lem:poleirred}
    For $T$ a set of triplets corresponding to an irreducible hypertree, the denominator of $f_T$ matches the doublets $D(T)$.
\end{lemma}
\begin{proof}
    For $\{i,j\} \in D(T)$, we can repeatedly apply square moves $(ija),(ijb) \mapsto (iab),(jab)$, rewriting $T$ such that $i,j$ appear in only one triplet $(ijk)$. This allows us to apply Lemma \ref{lem:lonepair} to conclude that $(ij)$ must be a pole of $f_T$ if $T' = (T \setminus \{(ijk)\})\bigg|_{j \mapsto i}$ satisfies the non-vanishing condition (\ref{eqn:non-vanishing}).

    Suppose the triplets $T'$ fail the non-vanishing condition. Then, there exists some subset of triplets $R' \subset T'$ such that
    \[
        \left|\bigcup_{(abc) \in R'} \{a,b,c\}\right| < |R'|+2.
    \]
    Denote by $R \subset T$ the subset of triplets that maps to $R'$ after the substitution $j \mapsto i$. Note that $(ijk) \notin R$.
    
    If the index $j$ does not appear in $R$, then $R = R'$, and consequently $T$ also fails the non-vanishing condition, leading to a contradiction with $T$ coming from an irreducible hypertree.

    Using that $j$ is the only index in $R$ but not in $R'$, one has
    \[
        \left|\bigcup_{(abc) \in R} \{a,b,c\}\right| = \left|\bigcup_{(abc) \in R'} \{a,b,c\}\right| + 1 \leq |R'|+2 = |R|+2.
    \]
    As $T$ satisfies the nonvanishing condition, the left-hand side and right-hand side must be equal. So $R$ is a valid subset of $T$.

    If $R$ does not correspond to a planar diagram, then $T$ has a valid non-planar subset, contradicting that it came from an irreducible hypertree.

    If $R$ does correspond to a planar diagram, we consider two cases both of which lead to contradictions. If the index $k$ is contained in $R$, then $R \cup \{(ijk)\}$ causes $T$ to fail the non-vanishing condition by adding a triplet but not an index. If the index $k$ is not contained in $R$, we either have that $R \cup \{(ijk)\} = T$ leading to a factorization of $f_T = f_R \times \frac{(ij)}{(ik)(kj)}$, or $R \cup \{(ijk)\} \subset T$ is a valid non-planar subset.

    Thus, $T'$ does not fail the non-vanishing condition, making $(ij)$ a pole of $f_T$. This shows that if $ij$ is a doublet of $T$, then $(ij)$ is a pole of $f_T$.

    To show the opposite, note that one can use a square move to replace any triplets with a repeated doublet $(ija),(ijb)$ with two triplets that do not have this doublet $(iab),(jab)$. For any $ij$ that is not a doublet of $T$, it appears in an even number of triplets, to which the square move can be applied to remove any appearances of the doublet. Using the set of triplets where the doublet $ij$ does not appear, one can see from the determinant formula that $(ij)$ cannot appear in the denominator.
\end{proof}

\begin{lemma}[Doublets from the form]\label{lem:doublets-from-form}
    For a set of triplets $T$, the following definitions of the doublets are equivalent:
    \begin{itemize}
        \item $D(T) = \{\{i,j\} | \text{$i$ and $j$ appear together in an odd number of triplets}\}$.
        \item $D(T) = \{\{i,j\} | \text{$(ij)$ appears with an odd power in $f_T$}\}$.
    \end{itemize}
\end{lemma}
\begin{proof}
    For diagrams corresponding to irreducible hypertrees, the numerator of the form has no monomials, so this statement is a corollary of Lemma \ref{lem:poleirred}.

    Otherwise, suppose the statement is true for all sets of triplets satisfying $|T| \leq n$. We show the statement is then true for arbitrary triplets $T$ where $|T| = n$. Supposing that $T$ not an irreducible hypertree, we apply the factorization formula (\ref{eqn:factorization-body}) for some proper subset $R$ that does not share all its indices with $T \setminus R$,
    \begin{equation}
        f_T = f_R \times f_{(T \setminus R) \cup P} \times \prod_{i=1}^r (a_i a_{i+1}).
    \end{equation}
    Since $|R| < |T| = n$ and $|(T \setminus R) \cup P| < |T| = n$, we can apply our inductive hypothesis. The parity of odd powers of monomials on the RHS is exactly the parity of the doublets appearances in $D(R)$ and $D(T \setminus R)$, where the parity coming from $D(P)$ is canceled by $\prod_{i=1}^r (a_i a_{i+1})$. Using a single triplet as the base case, the statement is true for all $n$.
\end{proof}

\begin{lemma}[Same doublets imply sphere moves]\label{lem:doublet-to-sphere}
    Suppose two sets of triplets $T$ and $T'$ with nonvanishing forms have the same doublets, $D(T) = D(T')$. Then, these sets of triplets are related by a sequence of sphere moves.
\end{lemma}
\begin{proof}    
    Without loss of generality, suppose $T$ and $T'$ have no triplets in common, $T \cap T' = \emptyset$. If they did share triplets, the shared triplets $T \cap T'$ can simply be ignored as we still have $D(T \setminus (T \cap T')) = D(T' \setminus (T \cap T'))$, and can construct sphere moves for the remaining triplets. Now, since $T,T'$ are disjoint, the equality $D(T) = D(T')$ implies that $D(T \cup T') = \emptyset$.

    We will now construct sets $R \subset T$ and $R' \subset T'$ related by a sphere move. We will do this by starting with $R, R' = \emptyset$, and incrementally adding triplets to each via the algorithm below. As this algorithm unfolds, it correspondingly constructs a cell complex with black faces labeled by $R$, white faces labeled by $R'$, and vertices labeled by indices in $R \cup R'$.
    
    At any step, each edge must appear in at most two faces. Each time a triplet is added in the algorithm below, a face is added to the complex. This face is to be glued to any edge that does not yet have two faces attached; otherwise new edges are created. Note that this means that complex may not be simplicial: the same pair of vertices may have multiple edges between them if they appear together in more than two faces. This procedure ensures that the boundary of the complex consists of edges identified with $D(R \cup R')$.

    To start, we choose a single triplet $t \in T$ and add it to $R$. Now, we proceed with the following cycle, repeating until we complete a full cycle without adding any new triplets to $R$ or $R'$.
    \begin{enumerate}
        \item Suppose there exists a triplet in $T \setminus R$ containing a pair of indices $ij \in D(R) \setminus D(R')$; add this triplet to $R$. Repeat this step until no such triplet exists. (Correspondingly in the complex, attach as many black faces to existing black faces as possible.)
        \item Suppose there exists a triplet in $T' \setminus R'$ containing a pair of indices $ij \in D(R) \setminus D(R')$; add this triplet to $R'$. Repeat this step until no such triplet exists. (Correspondingly in the complex, attach as many white faces to existing black faces as possible.)
        \item Suppose there exists a triplet in $T' \setminus R'$ containing a pair of indices $ij \in D(R') \setminus D(R)$; add this triplet to $R'$. Repeat this step until no such triplet exists. (Correspondingly in the complex, attach as many white faces to existing white faces as possible.)
        \item Suppose there exists a triplet in $T \setminus R$ containing a pair of indices $ij \in D(R') \setminus D(R)$; add this triplet to $R$. Repeat this step until no such triplet exists. (Correspondingly in the complex, attach as many black faces to existing white faces as possible.)
    \end{enumerate}
    Since $T$ and $T'$ are finite, this procedure must terminate. We now seek to show that the resulting complex has no boundary, i.e. $D(R \cup R') = \emptyset$.
    
    Suppose $\exists ~ij \in D(R \cup R')$. Then, either $ij \in D(R) \setminus D(R')$ or $ij \in D(R') \setminus D(R)$. Since the algorithm terminating implies no triplet satisfying the necessary conditions existed in each step, one finds that $T \setminus R$ and $T' \setminus R'$ contain no triplets with indices $ij$. Consequently, $ij \not \in D(T \setminus R)$ and $ij \not \in D(T'\setminus R')$. But then,
    \begin{equation}
        ij \in D( [R \cup R'] \cup [T \setminus R] \cup [T' \setminus R']) = D(T \cup T'),
    \end{equation}
    which is a contradiction with $D(T \cup T') = \emptyset$.

    So, the complex with black faces $R$ and white faces $R'$ has no boundary, and each edge is in exactly two faces. The number of faces is $F = |R| + |R'|$ and the number of edges is $E = \frac{3}{2}F$. The number of vertices $V$ is equal to the number of distinct indices in $R \cup R'$. The non-vanishing condition for either $T$ or $T'$ requires that $|R| \leq V-2$ and $|R'| \leq V-2$. The Euler characteristic of the complex is $\chi = V-E+F = V - (|R|+|R'|)/2 \geq V-(V-2) = 2$, implying that it is a sphere.

    There is a possible obstruction to identifying the complex as the triangulation of a sphere: one might have had a pinch-point, i.e. a vertex with multiple cycles of faces around it. However, a similar counting argument to the main text reveals that the existence of a pinch-point would contradict the non-vanishing condition. We remove the pinch-points by replacing each vertex by a pair of vertices with an extra edge for each cycle, e.g.
    \[
    \vcenter{\hbox{\begin{tikzpicture}
        \draw (0.1,0.98) -- (2,0) -- (0.1,-0.98);
        \draw[fill=gray!50!white] (0,0) ellipse (0.5 and 1);
        \draw (3.9,0.98) -- (2,0) -- (3.9,-0.98);
        \draw[fill=gray!50!white] (4,0) ellipse (0.5 and 1);
        \draw[fill=red] (2,0) circle (0.05);
    \end{tikzpicture}}}
    \hspace{2em} \rightarrow \hspace{2em}
    \vcenter{\hbox{\begin{tikzpicture}
        \draw[fill=gray!50!white] (0,0) ellipse (0.5 and 1);
        \draw[fill=gray!50!white] (4,0) ellipse (0.5 and 1);
        \draw (0,1) -- (4,1);
        \draw (0,-1) -- (4,-1);
        \draw[blue,dashed] (2,-1) arc (270:90:0.5 and 1);
        \draw[blue] (2,1) arc (90:-90:0.5 and 1);
        \draw[fill=red] (2,1) circle (0.05);
        \draw[fill=red] (2,-1) circle (0.05);
    \end{tikzpicture}
    }}
    \]
    Suppose this procedure introduces $k$ new vertices, $l$ new handles, and $k+l$ new edges (at each vertex, there is one more new edge than there are new handles). Then, the number of faces is still $\tilde F = |R|+|R'|$, the number of edges is now $\tilde E = \frac{3}{2} F + k + l$, and the number of vertices is $\tilde V = V + k$. Using the non-vanishing condition, one finds $\tilde \chi = \tilde V - \tilde E + \tilde F \geq 2 - l$. But with the $l$ extra handles one also has genus at least $l$, so $\tilde \chi \leq 2-2l$. So, we reach a contradiction unless $l=0$, the case when no such pinch-points existed.

    Consequently, $R$ and $R'$ together form a bicolored triangulation of the sphere, and are related by a sphere move. Thus, for any $T$ and $T'$ that have the same doublets, there is a sphere move that allows some subset of the different triplets to be related. Repeating until no triplets differ we find a sequence of sphere moves that relate $T$ and $T'$.

\end{proof}

We note that in the previous proof, we gloss over some topological details on when a pure 2-dimensional cell complex is a manifold.

\newtheorem*{thm:moves}{Theorem \ref{thm:moves}}
\begin{thm:moves}
    Let $T$ and $T'$ be two sets of triplets with nonvanishing forms (i.e. $T, T'$ satisfy \eqref{eqn:non-vanishing}). The following are equivalent:
    \begin{enumerate}[label=(\arabic*)]
        \item There is a sequence of sphere moves relating $T$ and $T'$.
        \item The corresponding forms are equal: $f_T = f_{T'}$.
        \item The corresponding doublets are the same: $D(T) = D(T')$.
    \end{enumerate}
\end{thm:moves}
\begin{proof}
    $(1) \implies (2)$ is a consequence of \cite{Cachazo:2019}.

    $(2) \implies (3)$ is a corollary of Lemma \ref{lem:doublets-from-form}.

    $(3) \implies (1)$ is proven by Lemma \ref{lem:doublet-to-sphere}.
\end{proof}

\section{Geometries induced by planar embeddings are strongly connected}\label{app:spherical-diagrams}

Before tackling the full question to do with geometries that come from orientations of triplets, we prove a more general result for geometries that come from a partial order on $\{1,\cdots,n\}$. For a partial order $P = \{ i \prec j, \cdots\}$, the set of linear extensions $S_n^{(P)}$ is
\begin{equation}
    S_n^{(P)} = \{\sigma \in S_n : \sigma_{i} < \sigma_{j} ~ \forall ~ i \prec j \in P\}.
\end{equation}
These define a corresponding geometry via a union of oriented regions,
\begin{equation}
    R_P(\varepsilon) = \bigcup_{\sigma \in S_n^{(P)}} R(\varepsilon\sigma),
\end{equation}
where $\varepsilon \in \{+,-\}^n$ encodes the sign of each label.

\begin{lemma}[Linear extensions of a poset are connected]
    Any element of $S_n^{(P)}$ can be transformed into any other using only adjacent transpositions without leaving the set.
\end{lemma}
\begin{proof}
    By relabeling the indices if necessary, suppose without loss of generality that the identity permutation $12\cdots n$ is in $S_n^{(P)}$.
    
    The pair of indices $i < j$ is an inversion of a permutation $\sigma$ if $\sigma_i > \sigma_j$. We denote the number of inversions by $\mathrm{inv}(\sigma)$. The only permutation with $\mathrm{inv}(\sigma) = 0$ is $12\cdots n$.

    We show that for any $\sigma \in S_n^{(P)}$ with $\mathrm{inv}(\sigma) > 0$, there exists a permutation $\sigma' \in S_n^{(P)}$ related by an adjacent transposition such that $\mathrm{inv}(\sigma') = \mathrm{inv}(\sigma)-1$.

    Because $\sigma \neq 12 \cdots n$, there is some $i$ such that $\sigma_i > \sigma_{i+1}$. Notice that $i \prec {i+1} \not\in P$ as these elements are oppositely ordered in $12\cdots n \in S^{(P)}_n$. Consequently, the permutation $\sigma' = \sigma (i ~ i+1)$ is also in $S^{(P)}_n$. This permutation has the same inversions as $\sigma$, except without the pair $i,i+1$, so $\mathrm{inv}(\sigma') = \mathrm{inv}(\sigma) - 1$ as desired.

    By induction on the inversion number, with the identity as the base case, every permutation in $S^{(P)}_n$ can be brought to the identity via adjacent transpositions without leaving $S^{(P)}_n$.
\end{proof}

\begin{corollary}\label{cor:poset-connected}
    The geometry $R_P$ corresponding to the poset $P$ is strongly connected.
\end{corollary}

We observe that the set of permutations defined by an orientation of triplets $S^{(T)}$ can be decomposed into a disjoint union of permutations defined by posets $S^{(P)}$ by choosing each possible cyclic shift of the triplets, except those involving $1$ which is fixed to be the lowest index. For example, with oriented triplets $T = \{(123),(345),(156),(264)\}$, we select a cyclic shift for each triplet not involving $1$,
\begin{equation}
    \begin{aligned}
        P_1 &= \{1 \prec 2 \prec 3, ~ 3 \prec 4 \prec 5, ~ 1 \prec 5 \prec 6, ~ 2 \prec 6 \prec 4, 1 \prec j ~\forall~ j\}, \\
        P_2 &= \{1 \prec 2 \prec 3, ~ 4 \prec 5 \prec 3, ~ 1 \prec 5 \prec 6, ~ 2 \prec 6 \prec 4, 1 \prec j ~\forall~ j\}, \\
        P_3 &= \{1 \prec 2 \prec 3, ~ 5 \prec 3 \prec 4, ~ 1 \prec 5 \prec 6, ~ 2 \prec 6 \prec 4, 1 \prec j ~\forall~ j\}, \\
        P_4 &= \{1 \prec 2 \prec 3, ~ 3 \prec 4 \prec 5, ~ 1 \prec 5 \prec 6, ~ 4 \prec 2 \prec 6, 1 \prec j ~\forall~ j\}, \\
        P_5 &= \{1 \prec 2 \prec 3, ~ 4 \prec 5 \prec 3, ~ 1 \prec 5 \prec 6, ~ 4 \prec 2 \prec 6, 1 \prec j ~\forall~ j\}, \\
        P_6 &= \{1 \prec 2 \prec 3, ~ 5 \prec 3 \prec 4, ~ 1 \prec 5 \prec 6, ~ 4 \prec 2 \prec 6, 1 \prec j ~\forall~ j\}, \\
        P_7 &= \{1 \prec 2 \prec 3, ~ 3 \prec 4 \prec 5, ~ 1 \prec 5 \prec 6, ~ 4 \prec 2 \prec 6, 1 \prec j ~\forall~ j\}, \\
        P_8 &= \{1 \prec 2 \prec 3, ~ 4 \prec 5 \prec 3, ~ 1 \prec 5 \prec 6, ~ 4 \prec 2 \prec 6, 1 \prec j ~\forall~ j\}, \\
        P_9 &= \{1 \prec 2 \prec 3, ~ 5 \prec 3 \prec 4, ~ 1 \prec 5 \prec 6, ~ 4 \prec 2 \prec 6, 1 \prec j ~\forall~ j\}.
    \end{aligned}
\end{equation}
Most of these end up not being valid posets, due to there being a cycle of inequalities, e.g. $4 \prec 5 \prec 6 \prec 4 \in P_1$. In this example, the only ones that are posets are $P_5$ and $P_9$, with linear extensions
\begin{equation}
    \begin{aligned}
        S^{(P_5)}_6 &= \{142536,145236,142563,145263\},\\
        S^{(P_9)}_6 &= \{125364,152364,125634,152634\}.
    \end{aligned}
\end{equation}

Since we enumerated all possible cyclic shifts of the triplets, one finds that $S^{(P_5)}_6 \cup S^{(P_9)}_6 = S_6^{(T)}$, such that
\begin{equation}\label{eqn:6point-poset-example}
    R_T = R_{P_5}(\varepsilon_5) \cup R_{P_9}(\varepsilon_9),
\end{equation}
where $\varepsilon_i$ parametrize the freedom to reflect the regions corresponding to each poset. Corollary \ref{cor:poset-connected} is sufficient to find the strong connectedness within each piece $R_{P_i}$, so one might ask whether we can choose $\varepsilon_i$ such that the entire $R_T$ is strongly connected. In this example, choosing $\varepsilon_5 = ({+}{+}{+}{+}{+}{+})$ and $\varepsilon_9 = ({+}{+}{+}{-}{+}{+})$, the region defined by \eqref{eqn:6point-poset-example} is strongly connected, as each region in $R_{P_5}(\varepsilon_5)$ is adjacent to a region $R_{P_9}(\varepsilon_9)$ one the boundary $(14) = 0$. With the idea of regions corresponding to posets in mind, we now return to proving the following theorem:

\begin{theorem}[Internally planar diagrams admit strongly connected geometries]
    The orientation of triplets $T$ induced by a planar embedding of the corresponding on-shell diagram defines a strongly connected geometry $R_T$ for an appropriate choice of signs.
\end{theorem}

\begin{proof}
    To prove this theorem, we demonstrate a bijection between the posets with non-empty sets of linear extensions, and perfect matchings of a particular graph. To make the procedure clear, we follow an example with the 6-point irreducible internally planar diagram with triplets $T = \{(123),(345),(156),(264)\}$.

    Starting with a planar embedding of the internal edges of the diagram, one may replace each black vertex with a triangular face,
    \begin{equation}
        \vcenter{\hbox{ \onshellgraph{{(0,0)/1,(3.3333,2.6944)/2,(4,6.9282)/3,(4.6667,2.6944)/4,(8,0)/5,(4,1.5396)/6}}{{(0,0)/W1,(3.3333,2.6944)/W2,(4,6.9282)/W3,(4.6667,2.6944)/W4,(8,0)/W5,(4,1.5396)/W6}}{{(2.6667,3.0792)/B1,(5.3333,3.0792)/B2,(4,0.7698)/B3,(4,2.3094)/B4}}{{{W1,B1},{W2,B1},{W3,B1},{W3,B2},{W4,B2},{W5,B2},{W5,B3},{W6,B3},{W1,B3},{W2,B4},{W6,B4},{W4,B4}}}{0.6} }} \hspace{+5em} \vcenter{\hbox{ \decoratedpolygongraph{{(0,0)/1,(3.3333,2.6944)/2,(4,6.9282)/3,(4.6667,2.6944)/4,(8,0)/5,(4,1.5396)/6}}{{{1,2,3},{3,4,5},{5,6,1},{2,6,4}}}{ {{1,2},{2,3},{3,1},{1,5},{5,6},{6,1},{2,6},{6,4},{4,2},{3,4},{4,5},{5,3}} }{ {} }{0.6} }}
    \end{equation}
    
    We fix an orientation of the triplets by reading counterclockwise around each black triangle.  Correspondingly, reading the edges counterclockwise gives us most of the inequalities used by the posets necessary to describe these triplets. To get a specific choice of cyclic shifts of the triplets $P_i$, we must throw away one inequality for each triplet, which we denote by decorating the corresponding edge in a graph $G(P_i)$.
    \begin{equation}\label{eqn:emptyextensions}
        \begin{aligned}
            P_1 = \{&1 \prec 2 \prec 3, ~ 3 \prec 4 \prec 5,\\&~ 1 \prec 5 \prec 6,~ 2 \prec 6 \prec 4,~1 \prec j ~\forall~ j\}
        \end{aligned}
        \longleftrightarrow
        G(P_1) = \hspace{-3em} \vcenter{\hbox{ \decoratedpolygongraph{{(0,0)/1,(3.3333,2.6944)/2,(4,6.9282)/3,(4.6667,2.6944)/4,(8,0)/5,(4,1.5396)/6}}{{{1,2,3},{3,4,5},{5,6,1},{2,6,4}}}{ {{1,2},{2,3},{1,5},{5,6},{2,6},{6,4},{3,4},{4,5}} }{ {{3,1},{6,1},{5,3},{4,2}} }{0.6} }}
    \end{equation}
    \begin{equation}
        \begin{aligned}
            P_5 = \{&1 \prec 2 \prec 3, ~ 4 \prec 5 \prec 3,\\&~ 1 \prec 5 \prec 6,~ 4 \prec 2 \prec 6,~1 \prec j ~\forall~ j\}
        \end{aligned}
        \longleftrightarrow
        G(P_5) = \hspace{-3em} \vcenter{\hbox{ \decoratedpolygongraph{{(0,0)/1,(3.3333,2.6944)/2,(4,6.9282)/3,(4.6667,2.6944)/4,(8,0)/5,(4,1.5396)/6}}{{{1,2,3},{3,4,5},{5,6,1},{2,6,4}}}{ {{1,2},{2,3},{1,5},{5,6},{2,6},{4,2},{4,5},{5,3}} }{ {{3,1},{6,1},{3,4},{6,4}} }{0.6} }}
    \end{equation}

    Those sets of inequalities that form poset cannot contain a cycle of inequalities. In the diagram, this corresponds to there being a directed cycle of non-decorated edges.

    \begin{equation}
        \begin{aligned}
             4 \prec 5 \prec 6 \prec 4 \in P_1 \longleftrightarrow & ~ (4\rightarrow 5\rightarrow 6 \rightarrow 4 \text{ non-decorated cycle in \eqref{eqn:emptyextensions})}\\ &\implies S^{(P_1)}_n = \emptyset
        \end{aligned}
    \end{equation}
    
    Suppose we have two triangles $(ijk)$ and $(ilj)$. To prevent the cycles $i \prec j \prec i$ and $i \prec l \prec j \prec k \prec i$, the decorated edges for these triangles must include one from each group. Generalizing to larger groups of black triangles glued into a black polygon, we conclude that exactly one of the edges of the polygon is decorated, with the remaining decorations spent on preventing 2-cycles in the polygon's triangulation. (No such polygons appear in the tracked example.)
    
    In order to prevent the existence of such cycles around the white faces, each white face must have at least one decorated edge. Since there are an equal number of white and black faces, the only way this is possible is if each white face actually has exactly one decorated edge. On the bipartite dual graph, this corresponds to a perfect matching. Note all edges which are incoming to vertex 1 are decorated, since we have the convention that 1 is the first index labeling an oriented region. So we only need to consider the dual graph $\tilde G(P_i)$ of the faces that do not have $1$ as a vertex.
    \begin{gather*}
        \{P_i : S_n^{(P_i)} \neq \emptyset\} \longleftrightarrow \{\text{perfect matchings of the graph dual to the faces that don't include $1$}\} \\
        \{P_5,P_9\} \longleftrightarrow \left\{
        \tilde G(P_5) = \hspace{-3em} \vcenter{\hbox{ \dualdecoratedpolygongraph{{(0,0)/1,(3.3333,2.6944)/2,(4,6.9282)/3,(4.6667,2.6944)/4,(8,0)/5,(4,1.5396)/6}}{{{1,2,3},{3,4,5},{5,6,1},{2,6,4}}}{ {{1,2},{2,3},{1,5},{5,6},{2,6},{4,2},{4,5},{5,3}} }{ {{3,1},{6,1},{3,4},{6,4}} }{0.6}{
        \node[draw,circle,fill=black](A) at (4,2.309) {};
        \node[draw,circle,fill=white](B) at (4,3.509) {};
        \node[draw,circle,fill=white](C) at (5.039,1.709) {};
        \node[draw,circle,fill=black](D) at (5.5,3.1745) {};
        \draw[black,line width=1pt] (A) -- (B);
        \draw[red,line width=1pt] (A) -- (C);
        \draw[red,line width=1pt] (B) -- (D);
        \draw[black,line width=1pt] (C) -- (D);
        } }}  \hspace{-2em} ,
        \tilde G(P_9) = \hspace{-3em} \vcenter{\hbox{ \dualdecoratedpolygongraph{{(0,0)/1,(3.3333,2.6944)/2,(4,6.9282)/3,(4.6667,2.6944)/4,(8,0)/5,(4,1.5396)/6}}{{{1,2,3},{3,4,5},{5,6,1},{2,6,4}}}{ {{1,2},{2,3},{1,5},{5,6},{2,6},{5,3},{6,4},{3,4}} }{ {{3,1},{6,1},{4,5},{4,2}} }{0.6}{
        \node[draw,circle,fill=black](A) at (4,2.309) {};
        \node[draw,circle,fill=white](B) at (4,3.509) {};
        \node[draw,circle,fill=white](C) at (5.039,1.709) {};
        \node[draw,circle,fill=black](D) at (5.5,3.1745) {};
        \draw[red,line width=1pt] (A) -- (B);
        \draw[black,line width=1pt] (A) -- (C);
        \draw[black,line width=1pt] (B) -- (D);
        \draw[red,line width=1pt] (C) -- (D);
        } }} \right\}
    \end{gather*}
    From this point forward, we consider only those posets that have non-empty sets of linear extensions, and the corresponding suitable decorations that are dual to a perfect matching.

    Consider an index $i$, such that in the graph $G(P_i)$ (with edges that are inside of black faces removed) the corresponding vertex is not on the same face as $1$. Suppose we find a valid decoration where all edges ingoing to $i$ are decorated. Then, in the poset $P$ corresponding to this decoration, one has $j \prec i ~ \forall ~ j$, from which we see that $S_n^{(P)}$ includes some element $1\cdots i$ where $i$ is the last index. The region $R(1, \cdots, i) \in R_P$ shares a codimension-1 boundary with the region $R(1,-i,\cdots)$. This latter region is included in $R_{P'}$, where $P'$ is the same poset as $P$ except instead of $i$ being a maximal element, it is the smallest element larger than $1$. This corresponds to the decoration where all outgoing edges to $i$ are decorated instead. So, to show the strong connectedness of $R_T$, we must show that the posets used to construct subsets $R_P$ are adjacent to each other by the move of `swapping all-ingoing decorations for all-outgoing decorations'.

    Luckily, this move has a neat interpretation in terms of the perfect matchings. Swapping the all-ingoing and all-outgoing decorations in $G(P_i)$ correspond to rotating a face in $\tilde G(P_i)$ when that face has alternating edges. It has been shown that for bipartite planar diagrams, all perfect matchings are related to each other by such moves \cite{Propp_LatticeStructure}.
\end{proof}


\bibliographystyle{JHEP}
\bibliography{main}

@article {LusTPFlag,
	AUTHOR = {Lusztig, G.},
	TITLE = {Total positivity in partial flag manifolds},
	JOURNAL = {Represent. Theory},
	FJOURNAL = {Representation Theory. An Electronic Journal of the American
	Mathematical Society},
	VOLUME = {2},
	YEAR = {1998},
	PAGES = {70--78},
}

@book {RietschCell,
	AUTHOR = {Rietsch, Konstanze Christina},
	TITLE = {Total Positivity and Real Flag Varieties},
	PUBLISHER = {Ph.D.\ thesis, Massachusetts Institute of Technology},
	YEAR = {1998},
	PAGES = {(no paging)},
	MRCLASS = {Thesis},
	MRNUMBER = {2716793},
	URL ={http://hdl.handle.net/1721.1/46139},
}

@article {GKL,
	AUTHOR = {Galashin, Pavel and Karp, Steven N. and Lam, Thomas},
	TITLE = {Regularity theorem for totally nonnegative flag varieties},
	JOURNAL = {J. Amer. Math. Soc.},
	FJOURNAL = {Journal of the American Mathematical Society},
	VOLUME = {35},
	YEAR = {2022},
	NUMBER = {2},
	PAGES = {513--579},
}

@article {Wil-enumeration,
    AUTHOR = {Williams, Lauren K.},
     TITLE = {Enumeration of totally positive {G}rassmann cells},
   JOURNAL = {Adv. Math.},
  FJOURNAL = {Advances in Mathematics},
    VOLUME = {190},
      YEAR = {2005},
    NUMBER = {2},
     PAGES = {319--342},
      ISSN = {0001-8708,1090-2082},
   MRCLASS = {05E15 (05A15 14M15)},
  MRNUMBER = {2102660},
MRREVIEWER = {Gregory\ S.\ Warrington},
       DOI = {10.1016/j.aim.2004.01.003},
       URL = {https://doi.org/10.1016/j.aim.2004.01.003},
}

@article{Cachazo:2023ltw,
    author = "Cachazo, Freddy and Early, Nick and Zhang, Yong",
    title = "{Generalized color orderings: CEGM integrands and decoupling identities}",
    eprint = "2304.07351",
    archivePrefix = "arXiv",
    primaryClass = "hep-th",
    doi = "10.1016/j.nuclphysb.2024.116552",
    journal = "Nucl. Phys. B",
    volume = "1004",
    pages = "116552",
    year = "2024"
}

@article{Postnikov:2006kva,
    author = "Postnikov, Alexander",
    title = "{Total positivity, Grassmannians, and networks}",
    eprint = "math/0609764",
    archivePrefix = "arXiv",
    month = "9",
    year = "2006"
}

@article{Bern:2015ple,
    author = "Bern, Zvi and Herrmann, Enrico and Litsey, Sean and Stankowicz, James and Trnka, Jaroslav",
    title = "{Evidence for a Nonplanar Amplituhedron}",
    eprint = "1512.08591",
    archivePrefix = "arXiv",
    primaryClass = "hep-th",
    reportNumber = "UCLA-15-TEP-103, CALT-TH-2015-064",
    doi = "10.1007/JHEP06(2016)098",
    journal = "JHEP",
    volume = "06",
    pages = "098",
    year = "2016"
}

@article{Arkani-Hamed:2009pfk,
    author = "Arkani-Hamed, Nima and Bourjaily, Jacob and Cachazo, Freddy and Trnka, Jaroslav",
    title = "{Local Spacetime Physics from the Grassmannian}",
    eprint = "0912.3249",
    archivePrefix = "arXiv",
    primaryClass = "hep-th",
    reportNumber = "PUPT-2326",
    doi = "10.1007/JHEP01(2011)108",
    journal = "JHEP",
    volume = "01",
    pages = "108",
    year = "2011"
}

@article{Arkani-Hamed:2009nll,
    author = "Arkani-Hamed, Nima and Cachazo, Freddy and Cheung, Clifford",
    title = "{The Grassmannian Origin Of Dual Superconformal Invariance}",
    eprint = "0909.0483",
    archivePrefix = "arXiv",
    primaryClass = "hep-th",
    doi = "10.1007/JHEP03(2010)036",
    journal = "JHEP",
    volume = "03",
    pages = "036",
    year = "2010"
}

@article{Britto:2004ap,
    author = "Britto, Ruth and Cachazo, Freddy and Feng, Bo",
    title = "{New recursion relations for tree amplitudes of gluons}",
    eprint = "hep-th/0412308",
    archivePrefix = "arXiv",
    doi = "10.1016/j.nuclphysb.2005.02.030",
    journal = "Nucl. Phys. B",
    volume = "715",
    pages = "499--522",
    year = "2005"
}

@article{Bourjaily:2013mma,
    author = "Bourjaily, Jacob L. and Caron-Huot, Simon and Trnka, Jaroslav",
    title = "{Dual-Conformal Regularization of Infrared Loop Divergences and the Chiral Box Expansion}",
    eprint = "1303.4734",
    archivePrefix = "arXiv",
    primaryClass = "hep-th",
    doi = "10.1007/JHEP01(2015)001",
    journal = "JHEP",
    volume = "01",
    pages = "001",
    year = "2015"
}

@article{Caron-Huot:2010ryg,
    author = "Caron-Huot, Simon",
    title = "{Notes on the scattering amplitude / Wilson loop duality}",
    eprint = "1010.1167",
    archivePrefix = "arXiv",
    primaryClass = "hep-th",
    doi = "10.1007/JHEP07(2011)058",
    journal = "JHEP",
    volume = "07",
    pages = "058",
    year = "2011"
}

@article{Bourjaily:2016mnp,
    author = "Bourjaily, Jacob L. and Franco, Sebastian and Galloni, Daniele and Wen, Congkao",
    title = "{Stratifying On-Shell Cluster Varieties: the Geometry of Non-Planar On-Shell Diagrams}",
    eprint = "1607.01781",
    archivePrefix = "arXiv",
    primaryClass = "hep-th",
    reportNumber = "CCNY-HEP-16-06",
    doi = "10.1007/JHEP10(2016)003",
    journal = "JHEP",
    volume = "10",
    pages = "003",
    year = "2016"
}

@article{Franco:2015rma,
    author = "Franco, Sebastian and Galloni, Daniele and Penante, Brenda and Wen, Congkao",
    title = "{Non-Planar On-Shell Diagrams}",
    eprint = "1502.02034",
    archivePrefix = "arXiv",
    primaryClass = "hep-th",
    doi = "10.1007/JHEP06(2015)199",
    journal = "JHEP",
    volume = "06",
    pages = "199",
    year = "2015"
}

@article{Engelund:2011fg,
    author = "Engelund, Oluf Tang and Roiban, Radu",
    title = "{On correlation functions of Wilson loops, local and non-local operators}",
    eprint = "1110.0758",
    archivePrefix = "arXiv",
    primaryClass = "hep-th",
    doi = "10.1007/JHEP05(2012)158",
    journal = "JHEP",
    volume = "05",
    pages = "158",
    year = "2012"
}

@article{Alday:2013ip,
    author = "Alday, Luis F. and Henn, Johannes M. and Sikorowski, Jakub",
    title = "{Higher loop mixed correlators in N=4 SYM}",
    eprint = "1301.0149",
    archivePrefix = "arXiv",
    primaryClass = "hep-th",
    doi = "10.1007/JHEP03(2013)058",
    journal = "JHEP",
    volume = "03",
    pages = "058",
    year = "2013"
}

@article{Henn:2023pkc,
    author = "Henn, Johannes M. and Lagares, Mart{\'\i}n and Zhang, Shun-Qing",
    title = "{Integrated negative geometries in ABJM}",
    eprint = "2303.02996",
    archivePrefix = "arXiv",
    primaryClass = "hep-th",
    reportNumber = "MPP-2023-41",
    doi = "10.1007/JHEP05(2023)112",
    journal = "JHEP",
    volume = "05",
    pages = "112",
    year = "2023"
}

@article{Lagares:2024epo,
    author = "Lagares, Mart{\'\i}n and Zhang, Shun-Qing",
    title = "{Higher-loop integrated negative geometries in ABJM}",
    eprint = "2402.17432",
    archivePrefix = "arXiv",
    primaryClass = "hep-th",
    reportNumber = "MPP-2024-37",
    doi = "10.1007/JHEP05(2024)142",
    journal = "JHEP",
    volume = "05",
    pages = "142",
    year = "2024"
}

@article{Glew:2024zoh,
    author = "Glew, Ross and Lukowski, Tomasz",
    title = "{Positive and negative ladders in loop space}",
    eprint = "2411.14989",
    archivePrefix = "arXiv",
    primaryClass = "hep-th",
    doi = "10.1007/JHEP06(2025)124",
    journal = "JHEP",
    volume = "06",
    pages = "124",
    year = "2025"
}

@article{Armstrong:2020ljm,
    author = "Armstrong, Connor and Farrow, Joseph A. and Lipstein, Arthur E.",
    title = "{$ \mathcal{N} $ = 7 On-shell diagrams and supergravity amplitudes in momentum twistor space}",
    eprint = "2010.11813",
    archivePrefix = "arXiv",
    primaryClass = "hep-th",
    doi = "10.1007/JHEP01(2021)181",
    journal = "JHEP",
    volume = "01",
    pages = "181",
    year = "2021"
}

@article{Heslop:2016plj,
    author = "Heslop, Paul and Lipstein, Arthur E.",
    title = "{On-shell diagrams for $ \mathcal{N} $ = 8 supergravity amplitudes}",
    eprint = "1604.03046",
    archivePrefix = "arXiv",
    primaryClass = "hep-th",
    doi = "10.1007/JHEP06(2016)069",
    journal = "JHEP",
    volume = "06",
    pages = "069",
    year = "2016"
}

@article{Paranjape:2023qsq,
    author = "Paranjape, Shruti and Trnka, Jaroslav",
    title = "{Gravity Amplitudes from Double Bonus Relations}",
    eprint = "2309.05710",
    archivePrefix = "arXiv",
    primaryClass = "hep-th",
    doi = "10.1103/PhysRevLett.131.251601",
    journal = "Phys. Rev. Lett.",
    volume = "131",
    number = "25",
    pages = "251601",
    year = "2023"
}

@article{Herrmann:2016qea,
    author = "Herrmann, Enrico and Trnka, Jaroslav",
    title = "{Gravity On-shell Diagrams}",
    eprint = "1604.03479",
    archivePrefix = "arXiv",
    primaryClass = "hep-th",
    reportNumber = "CALT-TH-2016-006",
    doi = "10.1007/JHEP11(2016)136",
    journal = "JHEP",
    volume = "11",
    pages = "136",
    year = "2016"
}

@article{Bourjaily:2018omh,
    author = "Bourjaily, Jacob L. and Herrmann, Enrico and Trnka, Jaroslav",
    title = "{Maximally supersymmetric amplitudes at infinite loop momentum}",
    eprint = "1812.11185",
    archivePrefix = "arXiv",
    primaryClass = "hep-th",
    doi = "10.1103/PhysRevD.99.066006",
    journal = "Phys. Rev. D",
    volume = "99",
    number = "6",
    pages = "066006",
    year = "2019"
}

@article{Bourjaily:2019iqr,
    author = "Bourjaily, Jacob L. and Herrmann, Enrico and Langer, Cameron and McLeod, Andrew J. and Trnka, Jaroslav",
    title = "{Prescriptive Unitarity for Non-Planar Six-Particle Amplitudes at Two Loops}",
    eprint = "1909.09131",
    archivePrefix = "arXiv",
    primaryClass = "hep-th",
    doi = "10.1007/JHEP12(2019)073",
    journal = "JHEP",
    volume = "12",
    pages = "073",
    year = "2019"
}

@article{Bourjaily:2019gqu,
    author = "Bourjaily, Jacob L. and Herrmann, Enrico and Langer, Cameron and McLeod, Andrew J. and Trnka, Jaroslav",
    title = "{All-Multiplicity Nonplanar Amplitude Integrands in Maximally Supersymmetric Yang-Mills Theory at Two Loops}",
    eprint = "1911.09106",
    archivePrefix = "arXiv",
    primaryClass = "hep-th",
    doi = "10.1103/PhysRevLett.124.111603",
    journal = "Phys. Rev. Lett.",
    volume = "124",
    number = "11",
    pages = "111603",
    year = "2020"
}

@article{Paranjape:2022ymg,
    author = "Paranjape, Shruti and Trnka, Jaroslav and Zheng, Minshan",
    title = "{Non-planar BCFW Grassmannian geometries}",
    eprint = "2208.02262",
    archivePrefix = "arXiv",
    primaryClass = "hep-th",
    doi = "10.1007/JHEP12(2022)084",
    journal = "JHEP",
    volume = "12",
    pages = "084",
    year = "2022"
}

@article{Arkani-Hamed:2023epq,
    author = "Arkani-Hamed, Nima and Flieger, Wojciech and Henn, Johannes M. and Schreiber, Anders and Trnka, Jaroslav",
    title = "{Coulomb Branch Amplitudes from a Deformed Amplituhedron Geometry}",
    eprint = "2311.10814",
    archivePrefix = "arXiv",
    primaryClass = "hep-th",
    doi = "10.1103/PhysRevLett.132.211601",
    journal = "Phys. Rev. Lett.",
    volume = "132",
    number = "21",
    pages = "211601",
    year = "2024"
}

@article{Chicherin:2024hes,
    author = "Chicherin, Dmitry and Henn, Johannes and Trnka, Jaroslav and Zhang, Shun-Qing",
    title = "{Positivity properties of five-point two-loop Wilson loops with Lagrangian insertion}",
    eprint = "2410.11456",
    archivePrefix = "arXiv",
    primaryClass = "hep-th",
    reportNumber = "MPP-2024-194,LAPTH-052/24",
    doi = "10.1007/JHEP04(2025)022",
    journal = "JHEP",
    volume = "04",
    pages = "022",
    year = "2025"
}

@article{Brown:2025plq,
    author = "Brown, Taro V. and Henn, Johannes M. and Mazzucchelli, Elia and Trnka, Jaroslav",
    title = "{All-loop Leading Singularities of Wilson Loops}",
    eprint = "2503.17185",
    archivePrefix = "arXiv",
    primaryClass = "hep-th",
    reportNumber = "MPP-2025-44",
    month = "3",
    year = "2025"
}

@article{Chicherin:2025cua,
    author = "Chicherin, Dmitry and Henn, Johannes and Mazzucchelli, Elia and Trnka, Jaroslav and Yang, Qinglin and Zhang, Shun-Qing",
    title = "{Geometric Landau Analysis and Symbol Bootstrap}",
    eprint = "2508.05443",
    archivePrefix = "arXiv",
    primaryClass = "hep-th",
    reportNumber = "MPP-2025-144, LAPTH-023/25",
    month = "8",
    year = "2025"
}

@article{Karp:2017ouj,
    author = "Karp, Steven N. and Williams, Lauren K. and Zhang, Yan X.",
    title = "{Decompositions of amplituhedra}",
    eprint = "1708.09525",
    archivePrefix = "arXiv",
    primaryClass = "math.CO",
    doi = "10.4171/aihpd/87",
    journal = "Ann. Inst. H. Poincare D Comb. Phys. Interact.",
    volume = "7",
    number = "3",
    pages = "303--363",
    year = "2020"
}

@article{Ranestad:2024svp,
    author = "Ranestad, Kristian and Sinn, Rainer and Telen, Simon",
    title = "{Adjoints and Canonical Forms of Tree Amplituhedra}",
    eprint = "2402.06527",
    archivePrefix = "arXiv",
    primaryClass = "math.AG",
    doi = "10.7146/math.scand.a-149816",
    journal = "Math. Scand.",
    volume = "130",
    pages = "433--466",
    year = "2024"
}

@article{Mazzucchelli:2025kzy,
    author = "Mazzucchelli, Elia and Pratt, Elizabeth",
    title = "{Exterior Cyclic Polytopes and Convexity of Amplituhedra}",
    eprint = "2507.17620",
    archivePrefix = "arXiv",
    primaryClass = "math.CO",
    month = "7",
    year = "2025"
}

@article{Galashin:2024ttp,
    author = "Galashin, Pavel",
    title = "{Amplituhedra and origami}",
    eprint = "2410.09574",
    archivePrefix = "arXiv",
    primaryClass = "hep-th",
    month = "10",
    year = "2024"
}

@article{Ferro:2022abq,
    author = "Ferro, Livia and Lukowski, Tomasz",
    title = "{The Loop Momentum Amplituhedron}",
    eprint = "2210.01127",
    archivePrefix = "arXiv",
    primaryClass = "hep-th",
    doi = "10.1007/JHEP05(2023)183",
    journal = "JHEP",
    volume = "05",
    pages = "183",
    year = "2023"
}

@article{Akhmedova:2023wcf,
    author = "Akhmedova, Evgeniya and Tessler, Ran J.",
    title = "{The Tropical Amplituhedron}",
    eprint = "2312.12319",
    archivePrefix = "arXiv",
    primaryClass = "hep-th",
    month = "12",
    year = "2023"
}

@article{Lam:2024gyg,
    author = "Lam, Thomas",
    title = "{On the face stratification of the $m=2$ amplituhedron}",
    eprint = "2403.06948",
    archivePrefix = "arXiv",
    primaryClass = "math.CO",
    month = "3",
    year = "2024"
}

@article{Even-Zohar:2025ydi,
    author = "Even-Zohar, Chaim and Lakrec, Tsviqa and Parisi, Matteo and Sherman-Bennett, Melissa and Tessler, Ran and Williams, Lauren",
    title = "{BCFW tilings and cluster adjacency for the amplituhedron}",
    eprint = "2504.01217",
    archivePrefix = "arXiv",
    primaryClass = "math.CO",
    doi = "10.1073/pnas.2408572122",
    journal = "Proc. Nat. Acad. Sci.",
    volume = "122",
    number = "12",
    pages = "e2408572122",
    year = "2025"
}

@article{Parisi:2024psm,
    author = "Parisi, Matteo and Sherman-Bennett, Melissa and Tessler, Ran and Williams, Lauren",
    title = "{The Magic Number Conjecture for the $m=2$ amplituhedron and Parke-Taylor identities}",
    eprint = "2404.03026",
    archivePrefix = "arXiv",
    primaryClass = "math.CO",
    month = "4",
    year = "2024"
}

@article{Even-Zohar:2024nvw,
    author = "Even-Zohar, Chaim and Lakrec, Tsviqa and Parisi, Matteo and Sherman-Bennett, Melissa and Tessler, Ran and Williams, Lauren",
    title = "{A cluster of results on amplituhedron tiles}",
    eprint = "2402.15568",
    archivePrefix = "arXiv",
    primaryClass = "math.CO",
    doi = "10.1007/s11005-024-01854-4",
    journal = "Lett. Math. Phys.",
    volume = "114",
    number = "5",
    pages = "111",
    year = "2024"
}

@article{Brown:2023mqi,
    author = "Brown, Taro V. and Oktem, Umut and Paranjape, Shruti and Trnka, Jaroslav",
    title = "{Loops of loops expansion in the amplituhedron}",
    eprint = "2312.17736",
    archivePrefix = "arXiv",
    primaryClass = "hep-th",
    doi = "10.1007/JHEP07(2024)025",
    journal = "JHEP",
    volume = "07",
    pages = "025",
    year = "2024"
}

@article{Even-Zohar:2023del,
    author = "Even-Zohar, Chaim and Lakrec, Tsviqa and Parisi, Matteo and Tessler, Ran and Sherman-Bennett, Melissa and Williams, Lauren",
    title = "{Cluster algebras and tilings for the m=4 amplituhedron}",
    eprint = "2310.17727",
    archivePrefix = "arXiv",
    primaryClass = "math.CO",
    month = "10",
    year = "2023"
}

@article{He:2023rou,
    author = "He, Song and Huang, Yu-tin and Kuo, Chia-Kai",
    title = "{The ABJM Amplituhedron}",
    eprint = "2306.00951",
    archivePrefix = "arXiv",
    primaryClass = "hep-th",
    doi = "10.1007/JHEP09(2023)165",
    journal = "JHEP",
    volume = "09",
    pages = "165",
    year = "2023",
    note = "[Erratum: JHEP 04, 064 (2024)]"
}

@article{Dian:2022tpf,
    author = "Dian, Gabriele and Heslop, Paul and Stewart, Alastair",
    title = "{Internal boundaries of the loop amplituhedron}",
    eprint = "2207.12464",
    archivePrefix = "arXiv",
    primaryClass = "hep-th",
    doi = "10.21468/SciPostPhys.15.3.098",
    journal = "SciPost Phys.",
    volume = "15",
    number = "3",
    pages = "098",
    year = "2023"
}

@article{Arkani-Hamed:2021iya,
    author = "Arkani-Hamed, Nima and Henn, Johannes and Trnka, Jaroslav",
    title = "{Nonperturbative negative geometries: amplitudes at strong coupling and the amplituhedron}",
    eprint = "2112.06956",
    archivePrefix = "arXiv",
    primaryClass = "hep-th",
    reportNumber = "MPP-2021-199",
    doi = "10.1007/JHEP03(2022)108",
    journal = "JHEP",
    volume = "03",
    pages = "108",
    year = "2022"
}

@article{Even-Zohar:2021sec,
    author = "Even-Zohar, Chaim and Lakrec, Tsviqa and Tessler, Ran J.",
    title = "{The amplituhedron BCFW triangulation}",
    eprint = "2112.02703",
    archivePrefix = "arXiv",
    primaryClass = "math-ph",
    doi = "10.1007/s00222-025-01316-1",
    journal = "Invent. Math.",
    volume = "239",
    number = "3",
    pages = "1009--1138",
    year = "2025"
}

@article{Parisi:2021oql,
    author = "Parisi, Matteo and Sherman-Bennett, Melissa and Williams, Lauren",
    title = "{The {\ensuremath{\mathit{m}}}=2 amplituhedron and the hypersimplex: Signs, clusters, tilings, Eulerian numbers}",
    eprint = "2104.08254",
    archivePrefix = "arXiv",
    primaryClass = "math.CO",
    doi = "10.1090/cams/23",
    journal = "Commun. Am. Math. Soc.",
    volume = "3",
    number = "7",
    pages = "329--399",
    year = "2023"
}

@article{Lukowski:2020dpn,
    author = "Lukowski, Tomasz and Parisi, Matteo and Williams, Lauren K.",
    title = "{The Positive Tropical Grassmannian, the Hypersimplex, and the m = 2 Amplituhedron}",
    eprint = "2002.06164",
    archivePrefix = "arXiv",
    primaryClass = "math.CO",
    doi = "10.1093/imrn/rnad010",
    journal = "Int. Math. Res. Not.",
    volume = "2023",
    number = "19",
    pages = "16778--16836",
    year = "2023"
}

@article{Arkani-Hamed:2018rsk,
    author = "Arkani-Hamed, Nima and Langer, Cameron and Yelleshpur Srikant, Akshay and Trnka, Jaroslav",
    title = "{Deep Into the Amplituhedron: Amplitude Singularities at All Loops and Legs}",
    eprint = "1810.08208",
    archivePrefix = "arXiv",
    primaryClass = "hep-th",
    doi = "10.1103/PhysRevLett.122.051601",
    journal = "Phys. Rev. Lett.",
    volume = "122",
    number = "5",
    pages = "051601",
    year = "2019"
}

@article{Karp:2016uax,
    author = "Karp, Steven N. and Williams, Lauren K.",
    title = "{The m=1 amplituhedron and cyclic hyperplane arrangements}",
    eprint = "1608.08288",
    archivePrefix = "arXiv",
    primaryClass = "math.CO",
    doi = "10.1093/imrn/rnx140",
    journal = "Int. Math. Res. Not.",
    volume = "5",
    pages = "1401--1462",
    year = "2019"
}

@article{Dennen:2016mdk,
    author = "Dennen, Tristan and Prlina, Igor and Spradlin, Marcus and Stanojevic, Stefan and Volovich, Anastasia",
    title = "{Landau Singularities from the Amplituhedron}",
    eprint = "1612.02708",
    archivePrefix = "arXiv",
    primaryClass = "hep-th",
    reportNumber = "BROWN-HET-1701",
    doi = "10.1007/JHEP06(2017)152",
    journal = "JHEP",
    volume = "06",
    pages = "152",
    year = "2017"
}

@article{Ferro:2016zmx,
    author = "Ferro, Livia and {\L}ukowski, Tomasz and Orta, Andrea and Parisi, Matteo",
    title = "{Yangian symmetry for the tree amplituhedron}",
    eprint = "1612.04378",
    archivePrefix = "arXiv",
    primaryClass = "hep-th",
    reportNumber = "LMU-ASC-63-16, QMUL-PH-16-21",
    doi = "10.1088/1751-8121/aa7594",
    journal = "J. Phys. A",
    volume = "50",
    number = "29",
    pages = "294005",
    year = "2017"
}

@article{Mason:2009qx,
    author = "Mason, L. J. and Skinner, David",
    title = "{Dual Superconformal Invariance, Momentum Twistors and Grassmannians}",
    eprint = "0909.0250",
    archivePrefix = "arXiv",
    primaryClass = "hep-th",
    doi = "10.1088/1126-6708/2009/11/045",
    journal = "JHEP",
    volume = "11",
    pages = "045",
    year = "2009"
}

@article{Caron-Huot:2011dec,
    author = "Caron-Huot, Simon and He, Song",
    title = "{Jumpstarting the All-Loop S-Matrix of Planar N=4 Super Yang-Mills}",
    eprint = "1112.1060",
    archivePrefix = "arXiv",
    primaryClass = "hep-th",
    doi = "10.1007/JHEP07(2012)174",
    journal = "JHEP",
    volume = "07",
    pages = "174",
    year = "2012"
}

@article{He:2020vob,
    author = "He, Song and Li, Zhenjie and Zhang, Chi",
    title = "{The symbol and alphabet of two-loop NMHV amplitudes from $\bar{Q}$ equations}",
    eprint = "2009.11471",
    archivePrefix = "arXiv",
    primaryClass = "hep-th",
    doi = "10.1007/JHEP03(2021)278",
    journal = "JHEP",
    volume = "03",
    pages = "278",
    year = "2021"
}

@article{He:2024fij,
    author = "He, Song and Jiang, Xuhang and Liu, Jiahao and Yang, Qinglin",
    title = "{Landau-based Schubert analysis}",
    eprint = "2410.11423",
    archivePrefix = "arXiv",
    primaryClass = "hep-th",
    doi = "10.1007/JHEP11(2025)053",
    journal = "JHEP",
    volume = "11",
    pages = "053",
    year = "2025"
}

@article{He:2025tyv,
    author = "He, Song and Jiang, Xuhang and Li, Xiang and Liu, Jiahao",
    title = "{Heptagon Symbols at Five Loops and All-Loop Sequences}",
    eprint = "2511.09669",
    archivePrefix = "arXiv",
    primaryClass = "hep-th",
    month = "11",
    year = "2025"
}

@article{Golden:2013xva,
    author = "Golden, John and Goncharov, Alexander B. and Spradlin, Marcus and Vergu, Cristian and Volovich, Anastasia",
    title = "{Motivic Amplitudes and Cluster Coordinates}",
    eprint = "1305.1617",
    archivePrefix = "arXiv",
    primaryClass = "hep-th",
    doi = "10.1007/JHEP01(2014)091",
    journal = "JHEP",
    volume = "01",
    pages = "091",
    year = "2014"
}

@article{Mago:2020kmp,
    author = "Mago, Jorge and Schreiber, Anders and Spradlin, Marcus and Volovich, Anastasia",
    title = "{Symbol alphabets from plabic graphs}",
    eprint = "2007.00646",
    archivePrefix = "arXiv",
    primaryClass = "hep-th",
    doi = "10.1007/JHEP10(2020)128",
    journal = "JHEP",
    volume = "10",
    pages = "128",
    year = "2020"
}

@article{Drummond:2018dfd,
    author = {Drummond, James and Foster, Jack and G{\"u}rdo{\u{g}}an, {\"O}mer},
    title = "{Cluster adjacency beyond MHV}",
    eprint = "1810.08149",
    archivePrefix = "arXiv",
    primaryClass = "hep-th",
    doi = "10.1007/JHEP03(2019)086",
    journal = "JHEP",
    volume = "03",
    pages = "086",
    year = "2019"
}

@article{Drummond:2018caf,
    author = {Drummond, James and Foster, Jack and G{\"u}rdo{\u{g}}an, {\"O}mer and Papathanasiou, Georgios},
    title = "{Cluster adjacency and the four-loop NMHV heptagon}",
    eprint = "1812.04640",
    archivePrefix = "arXiv",
    primaryClass = "hep-th",
    reportNumber = "DESY-18-214",
    doi = "10.1007/JHEP03(2019)087",
    journal = "JHEP",
    volume = "03",
    pages = "087",
    year = "2019"
}

@article{Dixon:2015iva,
    author = "Dixon, Lance J. and von Hippel, Matt and McLeod, Andrew J.",
    title = "{The four-loop six-gluon NMHV ratio function}",
    eprint = "1509.08127",
    archivePrefix = "arXiv",
    primaryClass = "hep-th",
    reportNumber = "SLAC-PUB-16352, CALT-2015-049",
    doi = "10.1007/JHEP01(2016)053",
    journal = "JHEP",
    volume = "01",
    pages = "053",
    year = "2016"
}

@article{Caron-Huot:2016owq,
    author = "Caron-Huot, Simon and Dixon, Lance J. and McLeod, Andrew and von Hippel, Matt",
    title = "{Bootstrapping a Five-Loop Amplitude Using Steinmann Relations}",
    eprint = "1609.00669",
    archivePrefix = "arXiv",
    primaryClass = "hep-th",
    reportNumber = "SLAC-PUB-16811",
    doi = "10.1103/PhysRevLett.117.241601",
    journal = "Phys. Rev. Lett.",
    volume = "117",
    number = "24",
    pages = "241601",
    year = "2016"
}

@article{Caron-Huot:2019vjl,
    author = "Caron-Huot, Simon and Dixon, Lance J. and Dulat, Falko and von Hippel, Matt and McLeod, Andrew J. and Papathanasiou, Georgios",
    title = "{Six-Gluon amplitudes in planar $ \mathcal{N} $ = 4 super-Yang-Mills theory at six and seven loops}",
    eprint = "1903.10890",
    archivePrefix = "arXiv",
    primaryClass = "hep-th",
    reportNumber = "DESY 19-042, HU-EP-19/04, DESY-19-042, HU-EP-19-04, SLAC-PUB-17413",
    doi = "10.1007/JHEP08(2019)016",
    journal = "JHEP",
    volume = "08",
    pages = "016",
    year = "2019"
}

@article{Basso:2024hlx,
    author = "Basso, Benjamin and Dixon, Lance J. and Tumanov, Alexander G.",
    title = "{The three-point form factor of Tr {\ensuremath{\phi}}$^{3}$ to six loops}",
    eprint = "2410.22402",
    archivePrefix = "arXiv",
    primaryClass = "hep-th",
    doi = "10.1007/JHEP02(2025)034",
    journal = "JHEP",
    volume = "02",
    pages = "034",
    year = "2025"
}

@article{Dixon:2021tdw,
    author = "Dixon, Lance J. and Gurdogan, Omer and McLeod, Andrew J. and Wilhelm, Matthias",
    title = "{Folding Amplitudes into Form Factors: An Antipodal Duality}",
    eprint = "2112.06243",
    archivePrefix = "arXiv",
    primaryClass = "hep-th",
    reportNumber = "SLAC-PUB-17637",
    doi = "10.1103/PhysRevLett.128.111602",
    journal = "Phys. Rev. Lett.",
    volume = "128",
    number = "11",
    pages = "111602",
    year = "2022"
}

@article{Dixon:2022xqh,
    author = {Dixon, Lance J. and G{\"u}rdo{\u{g}}an, {\"O}mer and Liu, Yu-Ting and McLeod, Andrew J. and Wilhelm, Matthias},
    title = "{Antipodal Self-Duality for a Four-Particle Form Factor}",
    eprint = "2212.02410",
    archivePrefix = "arXiv",
    primaryClass = "hep-th",
    reportNumber = "CERN-TH-2022-190, SLAC-PUB-17711",
    doi = "10.1103/PhysRevLett.130.111601",
    journal = "Phys. Rev. Lett.",
    volume = "130",
    number = "11",
    pages = "111601",
    year = "2023"
}

@article{Dixon:2023kop,
    author = "Dixon, Lance J. and Liu, Yu-Ting",
    title = "{An eight loop amplitude via antipodal duality}",
    eprint = "2308.08199",
    archivePrefix = "arXiv",
    primaryClass = "hep-th",
    reportNumber = "SLAC-PUB-17693",
    doi = "10.1007/JHEP09(2023)098",
    journal = "JHEP",
    volume = "09",
    pages = "098",
    year = "2023"
}

@article{Dixon:2025zwj,
    author = "Dixon, Lance J. and Duhr, Claude",
    title = "{Antipodal self-duality of square fishnet graphs}",
    eprint = "2502.00862",
    archivePrefix = "arXiv",
    primaryClass = "hep-th",
    reportNumber = "BONN-TH-2025-02",
    doi = "10.1103/PhysRevD.111.L101901",
    journal = "Phys. Rev. D",
    volume = "111",
    number = "10",
    pages = "L101901",
    year = "2025"
}

@misc{Propp_LatticeStructure,
      title={Lattice structure for orientations of graphs}, 
      author={James Propp},
      year={2025},
      eprint={math/0209005},
      archivePrefix={arXiv},
      primaryClass={math.CO},
      url={https://arxiv.org/abs/math/0209005}, 
}

@article{Early:2024asu,
    author = "Early, Nick and Antolini, Dario",
    title = "{The chirotropical Grassmannian}",
    eprint = "2411.07293",
    archivePrefix = "arXiv",
    primaryClass = "math.CO",
    doi = "10.4418/2025.80.1.4",
    journal = "Matematiche",
    volume = "80",
    number = "1",
    pages = "123--142",
    year = "2025"
}

@article{Arkani_Hamed_2017,
   title={Positive geometries and canonical forms},
   volume={2017},
   ISSN={1029-8479},
   url={http://dx.doi.org/10.1007/JHEP11(2017)039},
   DOI={10.1007/jhep11(2017)039},
   number={11},
   journal={Journal of High Energy Physics},
   publisher={Springer Science and Business Media LLC},
   author={Arkani-Hamed, Nima and Bai, Yuntao and Lam, Thomas},
   year={2017},
   month=nov }

@article{Castravet:2013,
url = {https://doi.org/10.1515/CRELLE.2011.189},
title = {Hypertrees, projections, and moduli of stable rational curves},
author = {Ana-Maria Castravet and Jenia Tevelev},
pages = {121--180},
volume = {2013},
number = {675},
journal = {Journal für die reine und angewandte Mathematik (Crelles Journal)},
doi = {doi:10.1515/CRELLE.2011.189},
year = {2013},
lastchecked = {2025-06-10}
}

@article{Cachazo:2019,
    author = "Cachazo, Freddy and Early, Nick and Guevara, Alfredo and Mizera, Sebastian",
    title = "{$\Delta$-algebra and scattering amplitudes}",
    eprint = "1812.01168",
    archivePrefix = "arXiv",
    primaryClass = "hep-th",
    doi = "10.1007/JHEP02(2019)005",
    journal = "JHEP",
    volume = "02",
    pages = "005",
    year = "2019"
}

@article{Bourjaily:2023ycy,
author = "Bourjaily, Jacob L. and Kalyanapuram, Nikhil and Patatoukos, Kokkimidis and Plesser, Michael and Zhang, Yaqi",
title = "{Gauge-Invariant Double Copies via Recursion Relations}",
eprint = "2307.02542",
archivePrefix = "arXiv",
primaryClass = "hep-th",
doi = "10.1103/PhysRevLett.131.191601",
journal = "Phys. Rev. Lett.",
volume = "131",
number = "19",
pages = "191601",
year = "2023"
}

@article{Arkani-Hamed:2009ljj,
    author = "Arkani-Hamed, Nima and Cachazo, Freddy and Cheung, Clifford and Kaplan, Jared",
    title = "{A Duality For The S Matrix}",
    eprint = "0907.5418",
    archivePrefix = "arXiv",
    primaryClass = "hep-th",
    doi = "10.1007/JHEP03(2010)020",
    journal = "JHEP",
    volume = "03",
    pages = "020",
    year = "2010"
}

@article{Arkani-Hamed:2009kmp,
    author = "Arkani-Hamed, Nima and Bourjaily, Jacob and Cachazo, Freddy and Trnka, Jaroslav",
    title = "{Unification of Residues and Grassmannian Dualities}",
    eprint = "0912.4912",
    archivePrefix = "arXiv",
    primaryClass = "hep-th",
    doi = "10.1007/JHEP01(2011)049",
    journal = "JHEP",
    volume = "01",
    pages = "049",
    year = "2011"
}

@article{Prlina:2017tvx,
    author = "Prlina, Igor and Spradlin, Marcus and Stankowicz, James and Stanojevic, Stefan",
    title = "{Boundaries of Amplituhedra and NMHV Symbol Alphabets at Two Loops}",
    eprint = "1712.08049",
    archivePrefix = "arXiv",
    primaryClass = "hep-th",
    doi = "10.1007/JHEP04(2018)049",
    journal = "JHEP",
    volume = "04",
    pages = "049",
    year = "2018"
}

@article{Arkani-Hamed:2014dca,
    author = "Arkani-Hamed, Nima and Hodges, Andrew and Trnka, Jaroslav",
    title = "{Positive Amplitudes In The Amplituhedron}",
    eprint = "1412.8478",
    archivePrefix = "arXiv",
    primaryClass = "hep-th",
    reportNumber = "CALT-TH-2014-168",
    doi = "10.1007/JHEP08(2015)030",
    journal = "JHEP",
    volume = "08",
    pages = "030",
    year = "2015"
}

@article{Drummond:2009fd,
    author = "Drummond, James M. and Henn, Johannes M. and Plefka, Jan",
    editor = "Liu, Feng and Xiao, Zhigang and Zhuang, Pengfei",
    title = "{Yangian symmetry of scattering amplitudes in N=4 super Yang-Mills theory}",
    eprint = "0902.2987",
    archivePrefix = "arXiv",
    primaryClass = "hep-th",
    reportNumber = "HU-EP-09-06, LAPTH-1308-09",
    doi = "10.1088/1126-6708/2009/05/046",
    journal = "JHEP",
    volume = "05",
    pages = "046",
    year = "2009"
}

@article{Drummond:2008vq,
    author = "Drummond, J. M. and Henn, J. and Korchemsky, G. P. and Sokatchev, E.",
    title = "{Dual superconformal symmetry of scattering amplitudes in N=4 super-Yang-Mills theory}",
    eprint = "0807.1095",
    archivePrefix = "arXiv",
    primaryClass = "hep-th",
    reportNumber = "LAPTH-1257-08, LPT-ORSAY-08-60",
    doi = "10.1016/j.nuclphysb.2009.11.022",
    journal = "Nucl. Phys. B",
    volume = "828",
    pages = "317--374",
    year = "2010"
}

@article{Herrmann:2020qlt,
    author = "Herrmann, Enrico and Langer, Cameron and Trnka, Jaroslav and Zheng, Minshan",
    title = "{Positive geometry, local triangulations, and the dual of the Amplituhedron}",
    eprint = "2009.05607",
    archivePrefix = "arXiv",
    primaryClass = "hep-th",
    doi = "10.1007/JHEP01(2021)035",
    journal = "JHEP",
    volume = "01",
    pages = "035",
    year = "2021"
}

@article{Arkani-Hamed:2017vfh,
    author = "Arkani-Hamed, Nima and Thomas, Hugh and Trnka, Jaroslav",
    title = "{Unwinding the Amplituhedron in Binary}",
    eprint = "1704.05069",
    archivePrefix = "arXiv",
    primaryClass = "hep-th",
    doi = "10.1007/JHEP01(2018)016",
    journal = "JHEP",
    volume = "01",
    pages = "016",
    year = "2018"
}

@article{Arkani-Hamed:2013jha,
    author = "Arkani-Hamed, Nima and Trnka, Jaroslav",
    title = "{The Amplituhedron}",
    eprint = "1312.2007",
    archivePrefix = "arXiv",
    primaryClass = "hep-th",
    doi = "10.1007/JHEP10(2014)030",
    journal = "JHEP",
    volume = "10",
    pages = "030",
    year = "2014"
}

@article{Alday:2007hr,
    author = "Alday, Luis F. and Maldacena, Juan Martin",
    title = "{Gluon scattering amplitudes at strong coupling}",
    eprint = "0705.0303",
    archivePrefix = "arXiv",
    primaryClass = "hep-th",
    reportNumber = "SPIN-07-16, ITP-UU-07-24",
    doi = "10.1088/1126-6708/2007/06/064",
    journal = "JHEP",
    volume = "06",
    pages = "064",
    year = "2007"
}

@article{Arkani-Hamed:2014bca,
    author = "Arkani-Hamed, Nima and Bourjaily, Jacob L. and Cachazo, Freddy and Postnikov, Alexander and Trnka, Jaroslav",
    title = "{On-Shell Structures of MHV Amplitudes Beyond the Planar Limit}",
    eprint = "1412.8475",
    archivePrefix = "arXiv",
    primaryClass = "hep-th",
    reportNumber = "CALT-TH-2014-169",
    doi = "10.1007/JHEP06(2015)179",
    journal = "JHEP",
    volume = "06",
    pages = "179",
    year = "2015"
}

@article{Goncharov:2010jf,
	archiveprefix = {arXiv},
	author = {Goncharov, Alexander B. and Spradlin, Marcus and Vergu, C. and Volovich, Anastasia},
	doi = {10.1103/PhysRevLett.105.151605},
	eprint = {1006.5703},
	journal = {Phys. Rev. Lett.},
	pages = {151605},
	primaryclass = {hep-th},
	reportnumber = {BROWN-HET-1602},
	title = {{Classical Polylogarithms for Amplitudes and Wilson Loops}},
	volume = {105},
	year = {2010}}

@article{Arkani-Hamed:2013kca,
    author = "Arkani-Hamed, Nima and Trnka, Jaroslav",
    title = "{Into the Amplituhedron}",
    eprint = "1312.7878",
    archivePrefix = "arXiv",
    primaryClass = "hep-th",
    reportNumber = "CALT-68-2876",
    doi = "10.1007/JHEP12(2014)182",
    journal = "JHEP",
    volume = "12",
    pages = "182",
    year = "2014"
}

@article{Galashin:2018fri,
    author = "Galashin, Pavel and Lam, Thomas",
    title = "{Parity duality for the amplituhedron}",
    eprint = "1805.00600",
    archivePrefix = "arXiv",
    primaryClass = "math.CO",
    doi = "10.1112/S0010437X20007411",
    journal = "Compos. Math.",
    volume = "156",
    number = "11",
    pages = "2207--2262",
    year = "2020"
}

@book{Arkani-Hamed:2012zlh,
    author = "Arkani-Hamed, Nima and Bourjaily, Jacob L. and Cachazo, Freddy and Goncharov, Alexander B. and Postnikov, Alexander and Trnka, Jaroslav",
    title = "{Grassmannian Geometry of Scattering Amplitudes}",
    eprint = "1212.5605",
    archivePrefix = "arXiv",
    primaryClass = "hep-th",
    reportNumber = "PUPT-2435",
    doi = "10.1017/CBO9781316091548",
    isbn = "978-1-107-08658-6, 978-1-316-57296-2",
    publisher = "Cambridge University Press",
    month = "4",
    year = "2016"
}

@article{Arkani-Hamed:2010zjl,
    author = "Arkani-Hamed, Nima and Bourjaily, Jacob L. and Cachazo, Freddy and Caron-Huot, Simon and Trnka, Jaroslav",
    title = "{The All-Loop Integrand For Scattering Amplitudes in Planar N=4 SYM}",
    eprint = "1008.2958",
    archivePrefix = "arXiv",
    primaryClass = "hep-th",
    doi = "10.1007/JHEP01(2011)041",
    journal = "JHEP",
    volume = "01",
    pages = "041",
    year = "2011"
}

@article{Dixon:2011nj,
    author = "Dixon, Lance J. and Drummond, James M. and Henn, Johannes M.",
    title = "{Analytic result for the two-loop six-point NMHV amplitude in N=4 super Yang-Mills theory}",
    eprint = "1111.1704",
    archivePrefix = "arXiv",
    primaryClass = "hep-th",
    reportNumber = "CERN-PH-TH-2011-251, SLAC-PUB-14632, LAPTH-042-11, HU-EP-11-44",
    doi = "10.1007/JHEP01(2012)024",
    journal = "JHEP",
    volume = "01",
    pages = "024",
    year = "2012"
}

@article{Arkani-Hamed:2022rwr,
    author = "Arkani-Hamed, Nima and Dixon, Lance J. and McLeod, Andrew J. and Spradlin, Marcus and Trnka, Jaroslav and Volovich, Anastasia",
    title = "{Solving Scattering in $N$ = 4 Super-Yang-Mills Theory}",
    eprint = "2207.10636",
    archivePrefix = "arXiv",
    primaryClass = "hep-th",
    reportNumber = "CERN-TH-2022-123, SLAC-PUB-17692",
    month = "7",
    year = "2022"
}

@article{Franco:2014csa,
    author = "Franco, Sebastian and Galloni, Daniele and Mariotti, Alberto and Trnka, Jaroslav",
    title = "{Anatomy of the Amplituhedron}",
    eprint = "1408.3410",
    archivePrefix = "arXiv",
    primaryClass = "hep-th",
    doi = "10.1007/JHEP03(2015)128",
    journal = "JHEP",
    volume = "03",
    pages = "128",
    year = "2015"
}

@article{Damgaard:2019ztj,
    author = "Damgaard, David and Ferro, Livia and Lukowski, Tomasz and Parisi, Matteo",
    title = "{The Momentum Amplituhedron}",
    eprint = "1905.04216",
    archivePrefix = "arXiv",
    primaryClass = "hep-th",
    reportNumber = "LMU-ASC 21/19",
    doi = "10.1007/JHEP08(2019)042",
    journal = "JHEP",
    volume = "08",
    pages = "042",
    year = "2019"
}

@article{Bourjaily:2017wjl,
	Archiveprefix = {arXiv},
	Author = {Bourjaily, Jacob L. and Herrmann, Enrico and Trnka, Jaroslav},
	Doi = {10.1007/JHEP06(2017)059},
	Eprint = {1704.05460},
	Journal = {JHEP},
	Pages = {059},
	Primaryclass = {hep-th},
	Reportnumber = {CALT-TH-2017-19},
	Title = {{Prescriptive Unitarity}},
	Volume = {06},
	Year = {2017}}

@article{Ferro:2015grk,
    author = "Ferro, Livia and Lukowski, Tomasz and Orta, Andrea and Parisi, Matteo",
    title = "{Towards the Amplituhedron Volume}",
    eprint = "1512.04954",
    archivePrefix = "arXiv",
    primaryClass = "hep-th",
    doi = "10.1007/JHEP03(2016)014",
    journal = "JHEP",
    volume = "03",
    pages = "014",
    year = "2016"
}

@article{Britto:2005fq,
    author = "Britto, Ruth and Cachazo, Freddy and Feng, Bo and Witten, Edward",
    title = "{Direct proof of tree-level recursion relation in Yang-Mills theory}",
    eprint = "hep-th/0501052",
    archivePrefix = "arXiv",
    doi = "10.1103/PhysRevLett.94.181602",
    journal = "Phys. Rev. Lett.",
    volume = "94",
    pages = "181602",
    year = "2005"
}

@article{Bern:1994zx,
    author = "Bern, Zvi and Dixon, Lance J. and Dunbar, David C. and Kosower, David A.",
    title = "{One loop n point gauge theory amplitudes, unitarity and collinear limits}",
    eprint = "hep-ph/9403226",
    archivePrefix = "arXiv",
    reportNumber = "SLAC-PUB-6415, SACLAY-SPH-T-94-20, UCLA-TEP-94-4, SWAT-94-17",
    doi = "10.1016/0550-3213(94)90179-1",
    journal = "Nucl. Phys. B",
    volume = "425",
    pages = "217--260",
    year = "1994"
}

@article{Bern:1994cg,
    author = "Bern, Zvi and Dixon, Lance J. and Dunbar, David C. and Kosower, David A.",
    title = "{Fusing gauge theory tree amplitudes into loop amplitudes}",
    eprint = "hep-ph/9409265",
    archivePrefix = "arXiv",
    reportNumber = "SLAC-PUB-6563, SACLAY-SPH-T-94-95, UCLA-TEP-94-29, SWAT-94-36",
    doi = "10.1016/0550-3213(94)00488-Z",
    journal = "Nucl. Phys. B",
    volume = "435",
    pages = "59--101",
    year = "1995"
}

@article{Bourjaily:2015jna,
    author = "Bourjaily, Jacob L. and Trnka, Jaroslav",
    title = "{Local Integrand Representations of All Two-Loop Amplitudes in Planar SYM}",
    eprint = "1505.05886",
    archivePrefix = "arXiv",
    primaryClass = "hep-th",
    reportNumber = "CALT-TH-2015-026",
    doi = "10.1007/JHEP08(2015)119",
    journal = "JHEP",
    volume = "08",
    pages = "119",
    year = "2015"
}

\end{document}